\documentclass[
  a4paper,UKenglish,cleveref,autoref,thm-restate
]{lipics-v2021}

\pdfoutput=1 
\hideLIPIcs  

\nolinenumbers

\usepackage{hyperref}  

\usepackage{bm}
\usepackage{cleveref}
\usepackage{amsmath,amssymb,amsfonts}
\usepackage{amsthm}

\usepackage{graphicx}
\usepackage{textcomp}
\usepackage[dvipsnames]{xcolor} 

\usepackage{multirow}     
\usepackage{booktabs}     
\usepackage{array}        

\usepackage{subcaption}
\usepackage{cleveref}

\usepackage{xcolor}
\usepackage{xpatch}  

\usepackage{orcidlink}
\usepackage{microtype}
\usepackage[utf8]{inputenc}

\usepackage{graphicx}
\usepackage{tikz}  
\usepackage{mathtools}
\usepackage{amsmath}
\usepackage{diagbox}
\usepackage{xstring}
\usepackage{algorithm2e}

\def\BibTeX{{\rm B\kern-.05em{\sc i\kern-.025em b}\kern-.08em
    T\kern-.1667em\lower.7ex\hbox{E}\kern-.125emX}}
\usepackage{mathtools}

\newcommand{\bigmlor}{%
  \mathop{%
    \bigvee\nolimits_{\mathrlap{\!m}}\kern-\scriptspace
  }\displaylimits
}
\newcommand{\mlor}{\mathbin{\lor_{\!m}}}

\usepackage{newfloat}
\usepackage{listings}

\title{Short Resolution Refutations for CNFs with Bounded Weighted Incidence Treewidth} 

\author{Shaowei Cai\footnote{Corresponding author.}}
  {Key Laboratory of System Software (Chinese Academy of Sciences), Beijing, China
   \and Institute of Software, Chinese Academy of Sciences, Beijing, China
   \and University of Chinese Academy of Sciences, Beijing, China}
  {caisw@ios.ac.cn}{}{}

\author{Ziqun Li\footnote{First author.}}
  {Key Laboratory of System Software (Chinese Academy of Sciences), Beijing, China
   \and Institute of Software, Chinese Academy of Sciences, Beijing, China
   \and University of Chinese Academy of Sciences, Beijing, China}
  {liziqun@ios.ac.cn}{}{}

\authorrunning{S. Cai and Z. Li} 

\Copyright{Shaowei Cai and Ziqun Li} 



\category{} 

\relatedversion{} 

\acknowledgements{}

\begin{document}

\hypersetup{
  pdfsubject={}
}

\maketitle

\begin{abstract}
It is an open problem in proof complexity whether every
unsatisfiable CNF formula has an FPT-sized resolution refutation
parameterized by incidence treewidth.
In this paper, we establish several upper bounds on resolution refutation length related to this problem.

Consider an unsatisfiable CNF formula $F$ with $n$ variables,
$m$ clauses, maximum clause width $k$, and incidence treewidth
$\mathrm{tw}^*(F)$.
In this paper, we introduce two variants of incidence treewidth.
Their definitions can be stated informally as follows.
The first is log-weighted incidence treewidth
$\mathrm{tw}_{\log}^*(F)$, which is the treewidth of the weighted incidence graph, in which variables have weight one
and each clause has weight equal to the logarithm of its width.
The second is partially log-weighted incidence treewidth
$\mathrm{tw}^*_{\mathrm{plog}}(F)$, which is a refinement of
log-weighted incidence treewidth.
In this variant, for a nice tree decomposition of the incidence graph, each clause has weight one along a path
selected for that clause and elsewhere has weight equal
to the logarithm of one plus the number of its literals
whose variables do not appear in any bag on that path,
and variables have weight one.

For every unsatisfiable CNF formula $F$, we prove the existence
of (i) an FPT-sized resolution refutation parameterized by
log-weighted incidence treewidth, with width at most
$\mathrm{tw}_{\log}^*(F)+k$;
(ii) a resolution refutation of length
$(n+m)k^{O(\mathrm{tw}^*(F))}$ and width at most
$\mathrm{tw}^*(F)+k$;
(iii) an FPT-sized resolution refutation parameterized by
partially log-weighted incidence treewidth; and
(iv) an FPT-sized regular resolution refutation parameterized
by log-weighted incidence treewidth.

Our main idea is to construct FPT-sized $k$-DNF resolution
refutations parameterized by incidence treewidth, and then
convert them into resolution refutations.
\end{abstract}

\paragraph*{Declaration of AI use}

Most results in this paper were discovered independently by
the authors, with the following exceptions:

The construction in Section~\ref{sec:regular-resolution}
is the result of collaboration between the authors and ChatGPT. For further details on the use of ChatGPT in this construction,
see Remark~\ref{rem:ai-regular-construction}.
Some auxiliary results and arguments, including the footnote in Subsection~\ref{subsec:comparison},
Theorem~\ref{thm:plog-one-sided}, and the argument concerning computability in Theorem~\ref{thm:subpolynomial-clause-width},
are credited to ChatGPT.

We also used ChatGPT to search the literature, check the
correctness of proofs, and assist with writing in English.
The authors have read and revised all AI-generated text to ensure
its correctness and readability. The authors take full responsibility
for the content of this paper.

\setcounter{tocdepth}{2}
\tableofcontents
\clearpage

\section{Introduction}

The resolution proof system is fundamental to propositional proof
complexity and is closely related to conflict-driven clause-learning
(CDCL) algorithms for SAT~\cite{BeameKautzSabharwal2004,
PipatsrisawatDarwiche2011,AtseriasFichteThurley2011}. The study of the length and width of resolution
refutations has been a central topic in proof complexity.
Haken~\cite{Haken1985} established the first exponential lower bounds for general
resolution, using formulas encoding the pigeonhole
principle.
Ben-Sasson and Wigderson~\cite{BenSassonWigderson2001} proved the well-known relations between
refutation length and width, providing a general method for deriving
length lower bounds from width lower
bounds.
Alongside these lower bounds, an important direction is to identify under which structural conditions unsatisfiable CNF formulas have
short refutations. The relationship between the treewidth of graph representations of
CNF formulas and their resolution refutation length has received
attention in previous work.

Two standard graph representations of CNF formulas are the primal
graph and the incidence graph.
Samer and Szeider~\cite{SamerSzeider2010} developed dynamic-programming
algorithms on tree decompositions of these graphs, showing that \#SAT, and hence SAT, is fixed-parameter tractable
when parameterized by primal treewidth and incidence treewidth.

The algorithmic tractability of CNF formulas with small incidence
treewidth leads to a corresponding question about their
resolution complexity.
In the report of Dagstuhl Seminar~19041, Szeider discussed
the question of whether unsatisfiable CNF formulas have
FPT-sized resolution refutations parameterized by incidence
treewidth~\cite[Section~4.14]{FominEtAl2019}.
This asks how the structure captured by an incidence tree decomposition
can be used to construct short resolution refutations.

In this paper, we make partial progress on this question
by establishing new upper bounds on resolution refutation length.

\subsection{Related Work} \label{sec:related-work}

It is known that unsatisfiable CNF formulas have FPT-sized
resolution refutations parameterized by primal
treewidth~\cite{RishDechter2000,AlekhnovichRazborov2011}. And it is also well known that the incidence treewidth of a CNF is not greater than its primal treewidth plus one. However, primal treewidth can be arbitrarily large even when
incidence treewidth is bounded. Therefore, to obtain FPT-sized
resolution refutations parameterized by incidence treewidth,
we need to improve these results.

Imanishi~\cite{Imanishi2017} proved that unsatisfiable CNF formulas
have FPT-sized regular resolution refutations parameterized by incidence
pathwidth.

Kolaitis and Vardi~\cite{KolaitisVardi2000} showed that a CNF formula
of incidence treewidth $w$ and maximum clause width $k$ has primal
treewidth at most $k(w+1)-1$.
Consequently, unsatisfiable CNF formulas of bounded clause width
have FPT-sized resolution refutations parameterized by incidence
treewidth.

Samer and Szeider~\cite{SamerSzeider2010CSP} showed that a CNF formula
$F$ of incidence treewidth $w$ can be transformed into an
equisatisfiable CNF formula $F'$ with at most three literals per
clause and primal treewidth at most $3w+3$.
As pointed out in the report of Dagstuhl Seminar~19041%
~\cite[Section~4.14]{FominEtAl2019}, if $F'$ is unsatisfiable,
the known upper bounds for primal treewidth give an FPT-sized
resolution refutation of $F'$ parameterized by $w$, but $F'$
contains additional variables absent from $F$.

F{\"u}rer~\cite{Furer2012} showed that a CNF formula of incidence
treewidth $w$ can be transformed into an equisatisfiable CNF formula
of primal treewidth at most $3w$.
Using this construction, he obtained FPT-sized refutations
parameterized by $w$ in an extension of resolution.
Actually, this extension can be implemented by introducing
new variables.

In a preprint, Cal\`i and Razgon~\cite{CaliRazgon2022} introduced
one-sided incidence treewidth and claimed that unsatisfiable
CNF formulas have FPT-sized regular resolution refutations
parameterized by $w$ and $p$, provided that one can delete at most $p$ clauses to obtain
a formula of one-sided incidence treewidth at most $w$.

\subsection{Main Results}

Recall that the incidence graph of a CNF formula $F$ is the bipartite graph
whose vertices are the variables and clauses of $F$, with an edge
between a variable and a clause if the variable appears in the clause.
The incidence treewidth $\mathrm{tw}^*(F)$ is the treewidth of
the incidence graph of $F$.

We introduce two variants of incidence treewidth.
The \emph{log-weighted incidence treewidth}
$\mathrm{tw}_{\log}^*(F)$ is the treewidth of the weighted incidence graph
in which each variable vertex has weight one and each clause
vertex $C$ has weight $\max\{\log|C|,1\}$, where $|C|$ denotes
the number of literals in $C$.
More specifically, the width of a tree decomposition of a weighted graph is the maximum total weight
of a bag minus one, and the treewidth of a weighted graph is
the minimum width over all its tree decompositions. By Lemma~\ref{lem:width-comparison}, if $F$ has maximum
clause width $k\geq2$, then
\[
\mathrm{tw}^*(F)+1
\leq \mathrm{tw}_{\log}^*(F)+1
\leq (\log k)(\mathrm{tw}^*(F)+1).
\]

The \emph{partially log-weighted incidence treewidth}
$\mathrm{tw}_{\mathrm{plog}}^*(F)$ refines log-weighted incidence
treewidth and can be defined informally as follows.
Given a nice tree decomposition of the incidence graph, for each
clause $C$, choose a root-to-leaf path in the subtree formed by
the bags containing $C$, or choose the empty path.
The clause $C$ has weight one in bags on the selected path.
In all other bags containing $C$, its weight is the logarithm
of one plus the number of literals of $C$ whose variables
do not appear in any bag on the selected path, with a minimum
weight of one.
The width determined by the decomposition and the selected paths
is the maximum total weight of a bag minus one.
The partially log-weighted incidence treewidth $\mathrm{tw}_{\mathrm{plog}}^*(F)$ is the minimum
of this width over all nice tree decompositions and all such
path selections. By Theorem~\ref{thm:plog-comparison}, partially log-weighted
incidence treewidth is no greater than log-weighted
incidence treewidth.
By Theorem~\ref{thm:plog-one-sided}, it is also no greater
than one-sided incidence treewidth defined
in~\cite{CaliRazgon2022}.

Our first main theorem gives upper bounds on refutation length
and width in terms of incidence treewidth and its weighted
variants.

\begin{restatable}{theorem}{MainBounds}
\label{thm:main-bounds}
Let $F$ be an unsatisfiable CNF formula of maximum clause width
$k\geq1$, with $n$ variables and $m$ clauses.
Then $F$ has the following refutations:
\begin{enumerate}
\renewcommand{\labelenumi}{(\arabic{enumi})}
\item a $k\mathrm{\text{-}DNF\text{-}Res}$ refutation of length
$(n+m)2^{O(\mathrm{tw}^*(F))}$;

\item a resolution refutation of length
$k(n+m)2^{O(\mathrm{tw}_{\log}^*(F))}$
and width at most $\mathrm{tw}_{\log}^*(F)+k$;

\item a resolution refutation of length
$(n+m)k^{O(\mathrm{tw}^*(F))}$
and width at most $\mathrm{tw}^*(F)+k$;

\item a resolution refutation of length
$k^2(n+m)2^{O(\mathrm{tw}_{\mathrm{plog}}^*(F))}$.
\end{enumerate}
\end{restatable}

Since $k \leq n$, the length bounds in (1), (2), and (4) are FPT bounds
parameterized by incidence treewidth, log-weighted incidence
treewidth, and partially log-weighted incidence treewidth,
respectively. (3) gives a new upper bound on resolution refutation length
with respect to incidence treewidth and maximum clause width.

Our second theorem gives an upper bound for regular resolution
with respect to log-weighted incidence treewidth.

\begin{restatable}{theorem}{RegularBounds}
\label{thm:regular-main-bounds}
Let $F$ be an unsatisfiable CNF formula with $n$ variables
and $m$ clauses.
Then $F$ has a regular resolution refutation of length
$(n+m)2^{O(\mathrm{tw}_{\log}^*(F))}$.
In particular, if $F$ has maximum clause width $k\geq2$,
then $F$ has a regular resolution refutation of length
$(n+m)k^{O(\mathrm{tw}^*(F))}$.
\end{restatable}

The first length bound in the above theorem is an FPT bound
parameterized by log-weighted incidence treewidth.

Our results not only improve previous upper bounds but also
help rule out certain formula families as candidates for
proving resolution length lower bounds with respect to
incidence treewidth, as discussed in
Subsection~\ref{subsec:lower-bounds}.

\subsection{Comparison with previous results}\label{subsec:comparison}

Let $F$ be an unsatisfiable CNF formula with $n$ variables,
$m$ clauses, maximum clause width $k\geq2$, and incidence
treewidth $\mathrm{tw}^*(F)$.
It is known that $F$ has a resolution refutation of length
$n^{O(1)}2^{O(\mathrm{tw}(F))}$%
~\cite{RishDechter2000,AlekhnovichRazborov2011}, where $\mathrm{tw}(F)$ is the primal treewidth of $F$.
Together with $\mathrm{tw}(F)\leq k(\mathrm{tw}^*(F)+1)-1$%
~\cite{KolaitisVardi2000}, this gives an upper bound of
$n^{O(1)}2^{O(k \mathrm{tw}^*(F))}$.
Our bound improves this upper bound.
By Theorem~\ref{thm:main-bounds}(3), $F$ has a resolution
refutation of length $(n+m)k^{O(\mathrm{tw}^*(F))}$.
We can prove that $m\leq n2^{O(\mathrm{tw}^*(F))}$,\footnote{%
Define a hypergraph $H$ whose vertices are the $2n$ literals
over the variables of $F$ and whose hyperedges are the clauses
of $F$, regarded as sets of literals.
Thus $H$ has $2n$ vertices and $m$ distinct hyperedges.
Let $I(H)$ be the incidence graph of the hypergraph $H$.
Replacing each variable vertex in every bag of a width-$\mathrm{tw}^*(F)$
incidence tree decomposition of $F$ by its two literal
vertices gives a tree decomposition of $I(H)$ of width
at most $2\mathrm{tw}^*(F)+1$.
Every minor $J$ of $I(H)$ has treewidth at most $2\mathrm{tw}^*(F)+1$
and hence satisfies
$|E(J)|\leq (2\mathrm{tw}^*(F)+1)|V(J)|$%
~\cite{FominVillanger2012}.
Thus, with $\nabla_1$ as defined in
Fomin, Oum, and Thilikos~\cite[Section~2]{FominOumThilikos2010},
we have $\nabla_1(I(H))\leq 2\mathrm{tw}^*(F)+1$.
Proposition~17 of ~\cite{FominOumThilikos2010} therefore gives
$m\leq 2n\,4^{2\mathrm{tw}^*(F)+1}=n2^{O(\mathrm{tw}^*(F))}$.}
so our upper bound can be written as
$nk^{O(\mathrm{tw}^*(F))}=n2^{O(\mathrm{tw}^*(F)\log k)}$. We improve the previous upper bound
$n^{O(1)}2^{O(\mathrm{tw}^*(F)k)}$ to $n2^{O(\mathrm{tw}^*(F)\log k)}$.

Allowing the introduction of extension variables makes it
possible to obtain FPT-sized resolution refutations
parameterized by incidence treewidth.
Previous constructions first transform the original formula
into an equisatisfiable CNF formula of small primal treewidth
and then apply the known results for primal
treewidth~\cite{SamerSzeider2010CSP,Furer2012,FominEtAl2019}.
In contrast, in Section~\ref{sec:extension-variables},
we directly use the structure of an incidence tree
decomposition of the original formula to construct FPT-sized
resolution refutations with extension variables.

Theorem~\ref{thm:plog-one-sided} shows that the partially
log-weighted incidence treewidth of a CNF $F$ is no greater than
its one-sided incidence treewidth.
Since every path decomposition of the incidence graph is also
a one-sided tree decomposition,
the one-sided incidence treewidth of $F$ is no greater
than its incidence pathwidth.
Furthermore, Remark~\ref{rem:one-sided-regular} shows that
we can use one-sided incidence tree decompositions and the construction in Theorem~\ref{thm:plog-one-sided} to
construct regular resolution refutations.
Thus, our construction gives regular resolution refutations of FPT length
parameterized by either one-sided incidence treewidth or
incidence pathwidth.
This recovers the FPT-length result of
Imanishi~\cite{Imanishi2017} for incidence pathwidth
and part of the results claimed by Cal\`i and Razgon
in their preprint~\cite{CaliRazgon2022}.

\section{Main Ideas and Proof Overview}

We try to explain the concepts used in this section
within the section itself.
For more complete definitions,
we refer the reader to Sections~\ref{sec:preliminaries}
and~\ref{sec:preparations}.

\subsection{Motivation for Our Construction}

Our construction was motivated by the following two facts.

First, introducing extension variables makes it possible
to construct FPT-sized resolution refutations parameterized
by incidence treewidth.
As discussed in Section~\ref{sec:related-work},
previous work introduces extension variables
to transform the original formula into an equisatisfiable
CNF formula of small primal treewidth and then apply
the known results for primal treewidth.
We also observed that, if we introduce variables to represent
whether a subclause of a clause in $F$ is satisfied, then we
might be able to simulate the algorithm of
Samer and Szeider~\cite{SamerSzeider2010}
on an incidence tree decomposition of $F$
directly within the resolution proof system
and obtain refutations of FPT length.

Second, Atserias and Bonet~\cite{AtseriasBonet2004}
established a correspondence between the proof system
$k\mathrm{\text{-}DNF\text{-}Res}$, also known as
$\mathrm{Res}(k)$, and the resolution proof system
with extension variables representing conjunctions
of at most $k$ literals over the original variables.
More precisely, a refutation of length $S$ in either system
can be converted into a refutation of length $O(kS)$
in the other.

These two facts suggested that we could construct a
$k\mathrm{\text{-}DNF\text{-}Res}$ refutation of FPT length parameterized
by incidence treewidth, where $k$ is the maximum clause
width of the original formula.
We then sought to convert this refutation into a resolution
refutation. Our basic idea was to associate a set of clauses with each
DNF appearing in the $k\mathrm{\text{-}DNF\text{-}Res}$ refutation
and derive these clauses within the resolution proof system,
following the order of the DNFs in the original refutation.

Our initial method was to expand each DNF directly.
For a DNF $D=T_1\lor\cdots\lor T_s$, define
\[
\operatorname{Exp}(D):=
\left\{
\ell_1\lor\cdots\lor\ell_s:
\ell_i\in T_i \text{ for every } i\in\{1,\ldots,s\}
\right\}.
\]
This expansion yields the upper bound with respect to
log-weighted incidence treewidth in
Theorem~\ref{thm:main-bounds}(2) and (3).

We then found a way to refine both the construction of the
$k\mathrm{\text{-}DNF\text{-}Res}$ refutation and the definition of expansion. These refinements motivated our definition of partially
log-weighted incidence treewidth and yielded the improved
bound in Theorem~\ref{thm:main-bounds}(4).

\subsection{Multiset proof systems} \label{subsec:introproofsys}

In order to describe our constructions more clearly, in Section~\ref{sec:multiset-proof-systems}, we introduce two new proof systems,
$\mathrm{mRes}$ and $k\mathrm{\text{-}mDNF\text{-}Res}$,
which can be viewed as multiset versions of the resolution proof system and
$k\mathrm{\text{-}DNF\text{-}Res}$, respectively.
A \emph{multiclause} is a finite multiset of literals,
and an \emph{mDNF} is a finite multiset of terms.
In $A\mlor B$, the multiplicity of each literal or term
is the sum of its multiplicities in $A$ and in $B$.
Thus, $\mlor$ has the same Boolean interpretation as $\lor$,
but does not combine repeated literals or terms into one; for example, we view $x\mlor x$ and $x$ as distinct multiclauses.

The inference rules of our new proof systems $\mathrm{mRes}$ and $k\mathrm{\text{-}mDNF\text{-}Res}$ are the multiset versions of the original
rules, together with the following contraction rules for
$\mathrm{mRes}$ and $k\mathrm{\text{-}mDNF\text{-}Res}$,
respectively:
\[
\frac{C\mlor\ell\mlor\ell}{C\mlor\ell}
\qquad\text{and}\qquad
\frac{D\mlor T\mlor T}{D\mlor T},
\]
where $C$ is a multiclause, $\ell$ is a literal,
$D$ is an mDNF, and $T$ is a term.
Contraction serves as an auxiliary rule for adjusting
multiplicities.
Lemmas~\ref{lem:mres-to-res}
and~\ref{lem-mdnfres-to-dnfres} show that refutations
in these multiset systems can be converted, step by step,
into refutations in the resolution proof system and
$k\mathrm{\text{-}DNF\text{-}Res}$, respectively.

We introduce these systems for two main reasons.
First, they allow us to describe our constructions more clearly.
For example, applying the resolution rule to $x_1\lor x_2$
and $\neg x_2\lor x_1$ gives $x_1$
in the resolution proof system.
In our constructions, we sometimes need $x_1\mlor x_1$
and sometimes need $x_1$.
In $\mathrm{mRes}$, the same rule applied to
$x_1\mlor x_2$ and $\neg x_2\mlor x_1$ gives
$x_1\mlor x_1$, which can then be contracted to $x_1$.
The multiset proof system allows us to choose the form we need.
Das~\cite{Das2015} and Bonacina and Bonet~\cite{BonacinaBonet2022}
also used multisets to define proof systems for the same reason,
although we were unaware of these works when introducing
our multiset systems. Most of our constructions are carried out in these two
multiset proof systems.
By Lemmas~\ref{lem:mres-to-res} and~\ref{lem-mdnfres-to-dnfres},
the upper bounds on refutation length obtained in these systems
also hold for the resolution and $k\mathrm{\text{-}DNF\text{-}Res}$ proof systems,
respectively, up to constant factors.

Second, the $\mathrm{mRes}$ proof system plays an important role in our
regularity argument.
Consider the following resolution step in $\mathrm{mRes}$:
\[
\frac{x_1\mlor x_1\qquad \neg x_1\mlor x_2}
     {x_1\mlor x_2}.
\]
In the resolution proof system, this step can be replaced
by a weakening step from $x_1$ to $x_1\lor x_2$.
Thus, a resolution step in $\mathrm{mRes}$ can be replaced
by a weakening step when converting the proof to the
resolution proof system.
Using this observation, in Lemma~\ref{lem:td-to-mres}, we  construct an $\mathrm{mRes}$ refutation which may not
be regular but can be converted into a regular resolution
refutation.

\subsection{Inconsistent States}

Our construction of refutations simulates the dynamic programming of a SAT
algorithm on a nice tree decomposition of the incidence graph.
The basic idea is to define inconsistent states at each node
of the tree decomposition to express the nonexistence of assignments
satisfying certain conditions, represent these states by mDNFs,
and derive these mDNFs in $k\mathrm{\text{-}mDNF\text{-}Res}$
from the leaves to the root of the tree decomposition.

In this subsection, we introduce the definition of inconsistent states.

Let $(T,\chi,r)$ be a nice tree decomposition of the incidence
graph of a CNF formula $F$, where $T$ is a tree rooted at $r$
and $\chi(t)$ is the bag associated with each node $t$ of $T$.
For each node $t$, let $\chi_v(t)$ and $\chi_c(t)$ denote
the variables and clauses in $\chi(t)$, respectively,
and let $T_t$ be the subtree rooted at $t$.
Let $\mathrm{Var}_t$ and $\mathrm{Cls}_t$ be the sets
of variables and clauses appearing in bags of $T_t$,
respectively.

A state is a triple $(t,\alpha,A)$, where
$\alpha:\chi_v(t)\to\{0,1\}$ is an assignment and
$A\subseteq\chi_c(t)$.
The state $(t,\alpha,A)$ is \emph{inconsistent} if
there is no assignment $\tau:\mathrm{Var}_t\to\{0,1\}$
extending $\alpha$ that satisfies every clause in
$A\cup(\mathrm{Cls}_t\setminus\chi_c(t))$.

Lemmas~\ref{lem:join-zero-recursion},
\ref{lem:introduce-var-zero}, \ref{lem:introduce-clause-zero},
\ref{lem:forget-var-zero}, and~\ref{lem:forget-clause-zero}
show how to derive inconsistent states at a node from
inconsistent states at its children for each node type (see Subsection~\ref{subsec:tree-decompositions}
for the node types of nice tree decompositions of incidence graphs).

\subsection{Log-weighted incidence treewidth and the construction}
\label{subsec:log-construction-overview}

Our initial approach was to represent inconsistent states
by the mDNFs defined below.
This representation leads to the bounds in
Theorem~\ref{thm:main-bounds}(1), (2), and~(3).

Given a nice tree decomposition $(T,\chi,r)$ of the incidence
graph of a CNF formula $F$, for each node $t$ and clause $c\in\chi_c(t)$, define
\[
c_t^{\mathrm{int}}
:=
\bigvee_{\substack{\ell\in c\\
\mathrm{var}(\ell)\in\mathrm{Var}_t}}\ell.
\]
Thus, $c_t^{\mathrm{int}}$ consists of the literals of $c$
whose variables appear in bags of the subtree rooted at $t$.

For an assignment $\alpha:\chi_v(t)\to\{0,1\}$, let
\[
C_\alpha
:=
\bigvee_{\substack{x\in\chi_v(t)\\\alpha(x)=0}}x
\;\lor\!
\bigvee_{\substack{x\in\chi_v(t)\\\alpha(x)=1}}\neg x.
\]
The clause $C_\alpha$ contains the literal falsified by
$\alpha$ for each variable in $\chi_v(t)$.

We represent an inconsistent state $(t,\alpha,A)$
by the mDNF
\[
C_\alpha\mlor\bigmlor_{c\in A}\neg c_t^{\mathrm{int}},
\]
where $C_\alpha$ is viewed as an mDNF consisting of singleton terms
and each $\neg c_t^{\mathrm{int}}$ is viewed as the term
$\bigwedge_{\ell\in c_t^{\mathrm{int}}}\neg\ell$. We denote this mDNF by $D(t,\alpha,A,\emptyset)$,
which is a special case of the more general definition
of $D(t,\alpha,A,B)$ given in the next subsection.
This mDNF represents the inconsistency of $(t,\alpha,A)$:
every assignment to $\mathrm{Var}_t$ satisfying all clauses
in $\mathrm{Cls}_t\setminus\chi_c(t)$ either disagrees with
$\alpha$ on some variable in $\chi_v(t)$ or fails to satisfy
at least one clause in $A$.

With $B=\emptyset$, Lemmas~\ref{lem:derive-join},
\ref{lem:derive-introduce-var}, \ref{lem:derive-introduce-clause},
\ref{lem:derive-forget-var}, and~\ref{lem:derive-forget-clause}
show how to derive the mDNF representing an inconsistent state
at a node from the mDNFs representing inconsistent states
at its children.
Theorem~\ref{thm:td-to-dnfres-tight} combines these derivations
from the leaves to the root of the nice tree decomposition
$(T,\chi,r)$ to construct a
$k\mathrm{\text{-}mDNF\text{-}Res}$ refutation of $F$
of length $|V(T)|2^{O(w)}$, where $w$ is the width
of $(T,\chi,r)$, and $|V(T)|$ is the number of nodes in $T$. Using this construction, we can prove Theorem~\ref{thm:main-bounds}(1).

We next consider how to convert the
$k\mathrm{\text{-}mDNF\text{-}Res}$ refutation into
a resolution refutation.
For this purpose, we define the \emph{expansion} of an mDNF
$D=T_1\mlor\cdots\mlor T_s$ by
\[
\operatorname{Exp}(D)
:=
\left\{
\ell_1\mlor\cdots\mlor\ell_s:
\ell_i\in T_i \text{ for every } i\in\{1,\ldots,s\}
\right\}.
\]
The idea is to derive the multiclauses in the expansion of
each mDNF within $\mathrm{mRes}$, following the order in which
the mDNFs appear in the original refutation.

For the mDNF representing an inconsistent state $(t,\alpha,A)$, we have
\[
\left|\operatorname{Exp}\bigl(D(t,\alpha,A,\emptyset)\bigr)\right|
\leq \prod_{c\in A}|c|.
\]
This bound motivated our definition of log-weighted incidence treewidth.
The \emph{log-weighted incidence treewidth}
$\mathrm{tw}_{\log}^*(F)$ is the treewidth of the weighted
incidence graph in which each variable vertex has weight one
and each clause vertex $c$ has weight $\max\{\log|c|,1\}$,
where $|c|$ denotes the number of literals in $c$.
Let $(T,\chi,r)$ be a rooted nice tree decomposition
of this weighted incidence graph of width $\mathrm{tw}_{\log}^*(F)$.
For an inconsistent state $(t,\alpha,A)$, we have
\[
\left|\operatorname{Exp}\bigl(D(t,\alpha,A,\emptyset)\bigr)\right|
\leq \prod_{c\in A}|c|
\leq 2^{\mathrm{tw}_{\log}^*(F)+1}.
\]

Each multiclause in
$\operatorname{Exp}(D(t,\alpha,A,\emptyset))$
has width $|\chi_v(t)|+|A|\leq|\chi(t)|$.
Thus, when $(T,\chi,r)$ has incidence width
$\mathrm{tw}^*(F)$, these multiclauses have width
at most $\mathrm{tw}^*(F)+1$.
The bounds on the size of the expansions and the widths
of multiclauses in expansions form the basis for the length and
width bounds of the resolution refutation.
Theorem~\ref{thm:res-bound} shows how to convert the
$k\mathrm{\text{-}mDNF\text{-}Res}$ refutation into
a resolution refutation using our definition of expansion, yielding the bounds in
Theorem~\ref{thm:main-bounds}(2) and~(3).

\subsection{Partially log-weighted incidence treewidth and the refined construction}
\label{subsec:plog-construction-overview}

A key limitation of the previous construction is the size
of the expansions.
Since we use $\neg c_t^{\mathrm{int}}$ to represent each
clause $c\in A$, the bound on the expansion size involves
the product of the widths of these clauses.
We therefore consider replacing $\neg c_t^{\mathrm{int}}$
with the external part of $c$, defined by
\[
c_t^{\mathrm{ext}}
:=
\bigvee_{\substack{\ell\in c\\
\mathrm{var}(\ell)\notin\mathrm{Var}_t}}\ell.
\]
$c_t^{\mathrm{ext}}$ only contributes a factor of one
to the expansion size.

However, replacing every $\neg c_t^{\mathrm{int}}$
with $c_t^{\mathrm{ext}}$ does not in general allow us
to derive the mDNF representing an inconsistent state
at a join node from the mDNFs representing inconsistent
states at its children.
Nevertheless, at a join node $t$ with children $t_1$ and $t_2$,
using $c_{t_1}^{\mathrm{ext}}$ and $\neg c_{t_2}^{\mathrm{int}}$
in the mDNFs representing inconsistent states at the children
allows us to derive the mDNF representing an inconsistent
state at $t$ using $c_t^{\mathrm{ext}}$.
This observation suggests that we may select a path for each clause $c$
and use $c_t^{\mathrm{ext}}$ at nodes on the path
and $\neg c_t^{\mathrm{int}}$ at all other nodes
whose bags contain $c$.
More precisely, let $T(c)$ be the subtree consisting of
the nodes whose bags contain $c$, rooted at the node
closest to $r$.
Then we may choose a path $P_c$ from
the root of $T(c)$ to one of its leaves.

For an inconsistent state $(t,\alpha,A)$, let
$B:=\{c\in A:t\in P_c\}$.
We represent $(t,\alpha,A)$ by the mDNF
\[
D(t,\alpha,A,B)
:=
C_\alpha
\mlor\bigmlor_{c\in B}c_t^{\mathrm{ext}}
\mlor\bigmlor_{c\in A\setminus B}\neg c_t^{\mathrm{int}}.
\]
When $B=\emptyset$, this is the same as the representation
introduced in the previous subsection.

Lemmas~\ref{lem:derive-join}, \ref{lem:derive-introduce-var},
\ref{lem:derive-introduce-clause}, \ref{lem:derive-forget-var},
and~\ref{lem:derive-forget-clause} show how to derive
the mDNF representing an inconsistent state at a node
from the mDNFs representing inconsistent states at its children.
Theorem~\ref{thm:td-to-dnfres-plog} combines these local
derivations to construct a
$k\mathrm{\text{-}mDNF\text{-}Res}$ refutation of $F$
of length $|V(T)|2^{O(w)}$, where $w$ is the width
of the nice tree decomposition $(T,\chi,r)$.

Moreover, we can reduce the size of the expansions of mDNFs by considering
the variables that appear in bags on the selected paths. We next show how to modify the expansion set to reduce its size.
Suppose that $t\in V(T(c))\setminus P_c$ and a variable $x$
appears both in $c_t^{\mathrm{int}}$ and in a bag on $P_c$.
By definition of tree decompositions, we can prove
$x\in\chi_v(t)$.
Fix an inconsistent state $(t,\alpha,A)$ and a clause
$c\in A$ such that $t\notin P_c$.
Let $x$ be a variable of $c_t^{\mathrm{int}}$ that appears
in a bag on $P_c$.
By the preceding observation, $x\in\chi_v(t)$.
Suppose that a literal on $x$ is selected from
$\neg c_t^{\mathrm{int}}$ when forming a multiclause
in $\operatorname{Exp}(D(t,\alpha,A,B))$.
Since $x\in\chi_v(t)$, the clause $C_\alpha$ also contains
a literal on $x$.
If these two literals are the same, we only need to keep the literal in $C_\alpha$.
If one literal is the negation of the other, the
multiclause can be derived by the weakening rule from the axiom $x\mlor\neg x$,
so we may omit this multiclause from the expansion.
After these modifications, each term $\neg c_t^{\mathrm{int}}$
contributes either a literal whose variable does not appear
in any bag on $P_c$, or no literal, to each multiclause
in the modified expansion. Let $\mathcal P=(P_c)_{c\in F}$, and let
$\rho_{\mathcal P}(c)$ be the number of literals of $c$
whose variables do not appear in any bag on $P_c$.
There are therefore at most $\rho_{\mathcal P}(c)+1$
possibilities for the contribution of
$\neg c_t^{\mathrm{int}}$ to the modified expansion.
A more careful analysis shows that we can reduce the size of expansions of all mDNFs
in the refutation constructed in Theorem~\ref{thm:td-to-dnfres-plog} similarly.

This motivates the definition of
\emph{partially log-weighted incidence treewidth}, which gives
an exponential upper bound on the size of the modified expansions.
To define partially log-weighted incidence treewidth,
we give each vertex of the incidence graph a weight
in each bag containing it in a nice incidence tree decomposition
$(T,\chi,r)$ and path family $\mathcal P=(P_c)_{c\in F}$.
Each variable vertex has weight one, while each clause vertex
$c$ has weight one in bags on $P_c$ and weight
$\max\{\log(\rho_{\mathcal P}(c)+1),1\}$
in all other bags containing $c$.
The partially log-weighted width of $(T,\chi,r)$ with respect to $\mathcal P$
is the maximum total weight of a bag minus one.
The partially log-weighted incidence treewidth
$\mathrm{tw}_{\mathrm{plog}}^*(F)$ is the minimum of this width
over all nice tree decompositions of the incidence graph
and all such path families. By Theorem~\ref{thm:plog-comparison}, partially log-weighted
incidence treewidth is no greater than log-weighted
incidence treewidth.

Theorem~\ref{thm:res-bound-plog} shows how to convert
this $k\mathrm{\text{-}mDNF\text{-}Res}$ refutation into
a resolution refutation using the modified expansions.
This gives the bound in Theorem~\ref{thm:main-bounds}(4).

Because we use $c_t^{\mathrm{ext}}$ in the new representation, the multiclauses in modified expansions may have large width.
So we do not obtain a width bound comparable to that of
the previous construction.

\subsection{Constructing regular resolution refutations}
\label{subsec:regular-construction-overview}

Our construction of regular resolution refutations is inspired
by the previous construction. However, these refutations are
not obtained by converting
$k\mathrm{\text{-}mDNF\text{-}Res}$ refutations into resolution refutations.
Instead, we represent each inconsistent state by a family
of multiclauses and derive these multiclauses directly
in $\mathrm{mRes}$.

Let $F$ be an unsatisfiable CNF. For each clause $c\in F$, fix a representation
$c=\ell_1^c\lor\cdots\lor\ell_{q_c}^c$.
For a subclause
$c'=\ell_{i_1}^c\lor\cdots\lor\ell_{i_s}^c$,
where $i_1<\cdots<i_s$, define
\[
\begin{aligned}
\mathcal C_1(c')
&:=
\left\{
\neg\ell_{i_h}^c
\mlor\bigmlor_{j=h+1}^{s}\ell_{i_j}^c:
1\leq h\leq s
\right\},\\
\mathcal C_0(c')
&:=
\left\{
\bigmlor_{j=h}^{s}\ell_{i_j}^c:
1\leq h\leq s+1
\right\}.
\end{aligned}
\]
For an inconsistent state $(t,\alpha,A)$, define
\[
\mathcal E(t,\alpha,A)
:=
\left\{
C_\alpha\mlor\bigmlor_{c\in\chi_c(t)}C_c:
\begin{array}{ll}
C_c\in\mathcal C_1(c_t^{\mathrm{int}}) & \text{if }c\in A,\\
C_c\in\mathcal C_0(c_t^{\mathrm{int}}) & \text{if }c\notin A
\end{array}
\right\}.
\]
These constructions may seem unusual.
We next explain how we obtained them.

As mentioned in Subsection~\ref{subsec:introproofsys}, a resolution step in $\mathrm{mRes}$ may become a weakening
step when we convert the $\mathrm{mRes}$ refutation into a resolution refutation.
More precisely, we call an application of the resolution rule
to $C\mlor x$ and $D\mlor\neg x$
\emph{variable-eliminating} if neither $x$ nor $\neg x$
appears in the resulting multiclause $C\mlor D$.
Lemma~\ref{lem:mres-to-res} shows that every resolution step
in the converted refutation comes from a variable-eliminating
resolution step in the $\mathrm{mRes}$ refutation.
It therefore suffices to ensure that each variable is used
in at most one variable-eliminating resolution on every
directed path in the proof DAG.

We observed that the refutation constructed in
Theorem~\ref{thm:res-bound} was already close to satisfying
this condition:

We attempted to prove inductively that, on every directed
path ending at a multiclause $M$, no variable in
$\operatorname{var}(M)$ had been used in a variable-eliminating
resolution.
We found that the inductive argument works for all node types except
forget-clause nodes (see Subsection~\ref{subsec:tree-decompositions}
for the node types of nice tree decompositions of incidence graphs).
This led us to replace
$\neg\ell_{i_h}^c$ with
$\neg\ell_{i_h}^c\mlor
\bigmlor\limits_{j=h+1}^{s}\ell_{i_j}^c$
in $\mathcal C_1(c')$.
After this modification, we found that the inductive argument
worked for all node types except join nodes, which led us
to include the multiclauses in $\mathcal C_0(c')$
for clauses outside $A$.

Lemma~\ref{lem:td-to-mres} presents our construction.
The proof of regularity in Lemma~\ref{lem:td-to-mres} relies mainly on two induction
properties for every $M\in\mathcal E(t,\alpha,A)$:
\begin{enumerate}
\item on every directed path in the proof DAG ending at $M$,
each variable is resolved in at most one variable-eliminating
resolution;
\item on every directed path in the proof DAG ending at $M$,
no variable in $\operatorname{var}(M)$ is resolved in a
variable-eliminating resolution.
\end{enumerate}
We also maintain the auxiliary property that every multiclause
in the derivation of $M$ contains only variables from
$\mathrm{Var}_t$.
Using these induction properties, we can show that the $\mathrm{mRes}$
refutation we construct can be converted into a regular
resolution refutation by Lemma~\ref{lem:mres-to-res}.

\begin{remark}[AI use in this construction] \label{rem:ai-regular-construction}
As we stated in this subsection, the authors observed that the refutation constructed in
Theorem~\ref{thm:res-bound}, after conversion to a resolution
refutation, was close to being regular.
When we used ChatGPT to check our proof, it found
an oversight in the argument for forget-clause nodes
and suggested a possible approach to correct it.
We did not pursue this approach because it was too complicated,
but it inspired our construction of $\mathcal C_1(c')$.
The rest of the construction was completed by the authors.
\end{remark}

\subsection{Further discussion}

In Subsection~\ref{subsec:clause-width}, we show that
unsatisfiable CNF formulas with $n$ variables and maximum
clause width bounded by a function $f(n)=n^{o(1)}$
have FPT-sized resolution refutations parameterized by
incidence treewidth, under a suitable computability
condition on $f$.

In Subsection~\ref{subsec:lower-bounds}, we discuss the
limitations of using the classic length--width relation of
Ben-Sasson and Wigderson%
~\cite{BenSassonWigderson2001}
to prove length lower bounds, and explain what our upper
bounds imply for candidate formula families.

In Subsection~\ref{sec:extension-variables}, we show how
our construction gives FPT-sized resolution refutations
parameterized by the incidence treewidth of $F$ when
we are allowed to introduce variables representing
every nonempty subclause of a clause in $F$.

\section{Preliminaries} \label{sec:preliminaries}

\subsection{Resolution-based proof systems}

A \emph{literal} is a Boolean variable $x$ or its negation $\neg x$. The literals $x$ and $\neg x$ are called literals on $x$.
A \emph{clause} is a disjunction of literals. For $k \ge 1$, a \emph{$k$-clause} is a disjunction of $k$ literals. 
A \emph{CNF} formula is a conjunction of clauses, and a \emph{$k$-CNF} is a conjunction of $k$-clauses. 
By a slight abuse of notation, we identify a clause with the set of its literals, and a CNF formula with the set of its clauses. We say that a clause $C'$ is a \emph{subclause} of a clause $C$ if $C'\subseteq C$. The \emph{width} of a clause is the number of literals it contains, denoted by $|C|$. The \emph{maximum clause width} of a CNF formula $F$
is the maximum width of its clauses, namely
$\max_{C\in F}|C|$.

For a literal $\ell$, we write $\mathrm{var}(\ell)$ for its underlying variable, i.e., $\mathrm{var}(x)=x$ and $\mathrm{var}(\neg x)=x$. 
For a clause $C$, let $\mathrm{var}(C):=\{\mathrm{var}(\ell): \ell \in C\}$.
For a CNF formula $F$, let $\mathrm{var}(F):=\bigcup_{C\in F}\mathrm{var}(C)$.

An \emph{assignment} for $F$ is a function $\alpha:X\to\{0,1\}$ for some $X\subseteq \mathrm{var}(F)$. In particular, throughout the paper, assignments are allowed to be partial, that is, $X$ may be a proper subset of $\mathrm{var}(F)$. An assignment $\alpha:X\to\{0,1\}$ satisfies a literal $\ell$ if $\mathrm{var}(\ell)\in X$ and either $\ell=x$ with $\alpha(x)=1$, or $\ell=\neg x$ with $\alpha(x)=0$. It satisfies a clause $C$ if it satisfies some literal in $C$. A CNF formula $F$ is satisfiable if there exists an assignment $\alpha:\mathrm{var}(F)\to\{0,1\}$ that satisfies every clause of $F$. 

A clause is \emph{tautological} if it contains both $x$
and $\neg x$ for some variable $x$.
Throughout the paper, all CNF formulas are assumed
to contain neither empty clauses nor tautological clauses.

A \emph{term} is a conjunction of literals. 
For $k \ge 1$, a \emph{$k$-term} is a conjunction of at most $k$ literals. 
A \emph{DNF} formula is a disjunction of terms, and a \emph{$k$-DNF} is a disjunction of $k$-terms. 
By a slight abuse of notation, we identify a term with the set of its literals, and a DNF formula with the set of its terms.

Next, we introduce the resolution proof system. 

A \emph{resolution derivation} from a CNF formula $F$ is a sequence of clauses $\pi=(C_1,\dots,C_s)$. Each $C_i$ either belongs to $F$, in which case it is called
an \emph{initial clause}, or is obtained from earlier clauses by one of the following rules:
\[
\frac{C}{C\lor E}
\quad(\text{weakening}),
\qquad
\frac{C\lor x\qquad D\lor\neg x}{C\lor D}
\quad(\text{resolution}).
\]
A \emph{resolution refutation} of an unsatisfiable CNF formula $F$ is a derivation of the empty clause $\bot$. The \emph{length} of a resolution refutation is the total number of clauses in the derivation, and the \emph{width} is the maximum width of any clause appearing in it. We denote the length of a resolution refutation $\pi$ by $|\pi|$. For an unsatisfiable CNF formula $F$, we define the \emph{resolution length} $L(F)$ and \emph{resolution width} $W(F)$ as the minimum length and width of all resolution refutations of $F$.

Given a resolution derivation $\pi=(C_1,\dots,C_s)$, its \emph{proof DAG (directed acyclic graph)} is the directed acyclic graph with vertices $v_1,\dots,v_s$, where $v_i$ is labelled by $C_i$, and with an edge from $v_i$ to $v_j$ whenever $C_i$ is used to derive $C_j$. The derivation is \emph{regular} if no variable is resolved in two different resolution steps on one directed path in its proof DAG.

We next define the proof system
$k\mathrm{\text{-}DNF\text{-}Res}$, also known as
$k$-DNF resolution or $\mathrm{Res}(k)$.

A $k\mathrm{\text{-}DNF\text{-}Res}$ derivation from a CNF formula $F$ is a sequence of DNF formulas in which each formula is either an initial clause from $F$ (viewed as a DNF consisting of singleton terms), or the axiom $x \lor \neg x$, or a $k$-DNF formula obtained from previous formulas by one of the following inference rules:
\[
\frac{D}{D\lor E}
\quad(\text{weakening}),
\]
\[
\frac{D\lor T}{D\lor T'}
\quad(\land\text{-elimination}),\qquad T'\subseteq T,
\]
\[
\frac{D_1\lor T_1\qquad D_2\lor T_2}
     {D_1\lor D_2\lor(T_1\land T_2)}
\quad(\land\text{-introduction}),
\]
\[
\frac{
D_1\lor\ell_1\lor\cdots\lor\ell_q
\qquad
D_2\lor(\neg\ell_1\land\cdots\land\neg\ell_q)
}{
D_1\lor D_2
}
\quad(\text{cut}),
\]
where $D,E,D_1,D_2$ are DNFs, $T,T',T_1,T_2$ are terms, and
$\{\ell_1,\ldots,\ell_q\}$ is a nonempty finite set of literals.
A \emph{$k\mathrm{\text{-}DNF\text{-}Res}$ refutation} of an unsatisfiable CNF formula $F$ is a derivation of the empty clause $\bot$.

We adopt the conventions that the conjunction over an empty
set is $\top$ (true) and the disjunction over an empty set
is $\bot$ (false).

For notational convenience, we use the following mild abuse of notation: for a DNF formula $D$ and a clause $C=\ell_1\lor\cdots\lor\ell_m$, we use $D\lor C$ and $D\lor\neg C$ to denote the DNF formulas $D\lor\ell_1\lor\cdots\lor\ell_m$ and $D\lor(\neg\ell_1\land\cdots\land\neg\ell_m)$, respectively.

\subsection{Graph representations and tree decompositions} \label{subsec:tree-decompositions}

Let $F$ be a CNF formula. The \emph{incidence graph} $G^*(F)$ is the bipartite graph
whose vertices are the variables and clauses of $F$,
where a variable $x$ is adjacent to a clause $C$ if and only
if $x\in\mathrm{var}(C)$. The \emph{primal graph} $G(F)$ is the graph with vertex set $\mathrm{var}(F)$, where two variables are adjacent if and only if they appear together in some clause of $F$.

\medskip

A \emph{tree decomposition} of a graph $G=(V,E)$ is a pair $(T,\chi)$ where $T$ is a tree and $\chi: V(T) \to 2^V$ assigns to each node $t \in V(T)$ a subset $\chi(t) \subseteq V$, such that:
\begin{enumerate}
    \item for every vertex $v\in V$, there exists $t\in V(T)$ such that $v\in \chi(t)$;
    \item for every edge $\{u,v\}\in E$, there exists $t\in V(T)$ such that $\{u,v\}\subseteq \chi(t)$;
    \item for every $v\in V$, the set $\{t\in V(T): v\in \chi(t)\}$ induces a connected subtree of $T$.
\end{enumerate}
The set $\chi(t)$ is called the bag at $t$.

The width of $(T,\chi)$, denoted by $w(T,\chi)$, is $\max_{t\in V(T)}|\chi(t)|-1$, and the treewidth $\mathrm{tw}(G)$ is the minimum width over all tree decompositions of $G$.

\medskip

The \emph{incidence treewidth} of $F$ is $\mathrm{tw}^*(F) := \mathrm{tw}(G^*(F))$, and the \emph{primal treewidth} is $\mathrm{tw}(F) := \mathrm{tw}(G(F))$. Throughout the paper, we use
$\chi_v(t)$ and $\chi_c(t)$ to denote
$\chi(t)\cap\mathrm{var}(F)$ and $\chi(t)\cap F$, respectively.

We will also use the notion of nice tree decompositions. A triple $(T,\chi,r)$ is a \emph{nice tree decomposition} if $(T,\chi)$ is a tree decomposition, the tree $T$ is rooted at $r$, and:
\begin{enumerate}
    \item $\chi(r)=\emptyset$, and $\chi(t)=\emptyset$ for every
    leaf $t$ of $T$;
    \item every node has at most two children;
    \item if a node $t$ has two children $t_1,t_2$, then $\chi(t)=\chi(t_1)=\chi(t_2)$, and $t$ is called a \emph{join node};
    \item if a node $t$ has exactly one child $t'$, then exactly one of the following holds:
    \begin{enumerate}
        \item $|\chi(t)|=|\chi(t')|+1$ and $\chi(t')\subset \chi(t)$; then $t$ is called an \emph{introduce node};
        \item $|\chi(t)|=|\chi(t')|-1$ and $\chi(t)\subset \chi(t')$; then $t$ is called a \emph{forget node}.
    \end{enumerate}
\end{enumerate}

In a nice tree decomposition $(T,\chi,r)$ of $G^*(F)$, an introduce node $t$ with child $t'$ is called an \emph{introduce-variable node} introducing $x$ if $\chi_v(t)=\chi_v(t')\cup\{x\}$ and $\chi_c(t)=\chi_c(t')$, and an \emph{introduce-clause node} introducing $c$ if $\chi_c(t)=\chi_c(t')\cup\{c\}$ and $\chi_v(t)=\chi_v(t')$. Forget-variable and forget-clause nodes are defined analogously. Thus, every non-leaf node of $(T,\chi,r)$ is one of five types:
a join node, an introduce-variable node, an introduce-clause
node, a forget-variable node, or a forget-clause node.

The following standard result guarantees the existence of a
small nice tree decomposition of minimum width.

\begin{lemma}[{\cite[Lemma~7.4]{CyganEtAl2015}}]\label{lem:nice-empty-incidence}
Let $F$ be a nonempty CNF formula with $n$ variables and $m$ clauses. Then $G^*(F)$ has a nice tree decomposition $(T,\chi)$ of width $\mathrm{tw}^*(F)$ and
\[
|V(T)|=O\bigl((n+m)\,\mathrm{tw}^*(F)\bigr).
\]
\end{lemma}

Let $\mathcal C$ be a class of unsatisfiable CNF formulas.
We say that $\mathcal C$ has resolution refutations of
FPT length parameterized by incidence treewidth if there
exist a computable function
$g:\mathbb N_{\geq0}\to\mathbb N_{\geq1}$ and a constant
$c\geq1$ such that every formula $F\in\mathcal C$ with
$n$ variables and $m$ clauses has a resolution refutation
of length at most $g(\mathrm{tw}^*(F))(n+m)^c$.\footnote{In our setting, it is easy to verify
that this definition is equivalent to the fpt-boundedness
condition in~\cite[Definition~2.6]{BeyersdorffGalesiLauriaRazborov2012}.} We define FPT-sized refutations for other proof systems analogously.

\section{Preparations for the constructions} \label{sec:preparations}

\subsection{Weighted incidence treewidth}

In addition to the standard notion of incidence treewidth, we will work with two weighted variants that play a central role in our results.

Given a CNF formula $F$, define the \emph{log-weighted width} of a tree decomposition $(T,\chi)$ of $G^*(F)$
\[
w_{\log}(T,\chi) :=
\max_{t\in V(T)}\left(|\chi_v(t)| + \sum_{C\in \chi_c(t)} \max\{ \log|C|,1\}\right)-1.
\]
The \emph{log-weighted incidence treewidth} $\mathrm{tw}^*_{\log}(F)$ is the minimum of $w_{\log}(T,\chi)$ over all tree decompositions $(T,\chi)$ of $G^*(F)$. This can be viewed as a weighted variant of treewidth in which variable vertices have unit weight and each clause vertex $C$ is assigned
weight $\max\{ \log|C|,1\}$.\footnote{All logarithms are base $2$ in this paper.}

The following lemma shows the relationship between incidence
treewidth and log-weighted incidence treewidth.

\begin{lemma}\label{lem:width-comparison}
Let $F$ be a CNF formula whose maximum clause width is $k\geq 2$, and let $(T,\chi)$ be a tree decomposition of $G^*(F)$. Let $w$ be the width of $(T,\chi)$ and let $w_{\log}$ be its log-weighted width. Then
\[
w+1 \le w_{\log}+1 \le (\log k)(w+1).
\]
Consequently,
\[
\mathrm{tw}^*(F)+1
\leq \mathrm{tw}_{\log}^*(F)+1
\leq (\log k)(\mathrm{tw}^*(F)+1).
\]
\end{lemma}

\begin{proof}
For each bag $t\in V(T)$, let $a_t:=|\chi_v(t)|$ and $b_t:=|\chi_c(t)|$. Then
\[
a_t+b_t
\le
a_t+\sum_{C\in\chi_c(t)}\max\{\log |C|,1\}
\le
a_t+(\log k)b_t
\le
(\log k)(a_t+b_t),
\]
where the second inequality uses $|C|\le k$ and $k\ge 2$. Taking the maximum over $t$ gives
$w+1\leq w_{\log}+1\leq(\log k)(w+1)$.
Taking the minimum over all tree decompositions of $G^*(F)$
then gives the inequalities between incidence
treewidth and log-weighted incidence treewidth.
\end{proof}

Let $(T,\chi,r)$ be a nice tree decomposition of $G^*(F)$. For each clause $c\in F$, let
$T(c)$ be the rooted subtree of $T$ consisting of the nodes $t$ such that
$c\in\chi_c(t)$, whose root is the unique node of $T(c)$
closest to $r$.

A \emph{clause-path family} of $(T,\chi,r)$ is a family
$\mathcal P=(P_c)_{c\in F}$ such that, for every $c\in F$, either
$P_c=\emptyset$, or $P_c$ is a path in $T(c)$ from the root of
$T(c)$ to one of the leaves of $T(c)$.

For each $c\in F$, define
\[
\rho_{\mathcal P}(c):=
\left|\left\{\ell\in c:
\operatorname{var}(\ell)\notin
\bigcup_{u\in P_c}\chi_v(u)\right\}\right|.
\]
That is, $\rho_{\mathcal P}(c)$ is the number of the literals of $c$ whose
variables are not in any bags along $P_c$.

For $t\in V(T)$ and $c\in\chi_c(t)$, define
\[
\omega_{\mathcal P}(t,c):=
\begin{cases}
1, & P_c\neq\emptyset\text{ and }t\in P_c,\\
\max\{\log(\rho_{\mathcal P}(c)+1),1\}, & \text{otherwise}.
\end{cases}
\]
The \emph{partially log-weighted width} of $(T,\chi,r)$ and $\mathcal P$ is
\[
w_{\mathrm{plog}}(T,\chi,r,\mathcal P):=
\max_{t\in V(T)}
\left(
|\chi_v(t)|
+\sum_{c\in\chi_c(t)}\omega_{\mathcal P}(t,c)
\right)-1.
\]
Define
\[
w_{\mathrm{plog}}(T,\chi,r):=
\min_{\mathcal P}
w_{\mathrm{plog}}(T,\chi,r,\mathcal P),
\]
where the minimum is taken over all clause-path families
$\mathcal P$. The \emph{partially log-weighted incidence
treewidth} of $F$ is
\[
\mathrm{tw}_{\mathrm{plog}}^*(F):=
\min_{(T,\chi,r)}
w_{\mathrm{plog}}(T,\chi,r),
\]
where the minimum is taken over all nice tree decompositions
$(T,\chi,r)$ of $G^*(F)$.

\begin{remark}
In the definition of partially log-weighted incidence
treewidth, we consider only nice tree decompositions.
If we instead take the minimum over all rooted tree
decompositions and all clause-path families, denote
the resulting parameter by $w$. We do not know whether $\mathrm{tw}_{\mathrm{plog}}^*(F)$
admits an upper bound linear in $w$, but we can prove that
$\mathrm{tw}_{\mathrm{plog}}^*(F)\leq(w+1)^2-1$. Consequently, if $F$ is unsatisfiable, with $n$ variables,
$m$ clauses, and maximum clause width $k$, then
Theorem~\ref{thm:main-bounds}(4) implies a resolution
refutation length bound of $k^2(n+m)2^{O(w^2)}$,
which is still an FPT bound parameterized by $w$.

To prove $\mathrm{tw}_{\mathrm{plog}}^*(F)\leq(w+1)^2-1$, let $(T,\chi,r)$ and $\mathcal P=(P_c)_{c\in F}$
have partially log-weighted width $w$.
Since every vertex has weight at least one, each bag
contains at most $w+1$ vertices.
For each clause $c$, if some node $t\notin P_c$ satisfies
$c\in\chi(t)$, then
$\max\{\log(\rho_{\mathcal P}(c)+1),1\}\leq w+1$,
since the total weight of $\chi(t)$ is at most $w+1$.
Otherwise, every bag containing $c$ lies on $P_c$.
Every variable of $c$ appears together with $c$ in some
bag, so $\rho_{\mathcal P}(c)=0$.
Thus $\max\{\log(\rho_{\mathcal P}(c)+1),1\}\leq w+1$
for every $c\in F$.

We convert $(T,\chi,r)$ into a nice tree decomposition
$(T',\chi',r')$ using the operations for converting a tree
decomposition into a nice tree decomposition described
in the proof of Lemma~\ref{lem:nice-empty-log-incidence} (see Appendix~\ref{app:weighted-decompositions}).
By the construction, for each $t\in V(T)$, we can choose a node $t^*\in V(T')$
such that $\chi'(t^*)=\chi(t)$, with $u^*$ an ancestor
of or equal to $v^*$ whenever $u$ is an ancestor of $v$.
For each nonempty $P_c$, the nodes $t^*$ with $t\in P_c$
therefore lie on a path in $T'(c)$ from a node to one of
its descendants.
Choose a root-to-leaf path $P'_c$ in $T'(c)$ containing
these nodes; if $P_c$ is empty, let $P'_c$ be empty.
Let $\mathcal P'=(P'_c)_{c\in F}$.
Every variable appearing in a bag on $P_c$ also appears
in a bag on $P'_c$, so
$\rho_{\mathcal P'}(c)\leq\rho_{\mathcal P}(c)$.
Every new bag is contained in a bag of $(T,\chi,r)$
and therefore contains at most $w+1$ vertices.
Each vertex has weight at most $w+1$ in every new bag
containing it.
Hence
$\mathrm{tw}_{\mathrm{plog}}^*(F)\leq(w+1)^2-1$.
\end{remark}

The following two lemmas prove the existence of
polynomial-size nice tree decompositions attaining the
minimum log-weighted and partially log-weighted widths,
respectively. Their proofs are given in
Appendix~\ref{app:weighted-decompositions}.

\begin{restatable}{lemma}{NiceLogDecomposition}
\label{lem:nice-empty-log-incidence}
Let $F$ be a nonempty CNF formula with $n$ variables and $m$ clauses.
Then $G^*(F)$ has a nice tree decomposition $(T,\chi,r)$ of
log-weighted width $\mathrm{tw}_{\log}^*(F)$ with
$O((n+m)(\mathrm{tw}_{\log}^*(F)+1))$ nodes.
\end{restatable}

\begin{restatable}{lemma}{PlogDecompositionSize}
\label{lem:plog-decomposition-size}
Let $F$ be a CNF formula of maximum clause width $k$,
with $n$ variables and $m$ clauses, and let $(T,\chi,r)$
be a nice tree decomposition of $G^*(F)$ such that
$w_{\mathrm{plog}}(T,\chi,r)=w$.
Then there is a nice tree decomposition $(T',\chi',r')$
of $G^*(F)$ such that
\[
w_{\mathrm{plog}}(T',\chi',r')\leq w
\]
and
\[
|V(T')|=O(k(n+m)w).
\]
\end{restatable}

\begin{theorem}\label{thm:plog-comparison}
Let $F$ be a nonempty CNF formula. Then
\[
\mathrm{tw}^*(F)
\leq \mathrm{tw}_{\mathrm{plog}}^*(F)
\leq \mathrm{tw}_{\log}^*(F).
\]
\end{theorem}

\begin{proof}
Since every vertex has weight at least one in the definition
of partially log-weighted width, we have
$\mathrm{tw}^*(F)\leq\mathrm{tw}_{\mathrm{plog}}^*(F)$.

By Lemma~\ref{lem:nice-empty-log-incidence}, there is a nice
tree decomposition $(T,\chi,r)$ of $G^*(F)$ of log-weighted
width $\mathrm{tw}_{\log}^*(F)$.
For each $c\in F$, choose a variable $x\in\mathrm{var}(c)$
and a node $t \in V(T)$ such that $\{x,c\}\subseteq\chi(t)$.
Choose a root-to-leaf path $P_c$ in $T(c)$ containing $t$.
Let $\mathcal P=(P_c)_{c\in F}$.
Then $\rho_{\mathcal P}(c)\leq |c|-1$ for every $c\in F$,
so the weight of $c$ in every bag containing it is at most
$\max\{\log|c|,1\}$.
Consequently,
$\mathrm{tw}_{\mathrm{plog}}^*(F)
\leq w_{\mathrm{plog}}(T,\chi,r,\mathcal P)
\leq\mathrm{tw}_{\log}^*(F)$.
\end{proof}

Cal\`i and Razgon~\cite{CaliRazgon2022} introduced
one-sided incidence treewidth. A rooted tree decomposition $(T,\chi,r)$ of $G^*(F)$ is
\emph{one-sided} if, for every $c\in F$, the nodes $t$ with
$c\in\chi_c(t)$ induce a path from a node to one of its
descendants.
The \emph{one-sided incidence treewidth} of $F$, denoted by
$\mathrm{tw}_{\mathrm{os}}(F)$, is the minimum of $w(T,\chi)$
over all one-sided tree decompositions $(T,\chi,r)$ of $G^*(F)$. The following theorem shows how small one-sided incidence
treewidth yields small partially log-weighted incidence
treewidth. Thus, partially log-weighted incidence treewidth
can be viewed as a generalization of one-sided incidence
treewidth.

\begin{theorem}\label{thm:plog-one-sided}
Let $F$ be a nonempty CNF formula, and let $(T,\chi,r)$ be a one-sided
tree decomposition of $G^*(F)$ of width
$\mathrm{tw}_{\mathrm{os}}(F)$.
Then there exist a nice tree decomposition $(T',\chi',r')$
of $G^*(F)$ and a clause-path family $\mathcal P=(P_c)_{c\in F}$
such that
$|V(T')|=O(|V(T)|(\mathrm{tw}_{\mathrm{os}}(F)+1))$
and
$w_{\mathrm{plog}}(T',\chi',r',\mathcal P)
\leq\mathrm{tw}_{\mathrm{os}}(F)$.
In particular,
$\mathrm{tw}_{\mathrm{plog}}^*(F)
\leq\mathrm{tw}_{\mathrm{os}}(F)$.
\end{theorem}

\begin{proof}
We use the construction that converts a tree decomposition
into a nice tree decomposition in the proof of
Lemma~\ref{lem:nice-empty-log-incidence} (see Appendix~\ref{app:weighted-decompositions}) to transform
$(T,\chi,r)$ into a nice tree decomposition $(T',\chi',r')$.
Every bag of $(T',\chi',r')$ is contained in a bag of
$(T,\chi,r)$, so
$w(T',\chi')\leq\mathrm{tw}_{\mathrm{os}}(F)$.
By an argument similar to the one used to bound the number
of nodes in the proof of
Lemma~\ref{lem:nice-empty-log-incidence}, we have
$|V(T')|=O(|V(T)|(\mathrm{tw}_{\mathrm{os}}(F)+1))$.

For each $t\in V(T)$, if $t$ is replaced by a binary tree
in the construction described in the proof of
Lemma~\ref{lem:nice-empty-log-incidence},
let $s_t$ be the root of that binary tree.
Otherwise, let $s_t=t$.
Let $t^*\in V(T')$ be the node corresponding to $s_t$
after the contractions described in the proof of
Lemma~\ref{lem:nice-empty-log-incidence}.
Since we contract only adjacent nodes with the same bags,
$\chi'(t^*)=\chi(t)$.
The construction also ensures that $u^*$ is an ancestor
of or equal to $v^*$ whenever $u$ is an ancestor of $v$ in $T$.

Fix $c\in F$.
Since $(T,\chi,r)$ is one-sided, the nodes whose bags
contain $c$ form a path from a node $a_c$ to a descendant $b_c$.
Then $c\in\chi'(a_c^*)\cap\chi'(b_c^*)$. So every bag on
the path from $a_c^*$ to $b_c^*$ contains $c$ by the
definition of tree decompositions.
Thus this path lies in $T'(c)$.
Since $a_c^*$ is an ancestor of $b_c^*$, or equal to $b_c^*$,
we can choose a path $P_c$ from the root of $T'(c)$
to a leaf of $T'(c)$ containing both nodes.
Let $\mathcal P=(P_c)_{c\in F}$ be the resulting
clause-path family.

Every node $t\in V(T)$ with $c\in\chi(t)$ lies on
the path from $a_c$ to $b_c$.
By the ancestor relation above, $t^*$ lies on
the path from $a_c^*$ to $b_c^*$ in $T'$.
For every $x\in\mathrm{var}(c)$, some bag $\chi(t)$ of $T$
contains both $x$ and $c$.
Thus $t^*\in P_c$ and $x\in\chi'_v(t^*)$.
It follows that $\rho_{\mathcal P}(c)=0$ for every $c\in F$.
Hence every clause vertex has weight one in every bag
containing it, and
\[
\mathrm{tw}_{\mathrm{plog}}^*(F)
\leq w_{\mathrm{plog}}(T',\chi',r',\mathcal P)
=w(T',\chi')
\leq\mathrm{tw}_{\mathrm{os}}(F).
\]
\end{proof}

\subsection{Multiset proof systems} \label{sec:multiset-proof-systems}

As stated in Subsection~\ref{subsec:introproofsys}, to simplify the description of our constructions, we introduce
two new proof systems, $\mathrm{mRes}$ and
$k\mathrm{\text{-}mDNF\text{-}Res}$, which can be viewed as
multiset versions of the resolution proof system and
$k\mathrm{\text{-}DNF\text{-}Res}$, respectively. These two proof systems do not differ essentially from
resolution and $k\mathrm{\text{-}DNF\text{-}Res}$, respectively;
we introduce them mainly to make it easier to deal with
some technical details of our constructions.

A \emph{multiclause} is a finite multiset of literals. For a
multiclause $C$ and a literal $\ell$, let $m_C(\ell)$ denote the
multiplicity of $\ell$ in $C$. For multiclauses $C$ and $D$, their
multiset disjunction $C\mlor D$ is defined by
\[
m_{C\mlor D}(\ell):=m_C(\ell)+m_D(\ell)
\]
for every literal $\ell$. A literal is identified with the
corresponding singleton multiclause, and the empty multiclause is
denoted by $\bot$. Define
\[
\operatorname{red}(C):=\{\ell:m_C(\ell)>0\},
\]
which is the ordinary clause obtained from $C$ by deleting repeated
literals.

The proof system $\mathrm{mRes}$ is defined using multiclauses. Given
a CNF formula $F$, its initial multiclauses are clauses $c \in F$,
with each literal having multiplicity one. It also has the axioms
$x\mlor\neg x$. Its inference rules are
\[
\frac{C}{C\mlor E}
\quad(\text{weakening}),
\qquad
\frac{C\mlor x\qquad D\mlor\neg x}{C\mlor D}
\quad(\text{resolution}),
\]
\[
\frac{C\mlor\ell\mlor\ell}{C\mlor\ell}
\quad(\text{contraction}),
\]
where $C,D,E$ are multiclauses, $x$ is a variable, and $\ell$ is a
literal. In the resolution rule, $C$ and $D$ may themselves contain
literals on $x$.

An $\mathrm{mRes}$ derivation from $F$ is a sequence of
multiclauses $\pi=(E_1,\ldots,E_s)$ such that each $E_i$ is
an initial multiclause, an axiom, or is obtained by zero or
more applications of the contraction rule from a multiclause
$E_i'$ satisfying one of the following:
\begin{enumerate}
    \item $E_i'=E_j$ for some $j<i$;
    \item $E_i'$ is derived from earlier multiclauses by one
    application of the weakening rule or the resolution rule.
\end{enumerate}
The length of $\pi$ is $|\pi|:=s$.

An $\mathrm{mRes}$ refutation of $F$ is a derivation
whose final multiclause is empty.

\begin{remark}
We adopt this slightly unconventional definition of an $\mathrm{mRes}$ derivation
for the following reasons.

First, we wish to treat contraction as an auxiliary operation
rather than count each application of the contraction rule as
a separate step. Our definition allows any number of these
applications within a single derivation step, so the length
does not depend on the number of contractions used to obtain
$E_i$ from $E_i'$.

Second, although an application of the resolution rule to two
multiclauses may produce a multiclause $E_i'$ of large width,
when bounding the width of our derivation we only need to
consider $E_i$, obtained from $E_i'$ by applications of the
contraction rule, rather than $E_i'$ itself. For example, let $C$ be a clause. In $\mathrm{mRes}$, an application of the resolution rule to $C\mlor x$ and
$C\mlor\neg x$ gives $C\mlor C$, which can be reduced to $C$
by applications of the contraction rule. In the resolution proof system,
an application of the resolution rule to $C\lor x$ and
$C\lor\neg x$ gives $C$ directly.
This conversion motivates excluding the intermediate
multiclause $C\mlor C$ when considering the width of the
$\mathrm{mRes}$ derivation.

Third, this definition allows us to describe our constructions more concisely and precisely.
\end{remark}

An \emph{mDNF} is a finite multiset of terms. Here a term is still a
set of literals, rather than a multiset of literals; this is
sufficient for our construction. For an mDNF $D$ and a term $T$,
let $m_D(T)$ denote the multiplicity of $T$ in $D$. For mDNFs
$D_1$ and $D_2$, their multiset disjunction is defined by
$m_{D_1\mlor D_2}(T):=m_{D_1}(T)+m_{D_2}(T)$ for every term $T$.
A term is identified with the mDNF containing that term with
multiplicity one, and the empty mDNF is denoted by $\bot$.
Define $\operatorname{red}(D):=\{T:m_D(T)>0\}$, which is the DNF
obtained from $D$ by deleting repeated terms. For $k\geq 1$, a
$k$-mDNF is an mDNF whose terms contain at most $k$ literals.

The proof system $k\mathrm{\text{-}mDNF\text{-}Res}$ is defined on
$k$-mDNFs. For each clause $c=\ell_1\lor\cdots\lor\ell_r$ of $F$,
the mDNF $\ell_1\mlor\cdots\mlor\ell_r$ is an
\emph{initial mDNF}, where each $\ell_i$ is viewed
as a singleton term. The system also
has the axioms $x\mlor\neg x$ and $\top$. Its inference rules are
\[
\frac{D}{D\mlor E}
\quad(\text{weakening}),
\]
\[
\frac{D\mlor T}{D\mlor T'}
\quad(\land\text{-elimination}),\qquad T'\subseteq T,
\]
\[
\frac{D_1\mlor T_1\qquad D_2\mlor T_2}
     {D_1\mlor D_2\mlor(T_1\land T_2)}
\quad(\land\text{-introduction}),
\]
\[
\frac{
D_1\mlor\ell_1\mlor\cdots\mlor\ell_q
\qquad
D_2\mlor(\neg\ell_1\land\cdots\land\neg\ell_q)
}{
D_1\mlor D_2
}
\quad(\text{cut}),
\]
\[
\frac{D\mlor T\mlor T}{D\mlor T}
\quad(\text{contraction}),
\]
where $D,E,D_1,D_2$ are mDNFs, $T,T',T_1,T_2$ are terms, and
$\{\ell_1,\ldots,\ell_q\}$ is a nonempty finite set of literals.
Every formula in a rule application must be a $k$-mDNF.
All term multiplicities in $D_1\mlor D_2$ are retained.

A $k\mathrm{\text{-}mDNF\text{-}Res}$ derivation from $F$ is a
sequence of $k$-mDNFs $\pi=(E_1,\ldots,E_s)$ such that each
$E_i$ is an initial mDNF, an axiom, or is obtained by zero or
more applications of the contraction rule from a $k$-mDNF
$E_i'$ satisfying one of the following:
\begin{enumerate}
    \item $E_i'=E_j$ for some $j<i$;
    \item $E_i'$ is derived from earlier mDNFs by one application
    of a rule other than contraction.
\end{enumerate}
The length of $\pi$ is $|\pi|:=s$. We call $\pi$ a
$k\mathrm{\text{-}mDNF\text{-}Res}$ refutation of $F$ if $E_s$
is the empty mDNF.

The \emph{width} of a multiclause $C$ is its size as a
multiset:
\[
\operatorname{width}(C)
:=
\sum_{\ell}m_C(\ell).
\]
The width of an $\mathrm{mRes}$ derivation
$\pi=(C_1,\ldots,C_s)$ is
\[
\operatorname{width}(\pi)
:=
\max_{1\le i\le s}\operatorname{width}(C_i).
\]

A multiclause is \emph{tautological} if it contains both $x$ and $\neg x$
for some variable $x$. An application of the resolution rule to
$P=C\mlor x$ and $Q=D\mlor\neg x$ is called a
\emph{variable-eliminating} resolution if
$m_P(x)=m_Q(\neg x)=1$ and $m_P(\neg x)=m_Q(x)=0$,
or equivalently, if neither $x$ nor $\neg x$ belongs to
$C\mlor D$.

The following two lemmas show that our length and width bounds
for refutations in the multiset proof systems also hold,
up to constant factors, for the original proof systems.
Thus, it suffices to construct refutations in the multiset
proof systems.
The proofs of these lemmas are given in
Appendix~\ref{app:multiset-proofs}.

\begin{restatable}{lemma}{MResToRes}\label{lem:mres-to-res}
Every $\mathrm{mRes}$ refutation $\pi$ of a CNF formula $F$ can be
transformed into a resolution refutation of $F$ of length at most
$|\pi|$ and width at most $\operatorname{width}(\pi)$.
Moreover, every resolution step in the resulting derivation
arises from a variable-eliminating resolution step in $\pi$.
If on every directed path in the proof DAG of $\pi$, each
variable is resolved in at most one variable-eliminating resolution
step, then the resulting refutation is regular.
\end{restatable}

\begin{restatable}{lemma}{MDNFResToDNFRes}\label{lem-mdnfres-to-dnfres}
Every $k\mathrm{\text{-}mDNF\text{-}Res}$ refutation $\pi$
of a CNF formula $F$ can be transformed into a
$k\mathrm{\text{-}DNF\text{-}Res}$ refutation of $F$
of length at most $3|\pi|$.
\end{restatable}

\subsection{Inconsistent states}
\label{subsec:inconsistent-states}

Except for Lemma~\ref{lem:leaf-base}, the proofs of the lemmas
in this subsection are given in
Appendix~\ref{app:inconsistent-states}.

In this subsection, we fix a CNF formula $F$ and a nice tree decomposition
$(T,\chi)$ of the incidence graph $G^*(F)$.

For $t \in V(T)$, let $T_t$ be the subtree rooted at $t$.
Define $\mathrm{Var}_t := \bigcup_{u \in V(T_t)} \chi_v(u)$ and $\mathrm{Cls}_t := \bigcup_{u \in V(T_t)} \chi_c(u)$.
For a node $t$, an assignment $\alpha : \chi_v(t) \to \{0,1\}$, and a set $A \subseteq \chi_c(t)$, let $N(t,\alpha,A)$ be the set of assignments $\tau : \mathrm{Var}_t \to \{0,1\}$ such that $\tau|_{\chi_v(t)}=\alpha$, every clause in $A$ is satisfied by $\tau$, and every clause in $\mathrm{Cls}_t \setminus \chi_c(t)$ is satisfied by $\tau$.
We say that $(t,\alpha,A)$ is \emph{inconsistent} if $N(t,\alpha,A)=\emptyset$, and \emph{consistent} otherwise.

Our definition of $N(t,\alpha,A)$ is inspired by the
\#SAT algorithm of Samer and Szeider~\cite{SamerSzeider2010}.
However, to construct
resolution refutations of small width, we modify their definition.

We first state three auxiliary lemmas about properties of
nice incidence tree decompositions, which will be useful
in our later proofs.

\begin{restatable}{lemma}{JoinVarStructure}\label{lem:join-var-structure}
Let $t$ be a join node with children $t_1,t_2$. Then the following hold.

(1) $\mathrm{Var}_{t_1} \cap \mathrm{Var}_{t_2} = \chi_v(t)$ and $\mathrm{Var}_{t_1} \cup \mathrm{Var}_{t_2} = \mathrm{Var}_t$.

(2) $\mathrm{Cls}_{t_1} \cup \mathrm{Cls}_{t_2} = \mathrm{Cls}_t$. For every clause $c\in \mathrm{Cls}_t\setminus \chi_c(t)$ and each $i\in\{1,2\}$, if $c\in \mathrm{Cls}_{t_i}$ and $c$ is satisfied by $\tau:\mathrm{Var}_t\to\{0,1\}$, then $c$ is satisfied by $\tau|_{\mathrm{Var}_{t_i}}$.
\end{restatable}

\begin{restatable}{lemma}{IntroduceNoX}\label{lem:introduce-no-x}
Let $t$ be an introduce-variable node introducing $x$ with child $t'$. Then every clause $c \in \mathrm{Cls}_{t'} \setminus \chi_c(t')$ contains no literal on $x$.
\end{restatable}

\begin{restatable}{lemma}{IntroduceClauseLocal}\label{lem:introduce-clause-local}
Let $t$ be an introduce-clause node introducing $c$, with child $t'$. Then
\[
\mathrm{var}(c)\cap\mathrm{Var}_t=\mathrm{var}(c)\cap\chi_v(t).
\]
Moreover, if an assignment $\tau:\mathrm{Var}_t\to\{0,1\}$ satisfies $c$, then $\tau|_{\chi_v(t)}$ also satisfies $c$.
\end{restatable}

Lemmas~\ref{lem:leaf-base}, \ref{lem:join-zero-recursion},
\ref{lem:introduce-var-zero}, \ref{lem:introduce-clause-zero},
\ref{lem:forget-var-zero}, and~\ref{lem:forget-clause-zero}
show how to compute the inconsistent states from the leaves
to the root of the nice tree decomposition of the incidence graph.

\begin{lemma}\label{lem:leaf-base}
Let $t$ be a leaf of the nice tree decomposition, and assume $\chi(t)=\emptyset$. Then $(t,\emptyset,\emptyset)$ is consistent. In particular, $t$ has no inconsistent state.
\end{lemma}

\begin{proof}
Since $t$ is a leaf and $\chi(t)=\emptyset$, we have $\chi_v(t)=\mathrm{Var}_t=\emptyset$ and $\chi_c(t)=\mathrm{Cls}_t=\emptyset$. Therefore the only state at $t$ is $(t,\emptyset,\emptyset)$. The empty assignment belongs to $N(t,\emptyset,\emptyset)$, since all relevant variable and clause sets are empty. Thus $(t,\emptyset,\emptyset)$ is consistent.
\end{proof}

\begin{restatable}{lemma}{JoinZeroRecursion}\label{lem:join-zero-recursion}
Let $t$ be a join node with children $t_1,t_2$. Let $\alpha:\chi_v(t)\to\{0,1\}$ and $A\subseteq \chi_c(t)$. Then $(t,\alpha,A)$ is inconsistent if and only if for all $A_1,A_2\subseteq \chi_c(t)$ with $A_1\cup A_2=A$, at least one of $(t_1,\alpha,A_1)$ or $(t_2,\alpha,A_2)$ is inconsistent.
\end{restatable}

\begin{restatable}{lemma}{IntroduceVarZero}\label{lem:introduce-var-zero}
Let $t$ be an introduce-variable node introducing $x$ with child $t'$. Let $\alpha:\chi_v(t)\to\{0,1\}$ and $A\subseteq \chi_c(t)$. Let $\ell:=x$ if $\alpha(x)=1$, and let $\ell:=\neg x$ otherwise. Define
\[
A_x:=\{c\in A:\ell\notin c\}.
\]
Then $(t,\alpha,A)$ is inconsistent if and only if $(t',\alpha|_{\chi_v(t')},A_x)$ is inconsistent.
\end{restatable}

\begin{restatable}{lemma}{IntroduceClauseZero}\label{lem:introduce-clause-zero}
Let $t$ be an introduce-clause node introducing $c$ with child $t'$. Let $\alpha:\chi_v(t)\to\{0,1\}$ and $A\subseteq \chi_c(t)$.

Then $(t,\alpha,A)$ is inconsistent if and only if one of the following holds:

(a) $c\notin A$ and $(t',\alpha,A)$ is inconsistent;

(b) $c\in A$, $\alpha$ satisfies $c$, and $(t',\alpha,A\setminus\{c\})$ is inconsistent;

(c) $c\in A$ and $\alpha$ does not satisfy $c$.

\end{restatable}

\begin{restatable}{lemma}{ForgetVarZero}\label{lem:forget-var-zero}
Let $t$ be a forget-variable node with unique child $t'$ forgetting a variable $x$. Let $\alpha:\chi_v(t)\to\{0,1\}$ and $A\subseteq \chi_c(t)$. For $b\in\{0,1\}$, let $\alpha_b:=\alpha\cup\{x\mapsto b\}$.

Then $(t,\alpha,A)$ is inconsistent if and only if both $(t',\alpha_0,A)$ and $(t',\alpha_1,A)$ are inconsistent.
\end{restatable}

\begin{restatable}{lemma}{ForgetClauseZero}\label{lem:forget-clause-zero}
Let $t$ be a forget-clause node with unique child $t'$ forgetting a clause $c$. Let $\alpha:\chi_v(t)\to\{0,1\}$ and $A\subseteq \chi_c(t)$.

Then $(t,\alpha,A)$ is inconsistent if and only if $(t',\alpha,A \cup \{ c\})$ is inconsistent.
\end{restatable}

\section{General resolution upper bounds}

\subsection{Local derivations}

In this subsection, we show how to represent inconsistent states
by mDNFs and derive the mDNFs representing inconsistent states
at a node from those representing inconsistent states at its
children.

Throughout this subsection, let $F$ be a CNF formula with maximum clause
width $k\geq 1$, and let $(T,\chi,r)$ be a rooted nice tree
decomposition of its incidence graph. Fix a clause-path family
$\mathcal P=(P_c)_{c\in F}$ for $(T,\chi,r)$, and let
$w:=w(T,\chi)$, $w_{\log}:=w_{\log}(T,\chi)$, and
$w_{\mathrm{plog}}:=w_{\mathrm{plog}}(T,\chi,r,\mathcal P)$.

For $c\in\chi_c(t)$, define its \emph{internal part} and
\emph{external part} by
\[
c_t^{\mathrm{int}}
:=
\bigvee_{\substack{\ell\in c\\
\mathrm{var}(\ell)\in\mathrm{Var}_t}}\ell,
\qquad
c_t^{\mathrm{ext}}
:=
\bigvee_{\substack{\ell\in c\\
\mathrm{var}(\ell)\notin\mathrm{Var}_t}}\ell.
\]
Thus $c=c_t^{\mathrm{int}}\lor c_t^{\mathrm{ext}}$.

For an assignment $\alpha:\chi_v(t)\to\{0,1\}$, define
\[
C_\alpha
:=
\bigvee_{x\in\chi_v(t)}
\begin{cases}
\neg x, & \alpha(x)=1,\\
x,      & \alpha(x)=0.
\end{cases}
\]
Thus $C_\alpha$ contains, for each $x\in\chi_v(t)$, the literal on
$x$ falsified by $\alpha$. If $\chi_v(t)=\emptyset$, then
$C_\alpha=\bot$.

For a state $(t,\alpha,A)$ and a set $B\subseteq A$, define the
mDNF
\[
D(t,\alpha,A,B)
:=
C_\alpha
\mlor
(\mathop{\mlor}_{c\in B}c_t^{\mathrm{ext}})
\mlor
(\mathop{\mlor}_{c\in A\setminus B}
\neg c_t^{\mathrm{int}}).
\]
Here $c_t^{\mathrm{ext}}$ is viewed as an mDNF consisting of singleton terms, while $\neg c_t^{\mathrm{int}}$ is viewed as a one-term mDNF.

If $\chi(t)=\emptyset$, then $A \subseteq \chi_c(t)$ must be empty, and $D(t,\alpha,\emptyset,\emptyset)=\bot$.

For an mDNF $D$, let $\operatorname{size}(D)$ denote the number of terms in $D$.

\begin{lemma}\label{lem:D-size-tlp}
Let $F$ be a CNF formula, and let $(T,\chi,r)$ be a rooted tree
decomposition of the incidence graph of $F$.
For every $t\in V(T)$, assignment
$\alpha:\chi_v(t)\to\{0,1\}$, and set
$A\subseteq\chi_c(t)$,
\[
\operatorname{size}(D(t,\alpha,A,\emptyset))
\leq w(T,\chi)+1.
\]
\end{lemma}

\begin{proof}
The mDNF $D(t,\alpha,A,\emptyset)$ contains
$|\chi_v(t)|$ singleton terms from $C_\alpha$ and one term
$\neg c_t^{\mathrm{int}}$ for each $c\in A$. Hence
\[
\operatorname{size}(D(t,\alpha,A,\emptyset))
=
|\chi_v(t)|+|A|
\leq|\chi(t)|
\leq w(T,\chi)+1.
\]
\end{proof}

Lemmas~\ref{lem:derive-join}, \ref{lem:derive-introduce-var},
\ref{lem:derive-introduce-clause}, \ref{lem:derive-forget-var},
and~\ref{lem:derive-forget-clause} show how to derive
the mDNF representing an inconsistent state at a node
from the mDNFs representing inconsistent states at its children.

\begin{lemma}\label{lem:derive-join}
Let $t_0$ be a join node with children $t_1,t_2$, let
$\alpha:\chi_v(t_0)\to\{0,1\}$, and let
$B_0\subseteq A\subseteq\chi_c(t_0)$.
Fix a partition $B_0=B_1\sqcup B_2$.

If $(t_0,\alpha,A)$ is inconsistent, then
$D(t_0,\alpha,A,B_0)$ has a
$k\mathrm{\text{-}mDNF\text{-}Res}$ derivation from the formulas
$D(t_i,\alpha,S,S\cap B_i)$ with $i\in\{1,2\}$ and $S\subseteq A$
such that $(t_i,\alpha,S)$ is inconsistent.

Every mDNF $E$ in the derivation can be written as
\[
E=C_\alpha\mlor L\mlor\bigmlor_{c\in A'}\neg C_c,
\]
where $L$ consists only of singleton terms, $A'\subseteq A$,
and, for each $c\in A'$, there is an index
$i_c\in\{0,1,2\}$ such that $c\notin B_{i_c}$ and
$C_c\subseteq c_{t_{i_c}}^{\mathrm{int}}$.
If $B_0=\emptyset$, then $\operatorname{size}(E)\leq w+1$.

\end{lemma}

\begin{proof}
Let $A=\{c_1,\ldots,c_a\}$.
For each $c\in A$ and $i\in\{1,2\}$, let
$L_{c,i}:=c_{t_i}^{\mathrm{ext}}$ if $c\in B_i$, and
$L_{c,i}:=\neg c_{t_i}^{\mathrm{int}}$ otherwise.
Let $Q_c:=c_{t_0}^{\mathrm{ext}}$ if $c\in B_0$, and
$Q_c:=\neg c_{t_0}^{\mathrm{int}}$ otherwise.

We first describe how to derive $G\mlor Q_c$ from
$G\mlor L_{c,1}$ and $G\mlor L_{c,2}$, for any $c\in A$ and
any $k$-mDNF $G$.

If $c\in A\setminus B_0$, then
$Q_c=\neg c_{t_0}^{\mathrm{int}}
=\neg c_{t_1}^{\mathrm{int}}\land\neg c_{t_2}^{\mathrm{int}}
=L_{c,1}\land L_{c,2}$, since
$\mathrm{Var}_{t_0}=\mathrm{Var}_{t_1}\cup\mathrm{Var}_{t_2}$.
An application of the $\land$-introduction rule, followed by
contraction, gives $G\mlor Q_c$.

Suppose instead that $c\in B_1$, so $c\notin B_2$.
Let $H$ be the subclause of $c$ consisting of the literals whose
variables belong to $\mathrm{Var}_{t_2}\setminus\chi_v(t_0)$.
Then $H\subseteq c_{t_2}^{\mathrm{int}}$.
By Lemma~\ref{lem:join-var-structure},
$L_{c,1}=c_{t_1}^{\mathrm{ext}}
=c_{t_0}^{\mathrm{ext}}\mlor H=Q_c\mlor H$.
Since $c\notin B_2$, the other formula is
$G\mlor L_{c,2}=G\mlor\neg c_{t_2}^{\mathrm{int}}$.
If $H=\bot$, then $G\mlor L_{c,1}$ is already $G\mlor Q_c$.
Otherwise, apply the $\land$-elimination rule, if needed, to
$G\mlor L_{c,2}$ to obtain $G\mlor\neg H$.
We apply the cut rule to $G\mlor Q_c\mlor H$ and
$G\mlor\neg H$, and then use the contraction rule to obtain
$G\mlor Q_c$.
The case $c\in B_2$ is symmetric.

In the $\land$-introduction and cut steps above, the contractions
reduce $G\mlor G\mlor Q_c$ to $G\mlor Q_c$.

For every partition $A=A_1\sqcup A_2$, we first derive
\[
C_\alpha
\mlor \bigmlor_{c\in A_1}L_{c,1}
\mlor \bigmlor_{c\in A_2}L_{c,2}.
\]
By Lemma~\ref{lem:join-zero-recursion},
$(t_i,\alpha,A_i)$ is inconsistent for some $i\in\{1,2\}$.
The required formula follows from
$D(t_i,\alpha,A_i,A_i\cap B_i)$ by one application of
the weakening rule, if needed.

We proceed by induction on $h=0,\ldots,a$ to derive
\[
C_\alpha
\mlor \bigmlor_{s=1}^{h}Q_{c_s}
\mlor \bigmlor_{c\in A_1}L_{c,1}
\mlor \bigmlor_{c\in A_2}L_{c,2}
\]
for every partition
$A\setminus\{c_1,\ldots,c_h\}=A_1\sqcup A_2$.
The case $h=0$ is established above.
For $h\geq 1$, fix such a partition.
By the induction hypothesis, the formulas for the partitions
$(A_1\cup\{c_h\},A_2)$ and $(A_1,A_2\cup\{c_h\})$
have already been derived.
They have the form $G\mlor L_{c_h,1}$ and
$G\mlor L_{c_h,2}$ with the same $G$.
The preceding construction gives $G\mlor Q_{c_h}$,
which is the required formula.
For $h=a$, we have $A_1=A_2=\emptyset$, and the resulting
formula is $D(t_0,\alpha,A,B_0)$.

It is easy to check that every mDNF $E$ in the derivation
can be written as
\[
E=C_\alpha\mlor\bigmlor_{c\in A}L_c,
\]
where each $L_c$ is either an empty set or, for some
$i\in\{0,1,2\}$, one of the following:
a single term $\neg C'$ with
$C'\subseteq c_{t_i}^{\mathrm{int}}$ and $c\notin B_i$;
or a clause $c_{t_i}^{\mathrm{ext}}$, viewed as an mDNF consisting
of singleton terms, with $c\in B_i$. Let $A':=\{c\in A:L_c\text{ is a single term $\neg C'$ with
$C'\subseteq c_{t_i}^{\mathrm{int}}$}\}$
and $L:=\bigmlor\limits_{c\in A\setminus A'}L_c$.
Then $E$ has the representation stated in the lemma.
Every term in $E$ is either a singleton term or the negation of a
subclause of a clause in $F$, and therefore contains at most
$k$ literals. Hence the construction gives a
$k\mathrm{\text{-}mDNF\text{-}Res}$ derivation.

For each partition $A=A_1\sqcup A_2$, we use one formula
$D(t_i,\alpha,A_i,A_i\cap B_i)$ with $(t_i,\alpha,A_i)$
inconsistent and at most one application of the weakening rule.
Since there are $2^a$ partitions, these derivations have total
length $O(2^a)$.
For each $h\in\{1,\ldots,a\}$, there are $2^{a-h}$ partitions
$A\setminus\{c_1,\ldots,c_h\}=A_1\sqcup A_2$.
For each partition, the induction step requires at most one
application of the $\land$-introduction rule if $c_h\notin B_0$,
and at most one application each of the $\land$-elimination
rule and the cut rule if $c_h\in B_0$.
Thus the induction step from $h-1$ to $h$ uses at most
$2\cdot 2^{a-h}$ rule applications.
Since $\sum_{h=1}^{a}2^{a-h}=2^a-1$, the total length is
$O(2^a)=O(2^w)$, where
$a\leq|\chi_c(t_0)|\leq w+1$.

If $B_0=\emptyset$, then $B_1=B_2=\emptyset$, so each $L_c$
contains at most one term. Consequently,
$\operatorname{size}(E)\leq|\chi_v(t_0)|+|A|
\leq|\chi(t_0)|\leq w+1$.

\end{proof}

\begin{lemma}\label{lem:derive-introduce-var}
Let $t$ be an introduce-variable node with child $t'$
introducing a variable $x$, let
$\alpha:\chi_v(t)\to\{0,1\}$, and let
$B\subseteq A\subseteq\chi_c(t)$.
Let $\alpha':=\alpha|_{\chi_v(t')}$.

If $(t,\alpha,A)$ is inconsistent, then $D(t,\alpha,A,B)$
has a $k\mathrm{\text{-}mDNF\text{-}Res}$ derivation from
the formulas $D(t',\alpha',S,S\cap B)$ with $S\subseteq A$
such that $(t',\alpha',S)$ is inconsistent.

The derivation has length $O(w)$.
Every mDNF $E$ in the derivation,
other than $D(t',\alpha',S,S\cap B)$ and the axiom, can be written as
\[
E=C_\alpha\mlor L\mlor\bigmlor_{c\in A'}\neg C_c,
\]
where $L$ consists only of singleton terms,
$A'\subseteq A\setminus B$, and, for each $c\in A'$,
$C_c= c_{u_c}^{\mathrm{int}}$ for some
$u_c\in\{t,t'\}$.

If $B=\emptyset$, then every mDNF $E$ in the derivation
satisfies $\operatorname{size}(E)\leq w+1$.
\end{lemma}

\begin{proof}
Let $\ell=x$ if $\alpha(x)=1$, and $\ell=\neg x$ otherwise.
Let $A_x:=\{c\in A:\ell\notin c\}$.
By Lemma~\ref{lem:introduce-var-zero},
$(t',\alpha',A_x)$ is inconsistent. We begin the derivation with
$D(t',\alpha',A_x,A_x\cap B)$.
We have $C_\alpha=C_{\alpha'}\mlor\neg\ell$.

For each $c\in A_x\cap B$, either $c$ contains neither $x$ nor
$\neg x$, in which case
$c_{t'}^{\mathrm{ext}}=c_t^{\mathrm{ext}}$, or
$\neg\ell\in c$, in which case
$c_{t'}^{\mathrm{ext}}=c_t^{\mathrm{ext}}\mlor\neg\ell$.
Starting from $D(t',\alpha',A_x,A_x\cap B)$, apply the
weakening rule to add $\neg\ell$, the clauses
$c_t^{\mathrm{ext}}$ for $c\in B\setminus A_x$, and the terms
$\neg c_t^{\mathrm{int}}$ for
$c\in A\setminus(A_x\cup B)$, then use the contraction rule to remove, for each
$c\in A_x\cap B$ with $\neg\ell\in c$, the singleton term
$\neg\ell$ contained in $c_{t'}^{\mathrm{ext}}$, while retaining
the singleton term $\neg\ell$ in $C_\alpha$.
The resulting formula is
\[
D_0:=C_\alpha
\mlor\bigmlor_{c\in B}c_t^{\mathrm{ext}}
\mlor\bigmlor_{c\in A_x\setminus B}\neg c_{t'}^{\mathrm{int}}
\mlor\bigmlor_{c\in A\setminus(A_x\cup B)}
\neg c_t^{\mathrm{int}}.
\]

It remains to replace $\neg c_{t'}^{\mathrm{int}}$ by
$\neg c_t^{\mathrm{int}}$ for $c\in A_x\setminus B$.
If $c$ contains neither $x$ nor $\neg x$, these terms are equal.
Otherwise, $\neg\ell\in c$, and hence
$\neg c_t^{\mathrm{int}}
=\neg c_{t'}^{\mathrm{int}}\land\ell$.

Suppose the current formula is
$G\mlor\neg c_{t'}^{\mathrm{int}}$.
Applying the $\land$-introduction rule to this formula
and the axiom $\ell\mlor\neg\ell$ gives
\[
G\mlor\neg\ell
\mlor\bigl(\neg c_{t'}^{\mathrm{int}}\land\ell\bigr).
\]
Since $G$ contains the singleton term $\neg\ell$ from
$C_\alpha$, one application of the contraction rule gives
$G\mlor\neg c_t^{\mathrm{int}}$.
Repeating this for every $c\in A_x\setminus B$ with
$\neg\ell\in c$, we obtain $D(t,\alpha,A,B)$.

It is easy to check that every mDNF $E$ in the derivation,
other than $D(t',\alpha',A_x,A_x\cap B)$ and the axiom
$\ell\mlor\neg\ell$, can be written as
\[
E=C_\alpha\mlor\bigmlor_{c\in A}L_c,
\]
where $L_c=c_t^{\mathrm{ext}}$ if $c\in B$, and, if
$c\in A\setminus B$, $L_c$ is a single term
$\neg c_{u_c}^{\mathrm{int}}$ for some $u_c\in\{t,t'\}$.
Taking
$A'=A\setminus B$ and
$L=\bigmlor\limits_{c\in B}c_t^{\mathrm{ext}}$
yields the representation of $E$ stated in the lemma.
Every term in such an $E$ is either a singleton term or the
negation of a subclause of a clause in $F$.

Every term of $D(t',\alpha',A_x,A_x\cap B)$ is also either a
singleton term or the negation of a subclause of a clause in $F$,
while both terms of the axiom $\ell\mlor\neg\ell$ are singleton
terms.
Therefore, every term appearing in the derivation contains at most
$k$ literals, and the construction gives a
$k\mathrm{\text{-}mDNF\text{-}Res}$ derivation.

The construction uses at most one application of the weakening
rule and at most one application of the $\land$-introduction rule
for each $c\in A_x\setminus B$.The axiom $\ell\mlor\neg\ell$ is included at most once
in the derivation.
Since $|A_x\setminus B|\leq |A|\leq w+1$, the derivation has
length $O(w)$.

Suppose that $B=\emptyset$. Then $A_x\cap B=\emptyset$, and
Lemma~\ref{lem:D-size-tlp} gives
$\operatorname{size}(D(t',\alpha',A_x,A_x\cap B))\leq w+1$.
The axiom $\ell\mlor\neg\ell$ has size two.
Every other mDNF in the construction consists
of $C_\alpha$ and one term for each $c\in A$, and hence has size
at most
\[
|\chi_v(t)|+|A|
\leq |\chi(t)|
\leq w+1.
\]
Thus every mDNF in the derivation has size at most $w+1$.
\end{proof}

\begin{lemma}\label{lem:derive-introduce-clause}
Let $t$ be an introduce-clause node with child $t'$ introducing
a clause $c$, let $\alpha:\chi_v(t)\to\{0,1\}$, and let
$B\subseteq A\subseteq\chi_c(t)$.

If $(t,\alpha,A)$ is inconsistent, then $D(t,\alpha,A,B)$ has a
$k\mathrm{\text{-}mDNF\text{-}Res}$ derivation from the formulas
$D(t',\alpha,S,S\cap B)$ with $S\subseteq A\setminus\{c\}$ such
that $(t',\alpha,S)$ is inconsistent.

The derivation has length $O(w)$. Apart from the formulas
$D(t',\alpha,S,S\cap B)$ used in the derivation, the initial
clause $c \in F$, the axioms $x\mlor\neg x$ and $\top$, and the final
formula $D(t,\alpha,A,B)$, every mDNF $E$ in the derivation has the form
\[
E=C_\alpha\mlor L\mlor\bigmlor_{d\in A'}\neg C_d,
\]
where $L$ consists only of singleton terms,
$A'\subseteq A\setminus B$, and
$C_d\subseteq d_t^{\mathrm{int}}$ for every $d\in A'$.
If $B=\emptyset$, then every such mDNF $E$ satisfies
$\operatorname{size}(E)\leq w+1$.
\end{lemma}

\begin{proof}
We distinguish the three cases in
Lemma~\ref{lem:introduce-clause-zero}.

\smallskip
\noindent
\emph{Case 1: $c\notin A$.}
By Lemma~\ref{lem:introduce-clause-zero}(a),
$(t',\alpha,A)$ is inconsistent.
Since $t$ introduces only the clause $c$, we have
$\chi_v(t)=\chi_v(t')$ and $\mathrm{Var}_t=\mathrm{Var}_{t'}$.
Thus the internal and external parts of every clause in $A$ are
the same at $t$ and $t'$, and hence
$D(t,\alpha,A,B)=D(t',\alpha,A,B)$.

\smallskip
\noindent
\emph{Case 2: $c\in A$ and $\alpha$ satisfies $c$.}
By Lemma~\ref{lem:introduce-clause-zero}(b),
$(t',\alpha,A\setminus\{c\})$ is inconsistent.
We begin with
$D(t',\alpha,A\setminus\{c\},B\setminus\{c\})$.
For every $d\in A\setminus\{c\}$, the internal and external
parts of $d$ are the same at $t$ and $t'$.
If $c\in B$, add
$c_t^{\mathrm{ext}}$ to $D(t',\alpha,A\setminus\{c\},B\setminus\{c\})$ by applying the weakening rule; if $c\notin B$, add
$\neg c_t^{\mathrm{int}}$ to $D(t',\alpha,A\setminus\{c\},B\setminus\{c\})$ by applying the weakening rule.
In either case, the resulting formula is $D(t,\alpha,A,B)$.

\smallskip
\noindent
\emph{Case 3: $c\in A$ and $\alpha$ does not satisfy $c$.}
By Lemma~\ref{lem:introduce-clause-local}, every variable of
$c_t^{\mathrm{int}}$ belongs to $\chi_v(t)$.
Since $\alpha$ does not satisfy $c$, every literal in $c_t^{\mathrm{int}}$ belongs to
$C_\alpha$.

Suppose first that $c\in B$.
Since every literal in $c_t^{\mathrm{int}}$ belongs to
$C_\alpha$, we begin with the initial clause
$c=c_t^{\mathrm{int}}\mlor c_t^{\mathrm{ext}}$ and apply the
weakening rule, if needed, to add the remaining singleton terms
of $C_\alpha$, the clauses $d_t^{\mathrm{ext}}$ for
$d\in B\setminus\{c\}$, and the terms
$\neg d_t^{\mathrm{int}}$ for $d\in A\setminus B$.
The resulting formula is $D(t,\alpha,A,B)$.

It remains to consider the case $c\notin B$.
We first derive
$C_\alpha\mlor\neg c_t^{\mathrm{int}}$.
Let
$c_t^{\mathrm{int}}=\ell_1\lor\cdots\lor\ell_m$. We have $\ell_1,\cdots,\ell_m \in C_\alpha$ by Lemma~\ref{lem:introduce-clause-local}.

Suppose that $m\geq 1$, and, for $j\in\{1,\ldots,m\}$, let
$T_j:=\neg\ell_1\land\cdots\land\neg\ell_j$. 
Starting with the axiom
$\ell_1\mlor\neg\ell_1$, we can use the weakening rule
to add the singleton terms of $C_\alpha$ other than $\ell_1$ to get $C_\alpha\mlor T_1$.
For each $j\in\{2,\ldots,m\}$, apply the
$\land$-introduction rule to
$C_\alpha\mlor T_{j-1}$ and the axiom
$\ell_j\mlor\neg\ell_j$ to get
$C_\alpha\mlor\ell_j\mlor T_j$ and contract it into
$C_\alpha\mlor T_j$.
After processing $\ell_1,\ldots,\ell_m$, we obtain
$C_\alpha\mlor T_m
=C_\alpha\mlor\neg c_t^{\mathrm{int}}$.

Having derived
$C_\alpha\mlor\neg c_t^{\mathrm{int}}$, apply the weakening rule,
if needed, to add the clauses $d_t^{\mathrm{ext}}$ for $d\in B$
and the terms $\neg d_t^{\mathrm{int}}$ for
$d\in A\setminus(B\cup\{c\})$.
The resulting formula is $D(t,\alpha,A,B)$.

If $m=0$, then $\neg c_t^{\mathrm{int}}=\top$.
We begin with the axiom $\top$ and apply the weakening rule,
if needed, to add $C_\alpha$, the clauses
$d_t^{\mathrm{ext}}$ for $d\in B$, and the terms
$\neg d_t^{\mathrm{int}}$ for
$d\in A\setminus(B\cup\{c\})$.
The resulting formula is $D(t,\alpha,A,B)$.

In Case~1, no rule application
is needed. In Case~2, the derivation uses one formula associated
with an inconsistent state at $t'$ and at most one application of
the weakening rule. In Case~3, if $c\in B$, the derivation uses the
initial clause $c$ and at most one application of the weakening
rule. If $c\notin B$ and $c_t^{\mathrm{int}}=\bot$, it
uses the axiom $\top$ and at most one application of the weakening
rule.

It remains to consider Case~3 when $c\notin B$ and
$c_t^{\mathrm{int}}=\ell_1\lor\cdots\lor\ell_m$ with $m\geq1$.
The construction uses the $m$ axioms
$\ell_j\mlor\neg\ell_j$, at most two applications of the weakening
rule, and $m-1$ applications of the $\land$-introduction rule.
By Lemma~\ref{lem:introduce-clause-local}, the variables occurring
in $c_t^{\mathrm{int}}$ belong to $\chi_v(t)$. Since $c$ also
belongs to $\chi_c(t)$, we have
$m+1\leq|\chi(t)|\leq w+1$, and hence $m\leq w$.
Thus the derivation has length $O(w)$.

Apart from the formulas
$D(t',\alpha,S,S\cap B)$ used in the derivation, the initial
clause $c$, the axioms $x\mlor\neg x$ and $\top$, and the final
formula $D(t,\alpha,A,B)$, the only mDNFs that occur in the derivation are
\[
C_\alpha\mlor
(\neg\ell_1\land\cdots\land\neg\ell_j),
\qquad 1\leq j\leq m.
\]
For each such mDNF, take $A'=\{c\}$,
$C_c=\ell_1\lor\cdots\lor\ell_j$, and let $L$ be the empty mDNF.
Then $C_c\subseteq c_t^{\mathrm{int}}$, so the mDNF has the form
stated in the lemma.

It is easy to check from the construction that every term is
either a singleton term or of the form $\neg C'$ for some
subclause $C'$ of a clause in $F$, where $C'$ may be empty.
Thus every term contains at most $k$ literals, and the construction
gives a $k\mathrm{\text{-}mDNF\text{-}Res}$ derivation.

Suppose that $B=\emptyset$. By
Lemma~\ref{lem:D-size-tlp}, every formula
$D(t',\alpha,S,\emptyset)$ used in the derivation, as well as the formula $D(t,\alpha,A,\emptyset)$, has size at most $w+1$.
The initial clause $c \in F$ is used only when $c\in B$, so it is not
used when $B=\emptyset$.
The axioms used in the derivation have size at most
$2\leq w+1$.

Every remaining mDNF has the form
$C_\alpha\mlor
(\neg\ell_1\land\cdots\land\neg\ell_j)$
for some $j\in\{1,\ldots,m\}$, and therefore has size
\[
|\chi_v(t)|+1\leq|\chi(t)|\leq w+1,
\]
in which the first inequality uses $c\in\chi_c(t)$. Thus every mDNF
in the derivation has size at most $w+1$.
\end{proof}

\begin{lemma}\label{lem:derive-forget-var}
Let $t$ be a forget-variable node with child $t'$ forgetting a
variable $x$, let $\alpha:\chi_v(t)\to\{0,1\}$, and let
$B\subseteq A\subseteq\chi_c(t)$. For $b\in\{0,1\}$, let
$\alpha_b:=\alpha\cup\{x\mapsto b\}$.

If $(t,\alpha,A)$ is inconsistent, then
$D(t,\alpha,A,B)$ has a
$k\mathrm{\text{-}mDNF\text{-}Res}$ derivation from
$D(t',\alpha_0,A,B)$ and $D(t',\alpha_1,A,B)$.
The derivation has length one. If $B=\emptyset$, then every mDNF
in the derivation has size at most $w+1$.
\end{lemma}

\begin{proof}
By Lemma~\ref{lem:forget-var-zero}, both
$(t',\alpha_0,A)$ and $(t',\alpha_1,A)$ are inconsistent.
Since $t$ forgets only $x$, we have
$\chi_c(t)=\chi_c(t')$ and
$\mathrm{Var}_t=\mathrm{Var}_{t'}$. Consequently,
$c_t^{\mathrm{int}}=c_{t'}^{\mathrm{int}}$ and
$c_t^{\mathrm{ext}}=c_{t'}^{\mathrm{ext}}$ for every $c\in A$.

Let
\[
G:=
C_\alpha
\mlor\bigmlor_{c\in B}c_t^{\mathrm{ext}}
\mlor\bigmlor_{c\in A\setminus B}\neg c_t^{\mathrm{int}}
=D(t,\alpha,A,B).
\]
Since $C_{\alpha_0}=C_\alpha\mlor x$ and
$C_{\alpha_1}=C_\alpha\mlor\neg x$, we have
\[
D(t',\alpha_0,A,B)=G\mlor x
\qquad\text{and}\qquad
D(t',\alpha_1,A,B)=G\mlor\neg x.
\]
Applying the cut rule to these two mDNFs and then contracting the
repeated terms in $G\mlor G$ gives $G=D(t,\alpha,A,B)$. Thus the
derivation has length one.

It follows directly from the definition of $D$ that all three
mDNFs are $k$-mDNFs. If $B=\emptyset$, their sizes are at most
$w+1$ by Lemma~\ref{lem:D-size-tlp}.
\end{proof}

\begin{lemma}\label{lem:forget-clause-internal}
Let $t$ be a forget-clause node with child $t'$, and suppose that $t$ forgets the clause $c$. Then $c_{t'}^{\mathrm{int}}=c$.
\end{lemma}

\begin{proof} Let $x\in \mathrm{var}(c)$. Since $\{x,c\}$ is an edge of $G^*(F)$, some bag contains both $x$ and $c$. The bags containing $c$ form a connected subtree. Since $c\in \chi_c(t')$ but $c\notin \chi_c(t)$, no bag containing $c$ can lie outside $T_{t'}$; otherwise the path to $t'$ would pass through $t$, forcing $c\in \chi_c(t)$. Hence the bag containing both $x$ and $c$ lies in $T_{t'}$, so $x\in \mathrm{Var}_{t'}$. Thus every variable of $c$ lies in $\mathrm{Var}_{t'}$, and therefore $c_{t'}^{\mathrm{int}}=c$. \end{proof}

\begin{lemma}\label{lem:derive-forget-clause}
Let $t$ be a forget-clause node with child $t'$ forgetting a
clause $c$, let $\alpha:\chi_v(t)\to\{0,1\}$, and let
$B\subseteq A\subseteq\chi_c(t)$.

If $(t,\alpha,A)$ is inconsistent, then $(t',\alpha,A\cup\{c\})$ is inconsistent, and $D(t,\alpha,A,B)$ is derivable in
$k\mathrm{\text{-}mDNF\text{-}Res}$ from each of the mDNFs
$D(t',\alpha,A\cup\{c\},B)$ and
$D(t',\alpha,A\cup\{c\},B\cup\{c\})$.
Both derivations have length at most one.
If $B=\emptyset$, then every mDNF in either derivation other than
the initial clause $c$ has size at most $w+1$.
\end{lemma}

\begin{proof}
By Lemma~\ref{lem:forget-clause-zero},
$(t',\alpha,A\cup\{c\})$ is inconsistent. Since $t$ forgets only
$c$, we have $\chi_v(t)=\chi_v(t')$ and
$\mathrm{Var}_t=\mathrm{Var}_{t'}$. Hence the internal and external
parts of every $d\in A$ are the same at $t$ and $t'$.
Moreover, Lemma~\ref{lem:forget-clause-internal} gives
$c_{t'}^{\mathrm{int}}=c$ and
$c_{t'}^{\mathrm{ext}}=c\setminus c_{t'}^{\mathrm{int}}$. Therefore,
\[
D(t',\alpha,A\cup\{c\},B)=D(t,\alpha,A,B)\mlor\neg c,
\]
\[
D(t',\alpha,A\cup\{c\},B\cup\{c\})=D(t,\alpha,A,B).
\]
Applying the cut rule to the first mDNF and the initial clause $c$
gives $D(t,\alpha,A,B)$, while no rule application is needed for
the second mDNF. Thus both derivations have length at most one.
All the mDNFs used are $k$-mDNFs.

Suppose that $B=\emptyset$. The size bounds for
$D(t',\alpha,A\cup\{c\},\emptyset)$ and
$D(t,\alpha,A,\emptyset)$ follow from
Lemma~\ref{lem:D-size-tlp}. The other formula at $t'$ equals
$D(t,\alpha,A,\emptyset)$. Hence every mDNF in either derivation,
other than the initial clause $c$, has size at most $w+1$.
\end{proof}

\subsection{From local derivations to refutations}

In this subsection, we present two constructions of
$k\mathrm{\text{-}mDNF\text{-}Res}$ refutations
(Theorems~\ref{thm:td-to-dnfres-tight}
and~\ref{thm:td-to-dnfres-plog}) and show how to convert
them into resolution refutations
(Theorems~\ref{thm:res-bound} and~\ref{thm:res-bound-plog}). 
The constructions in Theorems~\ref{thm:td-to-dnfres-tight}
and~\ref{thm:res-bound} follow the approach outlined in
Subsection~\ref{subsec:log-construction-overview}.
The constructions in Theorems~\ref{thm:td-to-dnfres-plog}
and~\ref{thm:res-bound-plog} follow the approach
outlined in Subsection~\ref{subsec:plog-construction-overview}. Finally, we use these results to prove
Theorem~\ref{thm:main-bounds}.

\begin{theorem}\label{thm:td-to-dnfres-tight}
Let $F$ be an unsatisfiable CNF formula of maximum clause width
$k$, with $n$ variables and $m$ clauses, and let $(T,\chi)$ be a
nice tree decomposition of $G^*(F)$ of width $w$ such that every
leaf and the root of $T$ have empty bags. Then one can construct from $(T,\chi)$ a
$k\mathrm{\text{-}mDNF\text{-}Res}$ refutation of $F$ of length
$|V(T)|2^{O(w)}$.

Every mDNF in the refutation, except for the initial clauses $c \in F$, has size at most $w+1$. Moreover,
for every mDNF $E$ in the refutation, there exist a node
$t_E\in V(T)$, a set $A_E\subseteq\chi_c(t_E)$, an mDNF $L_E$
consisting only of singleton terms, and, for each $c\in A_E$, a
possibly empty subclause $C_c\subseteq c$ such that
\[
E=L_E\mlor\bigmlor_{c\in A_E}\neg C_c.
\]
\end{theorem}

\begin{proof}
We process the nodes of $T$ in bottom-up order. If $t$ is a leaf,
then there is no inconsistent state at $t$ by
Lemma~\ref{lem:leaf-base}. Now let $t$ be a non-leaf node, and
suppose that, for every child $t_i$ of $t$, all inconsistent states
at $t_i$ have already been computed and
$D(t_i,\beta,S,\emptyset)$ has been derived for every such state
$(t_i,\beta,S)$.

Using Lemmas~\ref{lem:join-zero-recursion},
\ref{lem:introduce-var-zero},
\ref{lem:introduce-clause-zero},
\ref{lem:forget-var-zero}, and
\ref{lem:forget-clause-zero}, we can compute all inconsistent states
$(t,\alpha,A)$ of $t$. For each inconsistent state $(t,\alpha,A)$, we construct a
derivation of $D(t,\alpha,A,\emptyset)$. If $t$ is a join node,
we apply Lemma~\ref{lem:derive-join} with
$B_0=B_1=B_2=\emptyset$. If $t$ is a forget-clause node
forgetting $c$, we apply Lemma~\ref{lem:derive-forget-clause}
using the derivation from
$D(t',\alpha,A\cup\{c\},\emptyset)$. At introduce-variable,
introduce-clause, and forget-variable nodes, we apply
Lemmas~\ref{lem:derive-introduce-var},
\ref{lem:derive-introduce-clause}, and
\ref{lem:derive-forget-var} with $B=\emptyset$.
The mDNFs from child nodes used in these derivations
are all of the form $D(t_i,\beta,S,\emptyset)$, where
$(t_i,\beta,S)$ is inconsistent. Hence they have already been
derived by the induction hypothesis.

Let $r$ be the root of $T$. Since $\chi(r)=\emptyset$, the only
state at $r$ is $(r,\emptyset,\emptyset)$. This state is
inconsistent because $F$ is unsatisfiable. Hence we can
derive
\[
D(r,\emptyset,\emptyset,\emptyset)=\bot,
\]
and therefore get a $k\mathrm{\text{-}mDNF\text{-}Res}$
refutation of $F$.

At a node $t$, the number of states is at most
\[
2^{|\chi_v(t)|+|\chi_c(t)|}
=2^{|\chi(t)|}
\leq2^{w+1}.
\]
By the above derivation lemmas, the derivation of
$D(t,\alpha,A,\emptyset)$ has length $2^{O(w)}$ for every
inconsistent state $(t,\alpha,A)$. Therefore, the total
length of the refutation is
\[
|V(T)|\cdot 2^{w+1}\cdot 2^{O(w)}
=|V(T)|2^{O(w)}.
\]

We next verify that every mDNF in the refutation has the
form stated in the theorem. 
The mDNF representing an inconsistent state $(t,\alpha,A)$ is
\[
D(t,\alpha,A,\emptyset)
=
C_\alpha\mlor\bigmlor_{c\in A}\neg c_t^{\mathrm{int}}.
\]
For this mDNF, the form stated in the theorem is obtained by
taking $t_E=t$, $A_E=A$, $L_E=C_\alpha$, and
$C_c=c_t^{\mathrm{int}}$ for every $c\in A$. The initial clauses of $F$ and the axioms
$x\mlor\neg x$ have the stated form with $A_E=\emptyset$.
The axiom $\top$ can be viewed as the negation of the empty subclause of any
clause in $F$ and therefore also has the stated form.
All other intermediate mDNFs have the stated form by
Lemmas~\ref{lem:derive-join},
\ref{lem:derive-introduce-var}, and
\ref{lem:derive-introduce-clause}.

Since $B=\emptyset$ throughout the construction, by
Lemmas~\ref{lem:derive-join},
\ref{lem:derive-introduce-var},
\ref{lem:derive-introduce-clause},
\ref{lem:derive-forget-var}, and
\ref{lem:derive-forget-clause}, every mDNF in the refutation
has size at most $w+1$, except for the initial clauses of
$F$.
\end{proof}

\begin{theorem}\label{thm:res-bound}
Let $F$ be an unsatisfiable CNF formula of maximum clause width
$k$, and let $(T,\chi)$ be a nice tree decomposition of $G^*(F)$
of width $w$ and log-weighted width $w_{\log}$. Then $F$ has a resolution refutation of
length $k|V(T)|2^{O(w_{\log})}$ and width at most $w+k$.
\end{theorem}

\begin{proof}
Let $\pi_{\mathrm{mDNF}}$ be the
$k\mathrm{\text{-}mDNF\text{-}Res}$ refutation given by
Theorem~\ref{thm:td-to-dnfres-tight}. For an mDNF
$D=T_1\mlor\cdots\mlor T_s$, define the expansion of $D$
\[
\operatorname{Exp}(D):=
\left\{
\ell_1\mlor\cdots\mlor\ell_s:
\ell_i\in T_i\text{ for every }i\in\{1,\ldots,s\}
\right\}.
\]
$\operatorname{Exp}(D)$ is a set of multiclauses. We use the
conventions $\operatorname{Exp}(\bot)=\{\bot\}$ and
$\operatorname{Exp}(D)=\emptyset$ if $D$ contains the empty term.

We construct an $\mathrm{mRes}$ derivation
$\pi_{\mathrm{mRes}}$ by processing the mDNFs
$D \in \pi_{\mathrm{mDNF}}$ in order and deriving every multiclause in
$\operatorname{Exp}(D)$. If $c\in F$ is an initial clause, then
$\operatorname{Exp}(c)=\{c\}$. Moreover,
$\operatorname{Exp}(x\mlor\neg x)=\{x\mlor\neg x\}$, and $\operatorname{Exp}(\top)=\emptyset$.

Suppose that $D$ is neither an initial mDNF nor an axiom.
By the definition of a $k\mathrm{\text{-}mDNF\text{-}Res}$ derivation, $D$ is obtained from an
mDNF $\widetilde D$ by zero or more applications of the
contraction rule, where $\widetilde D$ is either an earlier
mDNF in $\pi_{\mathrm{mDNF}}$ or is derived from earlier mDNFs
by one application of a rule other than contraction. For one application of the contraction rule
$D_0\mlor T\mlor T\to D_0\mlor T$, every
$C'\in\operatorname{Exp}(D_0\mlor T)$ has the form
$C_0\mlor\ell$, where
$C_0\in\operatorname{Exp}(D_0)$ and $\ell\in T$, and
$C'$ can be obtained from $C_0\mlor\ell\mlor\ell\in
\operatorname{Exp}(D_0\mlor T\mlor T)$ by applications
of the contraction rule. Repeating this argument for each application of the contraction
rule, we obtain that, for every $C\in\operatorname{Exp}(D)$, there is a
multiclause $\widetilde C\in\operatorname{Exp}(\widetilde D)$
such that $C$ can be obtained from $\widetilde C$ by applications
of the contraction rule. 

Next, we show that $C\in\operatorname{Exp}(D)$ is derivable in at most $k$ steps.
It suffices to show that we can derive $\widetilde C$ in at most $k$ steps.
Indeed, if $\widetilde C$ has not yet been derived, the contractions
from $\widetilde C$ to $C$ can be included in the final step
of its derivation. Otherwise, $\widetilde C$ has already been derived,
and $C$ can be obtained from it in at most one step.

Suppose first that $\widetilde D=D_0\mlor D_1$ is obtained
from $D_0$ by the weakening rule, where $D_1$ is a $k$-mDNF.
Since $\widetilde C\in\operatorname{Exp}(D_0\mlor D_1)$,
let $\widetilde C=C_0\mlor C_1$, where
$C_i\in\operatorname{Exp}(D_i)$ for $i\in\{0,1\}$.
Thus $\widetilde C$ follows from $C_0$ by one application
of the $\mathrm{mRes}$ weakening rule.

Suppose that $\widetilde D=D_0\mlor T'$ is obtained from
$D_0\mlor T$ by the $\land$-elimination rule, where
$T'\subseteq T$. $\widetilde C$ has the form
$C_0\mlor\ell$, where
$C_0\in\operatorname{Exp}(D_0)$ and $\ell\in T'$. Since
$\ell\in T$, $\widetilde C$ already belongs to
$\operatorname{Exp}(D_0\mlor T)$.

Suppose next that
$\widetilde D=D_1\mlor D_2\mlor(T_1\land T_2)$ is obtained from
$D_1\mlor T_1$ and $D_2\mlor T_2$ by the
$\land$-introduction rule. Assume
$\widetilde C=C_1\mlor C_2\mlor\ell$, where
$C_i\in\operatorname{Exp}(D_i)$ for $i\in\{1,2\}$ and
$\ell\in T_1\cup T_2$. If $\ell\in T_1$, then
$C_1\mlor\ell\in\operatorname{Exp}(D_1\mlor T_1)$, and
$\widetilde C$ follows from this multiclause by the weakening rule. The case $\ell\in T_2$ is similar.

Finally, suppose that the cut rule gives
$\widetilde D=D_1\mlor D_2$ from
\[
D_1\mlor\ell_1\mlor\cdots\mlor\ell_q
\qquad\text{and}\qquad
D_2\mlor(\neg\ell_1\land\cdots\land\neg\ell_q).
\]
Let $\widetilde C=C_1\mlor C_2$, where
$C_i\in\operatorname{Exp}(D_i)$ for $i\in\{1,2\}$. The expansions
of the two mDNFs used here contain
$C_1\mlor\ell_1\mlor\cdots\mlor\ell_q$ and
$C_2\mlor\neg\ell_j$ for every $j\in\{1,\ldots,q\}$.
Let $C^{0}:=C_1\mlor\ell_1\mlor\cdots\mlor\ell_q$ and,
for $j\in\{1,\ldots,q\}$, let
\[
C^{j}:=C_1\mlor C_2\mlor
\ell_{j+1}\mlor\cdots\mlor\ell_q.
\]
For each $j\in\{1,\ldots,q\}$, $C^{j}$ is derivable from
$C^{j-1}$ and $C_2\mlor\neg\ell_j$ by one application of
the resolution rule and applications of the
contraction rule. Thus $\widetilde C=C^q$ has a derivation of length at most $q$
from $C^0$ and the multiclauses $C_2\mlor\neg\ell_j$,
$j\in\{1,\ldots,q\}$.
Since
$\neg\ell_1\land\cdots\land\neg\ell_q$ is a term of a
$k$-mDNF, $q\leq k$.

We next show that $|\operatorname{Exp}(D)|\leq 2^{w_{\log}+1}$
for every mDNF $D$ in $\pi_{\mathrm{mDNF}}$.
By Theorem~\ref{thm:td-to-dnfres-tight}, for every mDNF $D$ in
$\pi_{\mathrm{mDNF}}$, there are a node $t_D\in V(T)$, a set
$A_D\subseteq\chi_c(t_D)$, an mDNF $L_D$ consisting only of
singleton terms, and subclauses $C_c\subseteq c$ such that
\[
D=L_D\mlor\bigmlor_{c\in A_D}\neg C_c.
\]
Since every term in $L_D$ contains exactly one literal and
$C_c\subseteq c$ for every $c\in A_D$, we have
\[
\begin{aligned}
|\operatorname{Exp}(D)|
&\leq \prod_{c\in A_D}\max\{|C_c|,1\}\\
&\leq \prod_{c\in\chi_c(t_D)}|c|\\
&\leq
2^{\sum_{c\in\chi_c(t_D)}\max\{\log|c|,1\}}
\leq 2^{w_{\log}+1}.
\end{aligned}
\]
For each $C\in\operatorname{Exp}(D)$, at most $k$ multiclauses
are added to $\pi_{\mathrm{mRes}}$. Contraction does not
contribute to the length. Since
$|\pi_{\mathrm{mDNF}}|\leq |V(T)|2^{O(w)}$ and
$w\leq w_{\log}$,
\[
|\pi_{\mathrm{mRes}}|
\leq
k|\pi_{\mathrm{mDNF}}|2^{w_{\log}+1}
\leq
k|V(T)|2^{O(w_{\log})}.
\]

We next consider the width of $\pi_{\mathrm{mRes}}$.
For every $C\in\operatorname{Exp}(D)$, we have
$\operatorname{width}(C)=\operatorname{size}(D)$.
If $D$ is an initial clause of $F$, this width is at most $k$;
otherwise, it is at most $w+1$ by
Theorem~\ref{thm:td-to-dnfres-tight}.
If $D$ is derived by a rule other than the cut rule, then each
$C\in\operatorname{Exp}(D)$ is obtained in at most one step,
so deriving $C$ introduces no intermediate multiclauses.
It therefore remains to consider the intermediate multiclauses
in the derivations for the cut rule.

In the construction of $\pi_{\mathrm{mDNF}}$, the cut rule is
used only at forget-variable and forget-clause nodes. At a
forget-variable node, the two mDNFs used in the cut have the form
$D_0\mlor x$ and $D_0\mlor\neg x$. For each
$C_0\in\operatorname{Exp}(D_0)$, applying the resolution rule to
$C_0\mlor x$ and $C_0\mlor\neg x$ and then applying the contraction rule
gives $C_0$. All these multiclauses have width at most $w+1$.

At a forget-clause node, write the forgotten clause as
$c=\ell_1\lor\cdots\lor\ell_q$, where $q\leq k$. The cut is
applied to $c$ and $D_0\mlor\neg c$. For each
$C_0\in\operatorname{Exp}(D_0)$, the multiclause obtained after
the $j$th resolution and the subsequent contractions is
\[
C_0\mlor\ell_{j+1}\mlor\cdots\mlor\ell_q.
\]
Its width is at most
\[
\operatorname{width}(C_0)+q-j
\leq (w+1)+(k-1)
=w+k.
\]
Thus $\pi_{\mathrm{mRes}}$ has width at most $w+k$. Since
$\operatorname{Exp}(\bot)=\{\bot\}$, its final multiclause is
empty. By Lemma~\ref{lem:mres-to-res}, we have a resolution refutation
of $F$ of length at most $k|V(T)|2^{O(w_{\log})}$ and width at most $w+k$.
\end{proof}

\begin{theorem}\label{thm:td-to-dnfres-plog}
Let $F$ be an unsatisfiable CNF formula of maximum clause width
$k$, with $n$ variables and $m$ clauses, let $(T,\chi,r)$ be a
nice tree decomposition of $G^*(F)$ of width $w$, and let
$\mathcal P=(P_c)_{c\in F}$ be a clause-path family. Then one can
construct a
$k\mathrm{\text{-}mDNF\text{-}Res}$ refutation of $F$ of length
$|V(T)|2^{O(w)}$ from $(T,\chi,r)$ and $\mathcal P$.

Moreover, for every mDNF $E$ in the refutation other than
the initial clauses and axioms, there exist a node
$t_E\in V(T)$, an assignment
$\alpha_E:\chi_v(t_E)\to\{0,1\}$, a set
$A_E\subseteq\chi_c(t_E)$, an mDNF $L_E$ consisting only of
singleton terms, and, for each $c\in A_E$, a node $t_c$
that is either $t_E$ or a child of $t_E$ and a possibly
empty subclause $C_c\subseteq c_{t_c}^{\mathrm{int}}$ such that
$c\in\chi_c(t_c)$, $t_c\notin P_c$, and
\[
E=C_{\alpha_E}\mlor L_E
\mlor\bigmlor_{c\in A_E}\neg C_c.
\]
\end{theorem}

\begin{proof}
For every $t\in V(T)$ and $A\subseteq\chi_c(t)$, let
$B_t(A):=\{c\in A:t\in P_c\}$.
We process the nodes of $T$ in bottom-up order. If $t$ is a leaf,
then there is no inconsistent state at $t$ by
Lemma~\ref{lem:leaf-base}. Now let $t$ be a non-leaf node, and
suppose that, for every child $t_i$ of $t$, all inconsistent states
at $t_i$ have already been computed and
$D(t_i,\beta,S,B_{t_i}(S))$ has been derived for every inconsistent state
$(t_i,\beta,S)$.

Using Lemmas~\ref{lem:join-zero-recursion},
\ref{lem:introduce-var-zero},
\ref{lem:introduce-clause-zero},
\ref{lem:forget-var-zero}, and
\ref{lem:forget-clause-zero}, we can compute all inconsistent states
$(t,\alpha,A)$ at $t$. For each such state, we construct a
derivation of $D(t,\alpha,A,B_t(A))$.

Suppose first that $t$ is a join node with children $t_1,t_2$.
For every $c\in A$, both children belong to $T(c)$.
Since each nonempty $P_c$ is a path from the root of $T(c)$
to a leaf of $T(c)$, it contains $t$ if and only if it contains
exactly one of $t_1,t_2$. Thus $B_{t_1}(A)$ and $B_{t_2}(A)$
form a partition of $B_t(A)$.
We apply Lemma~\ref{lem:derive-join} with
$B_0=B_t(A)$ and $B_i=B_{t_i}(A)$ for $i\in\{1,2\}$.

Suppose next that $t$ has a unique child $t'$. For every clause
$d\in\chi_c(t)\cap\chi_c(t')$, the node $t'$ is the unique
child of $t$ in $T(d)$, so
$t\in P_d$ if and only if $t'\in P_d$.
Consequently, $S\cap B_t(A)=B_{t'}(S)$ for every
$S\subseteq A\cap\chi_c(t')$.
At introduce-variable, introduce-clause, and forget-variable
nodes, we apply
Lemmas~\ref{lem:derive-introduce-var},
\ref{lem:derive-introduce-clause}, and
\ref{lem:derive-forget-var} with $B=B_t(A)$.

Finally, suppose that $t$ is a forget-clause node forgetting $c$.
The node $t'$ is the root of $T(c)$, and hence
\[
B_{t'}(A\cup\{c\})=
\begin{cases}
B_t(A), & P_c=\emptyset,\\
B_t(A)\cup\{c\}, & P_c\neq\emptyset.
\end{cases}
\]
If $P_c=\emptyset$, Lemma~\ref{lem:derive-forget-clause}
gives a derivation of $D(t,\alpha,A,B_t(A))$ from
$D(t',\alpha,A\cup\{c\},B_t(A))$.
If $P_c\neq\emptyset$, Lemma~\ref{lem:derive-forget-clause} gives a derivation from
$D(t',\alpha,A\cup\{c\},B_t(A)\cup\{c\})$.

The mDNFs from child nodes used in these derivations are all of
the form $D(t_i,\beta,S,B_{t_i}(S))$, where
$(t_i,\beta,S)$ is inconsistent. Hence they have already been
derived by the induction hypothesis.

Since $\chi(r)=\emptyset$, the only state at the root is
$(r,\emptyset,\emptyset)$. This state is inconsistent because
$F$ is unsatisfiable. Hence we derive
$D(r,\emptyset,\emptyset,\emptyset)=\bot$, which gives a
$k\mathrm{\text{-}mDNF\text{-}Res}$ refutation of $F$.

At a node $t$, the number of states is at most
$2^{|\chi_v(t)|+|\chi_c(t)|}=2^{|\chi(t)|}\leq2^{w+1}$.
By Lemmas~\ref{lem:derive-join},
\ref{lem:derive-introduce-var},
\ref{lem:derive-introduce-clause},
\ref{lem:derive-forget-var}, and
\ref{lem:derive-forget-clause}, the derivation of
$D(t,\alpha,A,B_t(A))$ has length $2^{O(w)}$ for every
inconsistent state $(t,\alpha,A)$. Therefore, the total length
of the refutation is
\[
|V(T)|\cdot2^{w+1}\cdot2^{O(w)}
=|V(T)|2^{O(w)}.
\]

We next verify that every mDNF in the refutation other than
the initial clauses and axioms has the form stated in the
theorem. The mDNF representing an inconsistent state
$(t,\alpha,A)$ is
\[
D(t,\alpha,A,B_t(A))
=C_\alpha
\mlor\bigmlor_{c\in B_t(A)}c_t^{\mathrm{ext}}
\mlor\bigmlor_{c\in A\setminus B_t(A)}
\neg c_t^{\mathrm{int}}.
\]
For this mDNF, take $t_E=t$, $\alpha_E=\alpha$,
$A_E=A\setminus B_t(A)$, and
$L_E=\bigmlor\limits_{c\in B_t(A)}c_t^{\mathrm{ext}}$.
For every $c\in A_E$, take $t_c=t$ and
$C_c=c_t^{\mathrm{int}}$.
Since $c\notin B_t(A)$, we have $t_c\notin P_c$.
Then $D(t,\alpha,A,B_t(A))$ has the stated form.

It remains to consider the other intermediate mDNFs in the derivation.

First, let $t_0$ be a join node with children $t_1,t_2$.
By Lemma~\ref{lem:derive-join}, every other intermediate mDNF $E$
in the derivation of $D(t_0,\alpha,A,B_{t_0}(A))$
has the form
\[
E=C_\alpha\mlor L\mlor\bigmlor_{c\in A'}\neg C_c,
\]
where $L$ consists only of singleton terms, $A'\subseteq A$,
and, for each $c\in A'$, there is an index $i_c\in\{0,1,2\}$
such that $C_c\subseteq c_{t_{i_c}}^{\mathrm{int}}$ and
$c\notin B_{t_{i_c}}(A)$.
Take $t_E=t_0$, $\alpha_E=\alpha$, $A_E=A'$, $L_E=L$,
and $t_c=t_{i_c}$.
Then $t_c\notin P_c$, and, since
$c\in A\subseteq\chi_c(t_0)$ and $\chi(t_0)=\chi(t_1)=\chi(t_2)$,
$c\in\chi_c(t_c)$.
Thus $E$ has the form stated in the theorem.

Next, let $t$ be an introduce-variable node with child $t'$.
By Lemma~\ref{lem:derive-introduce-var}, every other intermediate mDNF $E$
in the derivation of $D(t,\alpha,A,B_t(A))$ has the form
\[
E=C_\alpha\mlor L\mlor\bigmlor_{c\in A'}\neg C_c,
\]
where $L$ consists only of singleton terms,
$A'\subseteq A\setminus B_t(A)$, and, for each $c\in A'$,
$C_c=c_{u_c}^{\mathrm{int}}$ for some $u_c\in\{t,t'\}$.
Take $t_E=t$, $\alpha_E=\alpha$, $A_E=A'$, $L_E=L$,
and $t_c=u_c$.
Since $A'\subseteq\chi_c(t)=\chi_c(t')$ and
$\chi_v(t')\subseteq\chi_v(t)$, we have
$c\in\chi_c(t_c)$ for every $c\in A_E$.
Moreover, $t\in P_c$ if and only if $t'\in P_c$, and
$c\notin B_t(A)$, so $t_c\notin P_c$.
Thus $E$ has the form stated in the theorem.

Next, let $t$ be an introduce-clause node.
By Lemma~\ref{lem:derive-introduce-clause}, every other intermediate mDNF $E$
in the derivation of $D(t,\alpha,A,B_t(A))$ has the form
\[
E=C_\alpha\mlor L\mlor\bigmlor_{c\in A'}\neg C_c,
\]
where $L$ consists only of singleton terms,
$A'\subseteq A\setminus B_t(A)$, and
$C_c\subseteq c_t^{\mathrm{int}}$ for every $c\in A'$.
Take $t_E=t$, $\alpha_E=\alpha$, $A_E=A'$, $L_E=L$,
and $t_c=t$ for every $c\in A_E$.
Then $c\in\chi_c(t_c)$ and $t_c\notin P_c$, since
$c\in A\setminus B_t(A)$.
Thus $E$ has the form stated in the theorem.

The derivations at forget-variable and forget-clause nodes contain no other intermediate
mDNFs. This completes the proof.
\end{proof}

\begin{theorem}\label{thm:res-bound-plog}
Let $F$ be an unsatisfiable CNF formula of maximum clause width
$k$, let $(T,\chi,r)$ be a nice tree decomposition of $G^*(F)$,
and let $\mathcal P=(P_c)_{c\in F}$ be a clause-path family.
Let $w_{\mathrm{plog}}:=w_{\mathrm{plog}}(T,\chi,r,\mathcal P)$.
Then $F$ has a resolution refutation of length
$k|V(T)|2^{O(w_{\mathrm{plog}})}$.
\end{theorem}

\begin{proof}
Let $w:=w(T,\chi)$, and let $\pi_{\mathrm{mDNF}}$ be the
$k\mathrm{\text{-}mDNF\text{-}Res}$ refutation given by
Theorem~\ref{thm:td-to-dnfres-plog}.
For every mDNF $D$ in this refutation other than the initial
clauses and axioms, fix an expression
\[
D=C_{\alpha_D}\mlor L_D
\mlor\bigmlor_{c\in A_D}\neg C_c
\]
satisfying the conditions in
Theorem~\ref{thm:td-to-dnfres-plog}. Let
\[
A_D^0:=
\{c\in A_D:\ell\notin C_{\alpha_D}
\text{ for every }\ell\in\neg C_c\}.
\]
Define
\[
\begin{aligned}
\operatorname{Exp}_{\mathrm{red}}(D):=
\bigl\{E={}&C_{\alpha_D}\mlor L_D
\mlor\bigmlor_{c\in A'}\ell_c :
A_D^0\subseteq A'\subseteq A_D,\\
&\ell_c\in\neg C_c,\quad
\ell_c,\neg\ell_c\notin C_{\alpha_D}
\text{ for every }c\in A'\bigr\}.
\end{aligned}
\]
For every initial clause $c\in F$, define
$\operatorname{Exp}_{\mathrm{red}}(c):=\{c\}$.
For the axioms, define
$\operatorname{Exp}_{\mathrm{red}}(x\mlor\neg x):=\{x\mlor\neg x\}$
and $\operatorname{Exp}_{\mathrm{red}}(\top):=\emptyset$.

We first show that
$|\operatorname{Exp}_{\mathrm{red}}(D)|
\leq2^{3(w_{\mathrm{plog}}+1)}$.
For an initial clause or an axiom, $|\operatorname{Exp}_{\mathrm{red}}(D)|
\leq 1$.
For every other $D$, by Theorem~\ref{thm:td-to-dnfres-plog},
$D$ has the following form:
\[
D=C_{\alpha_D}\mlor L_D
\mlor\bigmlor_{c\in A_D}\neg C_c,
\]
where $t_D\in V(T)$,
$\alpha_D:\chi_v(t_D)\to\{0,1\}$,
$A_D\subseteq\chi_c(t_D)$, and $L_D$ consists only of
singleton terms.
For each $c\in A_D$, there is a node $t_c$ that is either
$t_D$ or a child of $t_D$ such that
$c\in\chi_c(t_c)$, $t_c\notin P_c$, and
$C_c\subseteq c_{t_c}^{\mathrm{int}}$.

\par\medskip
\noindent
\textbf{Claim.}
For every $c\in A_D$, if $x\in\operatorname{var}(C_c)$ and
$x\in\chi_v(u)$ for some $u\in P_c$, then
$x\in\chi_v(t_D)$.

\par\medskip
\noindent
\emph{Proof of the claim.}
Let $x\in\operatorname{var}(C_c)$ and $u\in P_c$ with
$x\in\chi_v(u)$.
Since $C_c\subseteq c_{t_c}^{\mathrm{int}}$, we have
$x\in\mathrm{Var}_{t_c}$, so there is a node
$u'\in V(T_{t_c})$ such that $x\in\chi_v(u')$.
Since $c\in\chi_c(t_c)$, the root of $T(c)$ is an ancestor
of or equal to $t_c$.
If $u\in V(T_{t_c})$, the path from the root of $T(c)$
to $u$ would contain $t_c$. This path is contained in $P_c$,
contrary to $t_c\notin P_c$.
Thus $u\notin V(T_{t_c})$.
Since $u'\in V(T_{t_c})$ and $t_c$ is either $t_D$
or a child of $t_D$, the path in $T$ from $u'$ to $u$
contains $t_D$.
Since $x\in\chi_v(u')\cap\chi_v(u)$ and the nodes whose
bags contain $x$ form a connected subtree of $T$,
we have $x\in\chi_v(t_D)$. This proves the claim.

\par\medskip
\noindent

Let $E\in\operatorname{Exp}_{\mathrm{red}}(D)$. By definition,
\[
E=C_{\alpha_D}\mlor L_D
\mlor\bigmlor_{c\in A'}\ell_c
\]
for some $A_D^0\subseteq A'\subseteq A_D$ and literals
$\ell_c\in\neg C_c$ satisfying
$\ell_c,\neg\ell_c\notin C_{\alpha_D}$ for every $c\in A'$.
Since $C_{\alpha_D}$ contains one literal for every variable
in $\chi_v(t_D)$, we have
$\operatorname{var}(\ell_c)\notin\chi_v(t_D)$.
By the claim,
$\operatorname{var}(\ell_c)\notin\chi_v(u)$ for every $u\in P_c$.
Since $\ell_c\in\neg C_c$ and $C_c\subseteq c$, the literal
$\neg\ell_c$ belongs to $c$, and its variable belongs to
no bag on $P_c$. There are $\rho_{\mathcal P}(c)$ such
literals in $c$, so there are at most $\rho_{\mathcal P}(c)$
possible values of $\ell_c$.

The multiclause $E$ is determined by $A'$ and the literals
$\ell_c$ for $c\in A'$.
For each $c\in A_D$, there are at most $\rho_{\mathcal P}(c)$
possible values of $\ell_c$ when $c\in A'$, together with
the possibility that $c\notin A'$.
Hence
$|\operatorname{Exp}_{\mathrm{red}}(D)|
\leq\prod_{c\in A_D}(\rho_{\mathcal P}(c)+1)$.

Let $U_D$ be the set consisting of $t_D$ and its children.
Since $t_c\notin P_c$, the definition of
$\omega_{\mathcal P}(t_c,c)$ gives
$\rho_{\mathcal P}(c)+1\leq
2^{\omega_{\mathcal P}(t_c,c)}$.
For every $c\in A_D$, we also have $t_c\in U_D$ and
$c\in\chi_c(t_c)$, so
\[
\begin{aligned}
|\operatorname{Exp}_{\mathrm{red}}(D)|
&\leq 2^{\sum_{c\in A_D}\omega_{\mathcal P}(t_c,c)}\\
&\leq 2^{\sum_{u\in U_D}
          \sum_{c\in\chi_c(u)}\omega_{\mathcal P}(u,c)}\\
&\leq 2^{|U_D|(w_{\mathrm{plog}}+1)}
\leq 2^{3(w_{\mathrm{plog}}+1)},
\end{aligned}
\]
where the last two inequalities follow from the definition of
$w_{\mathrm{plog}}$ and the fact that $t_D$ has at most two
children.

We then show that every multiclause in
$\operatorname{Exp}_{\mathrm{red}}(D)$ is derivable in at most
one step from a multiclause in $\operatorname{Exp}(D)$, and
that every non-tautological multiclause in
$\operatorname{Exp}(D)$ is derivable in at most one step from
a multiclause in $\operatorname{Exp}_{\mathrm{red}}(D)$.

If $D$ is an initial clause or an axiom, then
$\operatorname{Exp}_{\mathrm{red}}(D)=\operatorname{Exp}(D)$,
so no rule application is needed.

Now suppose that $D$ is neither an initial clause nor an axiom.
Let $C\in\operatorname{Exp}_{\mathrm{red}}(D)$ have the form
$C=C_{\alpha_D}\mlor L_D\mlor\bigmlor\limits_{c\in A'}\ell_c$
in the definition above.
For every $c\in A_D\setminus A'$, choose a literal $\ell_c$ such that 
$\ell_c\in\neg C_c$ and $\ell_c\in C_{\alpha_D}$; such a literal exists
by the condition on $A'$.
Then
\[
\widetilde C:=
C\mlor\bigmlor_{c\in A_D\setminus A'}\ell_c
\in\operatorname{Exp}(D).
\]
Since each additional $\ell_c$ belongs to $C_{\alpha_D}$,
$C$ is obtained from $\widetilde C$ in at most one step by
applications of the contraction rule.

Conversely, let $\widetilde C\in\operatorname{Exp}(D)$.
If $\widetilde C$ is tautological, then it contains $x$ and
$\neg x$ for some variable $x$, and is therefore derivable
from the axiom $x\mlor\neg x$ by at most one application of
the weakening rule.

Suppose now that $\widetilde C$ is non-tautological, and let
$\widetilde C=C_{\alpha_D}\mlor L_D
\mlor\bigmlor\limits_{c\in A_D}\ell_c$,
where $\ell_c\in\neg C_c$ for every $c\in A_D$.
Let $A':=\{c\in A_D:\ell_c\notin C_{\alpha_D}\}$ and
$C:=C_{\alpha_D}\mlor L_D\mlor\bigmlor\limits_{c\in A'}\ell_c$.
For every $c\in A_D\setminus A'$, we have
$\ell_c\in\neg C_c\cap C_{\alpha_D}$.
For every $c\in A'$, we also have
$\neg\ell_c\notin C_{\alpha_D}$, since $\widetilde C$
is non-tautological.
Thus $C\in\operatorname{Exp}_{\mathrm{red}}(D)$, and
$\widetilde C=C\mlor\bigmlor\limits_{c\in A_D\setminus A'}\ell_c$
is derivable from $C$ by at most one application of the
weakening rule.

If $D$ is an initial clause, $\operatorname{Exp}_{\mathrm{red}}(D)$ consists
of that initial clause. The reduced expansions of the axioms
are either empty or consist of an axiom.

Next, we construct an $\mathrm{mRes}$ derivation by deriving every
multiclause in $\operatorname{Exp}_{\mathrm{red}}(D)$ for each
$D$ in $\pi_{\mathrm{mDNF}}$, in the order of $\pi_{\mathrm{mDNF}}$.

Now suppose that $D$ is neither an initial clause nor an axiom,
and that every multiclause in
$\operatorname{Exp}_{\mathrm{red}}(D')$ has already been
derived for every earlier mDNF $D'$ in the derivation $\pi_{\mathrm{mDNF}}$.

Fix $C\in\operatorname{Exp}_{\mathrm{red}}(D)$, and let $\widetilde C\in\operatorname{Exp}(D)$ be a
multiclause from which $C$ is derivable in at most one step.
The construction in the proof of Theorem~\ref{thm:res-bound}
shows that $\widetilde C$ can be derived in at most $k$ steps
from at most $k+1$ multiclauses, each belonging to
$\operatorname{Exp}(D')$ for some earlier mDNF $D'$
in $\pi_{\mathrm{mDNF}}$.

Every non-tautological multiclause in $\operatorname{Exp}(D')$
is then derivable in at most one step from a multiclause in
$\operatorname{Exp}_{\mathrm{red}}(D')$.
Every tautological multiclause in $\operatorname{Exp}(D')$
is derivable using an axiom and at most one application of
the weakening rule.
Thus any multiclause in $\operatorname{Exp}(D')$ can be derived from $\operatorname{Exp}_{\mathrm{red}}(D')$
in at most two additional steps.

So we can derive $C\in\operatorname{Exp}_{\mathrm{red}}(D)$ in the following way: we first derive the at most $k+1$ multiclauses needed to
derive $\widetilde C$, using at most two additional steps
for each. We then derive $\widetilde C$ in at most $k$ steps
and obtain $C$ from $\widetilde C$ in at most one further step.
Thus deriving each $C\in\operatorname{Exp}_{\mathrm{red}}(D)$
requires $O(k)$ steps.

Consequently, the resulting $\mathrm{mRes}$ derivation has length
\[
O(k)\,|\pi_{\mathrm{mDNF}}|\,
2^{3(w_{\mathrm{plog}}+1)}
\leq k|V(T)|2^{O(w_{\mathrm{plog}})},
\]
where we use
$|\pi_{\mathrm{mDNF}}|\leq|V(T)|2^{O(w)}$ and
$w\leq w_{\mathrm{plog}}$.
Since $\operatorname{Exp}_{\mathrm{red}}(\bot)=\{\bot\}$,
this derivation is a refutation.
By Lemma~\ref{lem:mres-to-res}, we have a resolution
refutation of $F$ of length
$k|V(T)|2^{O(w_{\mathrm{plog}})}$.
\end{proof}

\begin{remark} \label{rem:one-sided-regular}
If the nice tree decomposition and clause-path family used
in Theorem~\ref{thm:res-bound-plog} are obtained from a
one-sided tree decomposition by the construction in
Theorem~\ref{thm:plog-one-sided}, then the resulting
resolution refutation is regular.

Indeed, in the refutation constructed in
Theorem~\ref{thm:td-to-dnfres-plog}, the cut rule is used
only at join, forget-clause, and forget-variable nodes.
The construction in Theorem~\ref{thm:plog-one-sided}
ensures that every variable of each clause $c\in F$
appears in a bag on $P_c$.
By the connectedness condition of tree decompositions,
$\operatorname{var}(c_t^{\mathrm{int}})
\subseteq\chi_v(t)$ whenever $c\in\chi_c(t)$ and
$t\notin P_c$: every variable of $c_t^{\mathrm{int}}$
appears in a bag on $P_c$ and a bag in $T_t$, so it belongs to
$\chi_v(t)$.
Consequently, the subclause $H$ in the join construction
of Lemma~\ref{lem:derive-join} is empty, so no cut rule is
used there. At a forget-clause node forgetting $c$,
the construction uses
$D(t',\alpha,A\cup\{c\},B\cup\{c\})$,
which already equals $D(t,\alpha,A,B)$.

Thus the cut rule is used only at forget-variable nodes, and each
such local derivation uses one cut on the forgotten variable.
The conversion in Theorem~\ref{thm:res-bound-plog}
therefore uses resolution only at these nodes, with at
most one resolution step on each directed path in
a local derivation. Since each variable is forgotten
at a unique node, the resulting refutation is regular.
\end{remark}

Now, we can prove Theorem~\ref{thm:main-bounds}.

\MainBounds*

\begin{proof}
If $k=1$, then $F$ contains the unit clauses $x$ and $\neg x$
for some variable $x$.
Applying the resolution rule to these clauses gives the
empty clause.
All four assertions follow. Assume that $k\geq2$.

By Lemma~\ref{lem:nice-empty-incidence}, $G^*(F)$ has a nice
tree decomposition $(T,\chi,r)$ of width $\mathrm{tw}^*(F)$
with $O((n+m)(\mathrm{tw}^*(F)+1))$ nodes.
By Theorem~\ref{thm:td-to-dnfres-tight}, we can construct a
$k\mathrm{\text{-}mDNF\text{-}Res}$ refutation of length
$(n+m)2^{O(\mathrm{tw}^*(F))}$.
By Lemma~\ref{lem-mdnfres-to-dnfres}, we can convert this
refutation into a $k\mathrm{\text{-}DNF\text{-}Res}$
refutation of the same asymptotic length.
This proves (1).

For (2), by Lemma~\ref{lem:nice-empty-log-incidence},
$G^*(F)$ has a nice tree decomposition
$(T_{\log},\chi_{\log},r_{\log})$ of log-weighted width
$\mathrm{tw}_{\log}^*(F)$ with
$O((n+m)(\mathrm{tw}_{\log}^*(F)+1))$ nodes.
By Lemma~\ref{lem:width-comparison}, its width is at most $\mathrm{tw}_{\log}^*(F)$.
By Theorem~\ref{thm:res-bound}, we can construct a resolution
refutation of length $k(n+m)2^{O(\mathrm{tw}_{\log}^*(F))}$
and width at most $\mathrm{tw}_{\log}^*(F)+k$.
This proves (2).

For (3), use the decomposition $(T,\chi,r)$ chosen for (1).
By Lemma~\ref{lem:width-comparison},
$w_{\log}(T,\chi)+1\leq(\mathrm{tw}^*(F)+1)\log k$.
By Theorem~\ref{thm:res-bound}, we can construct a resolution
refutation of width at most $\mathrm{tw}^*(F)+k$ and length
\[
\begin{aligned}
k|V(T)|2^{O(w_{\log}(T,\chi))}
&\leq k(n+m)(\mathrm{tw}^*(F)+1)
  2^{O((\mathrm{tw}^*(F)+1)\log k)}\\
&=(n+m)k^{O(\mathrm{tw}^*(F))}.
\end{aligned}
\]
This proves (3).

Finally, by the definition of
$\mathrm{tw}_{\mathrm{plog}}^*(F)$ and
Lemma~\ref{lem:plog-decomposition-size}, there are a nice
tree decomposition
$(T_{\mathrm{plog}},\chi_{\mathrm{plog}},r_{\mathrm{plog}})$
of $G^*(F)$ and a clause-path family $\mathcal P$
whose partially log-weighted width is
$\mathrm{tw}_{\mathrm{plog}}^*(F)$, with
$|V(T_{\mathrm{plog}})|
=O(k(n+m)(\mathrm{tw}_{\mathrm{plog}}^*(F)+1))$.
By Theorem~\ref{thm:res-bound-plog}, we can construct a
resolution refutation of length
\[
k|V(T_{\mathrm{plog}})|
2^{O(\mathrm{tw}_{\mathrm{plog}}^*(F))}
\leq k^2(n+m)2^{O(\mathrm{tw}_{\mathrm{plog}}^*(F))},
\]
which proves (4).
\end{proof}

\section{Regular resolution upper bounds}\label{sec:regular-resolution}
Let $F$ be an unsatisfiable CNF formula, and let $(T,\chi,r)$ be a
nice tree decomposition of $G^*(F)$ whose leaves and root have empty
bags. Let $\operatorname{var}(F)=\{x_1,\ldots,x_n\}$.

For each clause $c\in F$, let $q_c:=|c|$ and fix a representation
of $c$ as $\ell_1^c\lor\cdots\lor\ell_{q_c}^c$. Let
$c'\subseteq c$ be a subclause of $c$, and suppose that
$c'=\ell_{i_1}^c\lor\cdots\lor\ell_{i_s}^c$, in which
$1\leq i_1<\cdots<i_s\leq q_c$, and define
\[
\begin{aligned}
\mathcal C_1(c')
&:=
\left\{
\neg\ell_{i_h}^c
\mlor\bigmlor_{j=h+1}^{s}\ell_{i_j}^c:
1\leq h\leq s
\right\},\\
\mathcal C_0(c')
&:=
\left\{
\bigmlor_{j=h}^{s}\ell_{i_j}^c:
1\leq h\leq s+1
\right\}.
\end{aligned}
\]
Here $\bigmlor\limits_{j=s+1}^{s}\ell_{i_j}^c$ denotes the
multiclause containing no literals. For $1\leq h\leq s$, the
literal $\ell_{i_h}^c\in c'$ is called the \emph{leading literal}
of the multiclause $\neg\ell_{i_h}^c \mlor \bigmlor\limits_{j=h+1}^{s}\ell_{i_j}^c$.

For $A\subseteq\chi_c(t)$, let $\mathbf 1_A(c)=1$ if $c\in A$
and $\mathbf 1_A(c)=0$ otherwise. For an inconsistent state
$(t,\alpha,A)$, define
\[
\mathcal E(t,\alpha,A):=
\left\{
C_\alpha\mlor
\bigmlor_{c\in\chi_c(t)}C_c
\;\middle|\;
C_c\in
\mathcal C_{\mathbf 1_A(c)}(c_t^{\mathrm{int}})
\text{ for every }c\in\chi_c(t)
\right\}.
\]

For the motivation behind the definition of
$\mathcal E(t,\alpha,A)$, see
Subsection~\ref{subsec:regular-construction-overview}.

\begin{lemma}\label{lem:td-to-mres}
Let $F$ be an unsatisfiable CNF formula, and let $(T,\chi)$
be a nice tree decomposition of $G^*(F)$ of log-weighted width
$w_{\log}$ such that every leaf and the root of $T$ have empty
bags. Then one can construct from $(T,\chi)$ an $\mathrm{mRes}$
refutation $\pi$ of $F$ of length $|V(T)|2^{O(w_{\log})}$
such that, on every directed path in the proof DAG of $\pi$,
each variable is used in at most one variable-eliminating
resolution.
\end{lemma}

\begin{proof}
We construct an $\mathrm{mRes}$ derivation containing every
multiclause in $\mathcal E(t,\alpha,A)$ for every inconsistent
state $(t,\alpha,A)$, and prove by induction that, for every
$M\in\mathcal E(t,\alpha,A)$, the following three induction properties hold:
\begin{enumerate}
\item every multiclause in the derivation of $M$ contains only
variables from $\mathrm{Var}_t$;
\item on every directed path in the proof DAG ending at $M$,
each variable is resolved in at most one variable-eliminating
resolution;
\item on every directed path in the proof DAG ending at $M$,
no variable in $\operatorname{var}(M)$ is resolved in a
variable-eliminating resolution.
\end{enumerate}

The construction proceeds from the leaves to the root.
By Lemma~\ref{lem:leaf-base}, no leaf has an inconsistent state.
Suppose that, for every child $t_i$ of a node $t$ and every
inconsistent state $(t_i,\beta,S)$, all multiclauses in
$\mathcal E(t_i,\beta,S)$ have been derived and satisfy all three induction properties. Fix an inconsistent state $(t,\alpha,A)$
and a multiclause $M\in\mathcal E(t,\alpha,A)$.

\emph{Join node.}
Suppose that $t$ has children $t_1,t_2$.
Fix $M\in\mathcal E(t,\alpha,A)$, and suppose that
\[
M=C_\alpha\mlor
\bigmlor_{c\in\chi_c(t)}C_c.
\]
For each $c\in A$, let $a_c$ be the leading literal of $C_c$.
For $i\in\{1,2\}$, let
$A_i:=\{c\in A:a_c\in c_{t_i}^{\mathrm{int}}\}$.
By Lemma~\ref{lem:join-var-structure}(1),
$c_t^{\mathrm{int}}
=c_{t_1}^{\mathrm{int}}\lor c_{t_2}^{\mathrm{int}}$
for every $c\in\chi_c(t)$, so $A_1\cup A_2=A$.
By Lemma~\ref{lem:join-zero-recursion}, at least one of
$(t_1,\alpha,A_1)$ and $(t_2,\alpha,A_2)$ is inconsistent.
Fix $i\in\{1,2\}$ such that $(t_i,\alpha,A_i)$ is inconsistent.

For each $c\in\chi_c(t)$, let $C'_c$ be obtained from $C_c$
by deleting all literals whose variables do not belong to
$\mathrm{Var}_{t_i}$.
We verify that
$C'_c\in\mathcal C_{\mathbf 1_{A_i}(c)}
(c_{t_i}^{\mathrm{int}})$.

If $c\in A_i$, then $a_c\in c_{t_i}^{\mathrm{int}}$.
By the definition of $\mathcal C_1$, the multiclause $C'_c$
belongs to $\mathcal C_1(c_{t_i}^{\mathrm{int}})$ and has
leading literal $a_c$.

If $c\in A\setminus A_i$, let $a_c=\ell_p^c$.
Since $a_c\in c$ and $a_c\notin c_{t_i}^{\mathrm{int}}$,
we have $\operatorname{var}(a_c)\notin\mathrm{Var}_{t_i}$.
Thus $\neg a_c$ is deleted from $C_c$. The multiclause $C'_c$
therefore consists of the literals
$\ell_j^c\in c_{t_i}^{\mathrm{int}}$ with $j>p$, and hence
belongs to $\mathcal C_0(c_{t_i}^{\mathrm{int}})$.

If $c\in\chi_c(t)\setminus A$, then
$C_c\in\mathcal C_0(c_t^{\mathrm{int}})$.
It is easy to check from the definition of $\mathcal C_0$
that $C'_c\in\mathcal C_0(c_{t_i}^{\mathrm{int}})$.

Since $\chi(t_i)=\chi(t)$, we obtain
\[
M':=C_\alpha\mlor
\bigmlor_{c\in\chi_c(t)}C'_c
\in\mathcal E(t_i,\alpha,A_i).
\]
Thus $M'$ has already been derived. Each $C'_c$ is obtained
from $C_c$ by deleting literals, so either $M=M'$ or $M$
is derived from $M'$ by the weakening rule.

We verify the three induction properties.
By definition, $\operatorname{var}(M)\subseteq\mathrm{Var}_t$.
Since $\mathrm{Var}_{t_i}\subseteq\mathrm{Var}_t$, the first
induction property for $M'$ gives the first induction
property for $M$.
The second induction property also holds, since either
$M=M'$ or $M$ is derived from $M'$ by one weakening step.

For the third induction property, each $C'_c$ contains
exactly the literals of $C_c$ whose variables belong to
$\mathrm{Var}_{t_i}$. Since
$\operatorname{var}(C_\alpha)=\chi_v(t)\subseteq
\mathrm{Var}_{t_i}$, we have
$\operatorname{var}(M')=
\operatorname{var}(M)\cap\mathrm{Var}_{t_i}$.
Thus every variable in
$\operatorname{var}(M)\setminus\operatorname{var}(M')$
does not belong to $\mathrm{Var}_{t_i}$.
By the first induction property for $M'$, none of these
variables appears in its derivation, so none has been resolved.
By the third induction property for $M'$, no variable in
$\operatorname{var}(M')$ is resolved in a variable-eliminating
resolution on any directed path in the proof DAG ending at $M'$.
Since deriving $M$ from $M'$ requires no resolution,
the third induction property also holds for $M$.

\smallskip
\noindent
\emph{Introduce-variable node.}
Suppose that $t$ introduces a variable $x$ and has child $t'$.
Let $\alpha':=\alpha|_{\chi_v(t')}$. Define $\ell:=x$ if
$\alpha(x)=1$ and $\ell:=\neg x$ otherwise, and let
$A_x:=\{c\in A:\ell\notin c\}$.
By Lemma~\ref{lem:introduce-var-zero},
$(t',\alpha',A_x)$ is inconsistent.

Fix $M\in\mathcal E(t,\alpha,A)$, and suppose that
\[
M=C_\alpha\mlor
\bigmlor_{c\in\chi_c(t)}C_c.
\]
For each $c\in A$, let $a_c$ be the leading literal of $C_c$, and
define $B:=\{c\in A:\operatorname{var}(a_c)\neq x\}$.

If $A_x\not\subseteq B$, choose $c\in A_x\setminus B$. Since
$c\notin B$, we have $\operatorname{var}(a_c)=x$. Since
$c\in A_x$, we have $\ell\notin c$, and hence $a_c=\neg\ell$.
Thus $C_c$ contains $\ell$, whereas $C_\alpha$ contains
$\neg\ell$. Therefore, $M$ can be derived from the axiom
$\ell\mlor\neg\ell$ using the weakening rule.

We may therefore assume that $A_x\subseteq B$. It follows
immediately from the definition of $N$ that
$N(t',\alpha',B)\subseteq N(t',\alpha',A_x)=\emptyset$.
Hence $(t',\alpha',B)$ is inconsistent. For each
$c\in\chi_c(t)$, we choose a multiclause $C'_c$ such that
$C_{\alpha'}\mlor\bigmlor\limits_{c\in\chi_c(t)}C'_c
\in\mathcal E(t',\alpha',B)$ and, for every
$c\in\chi_c(t)$, either $C_c=C'_c$ or $C_c$ can be derived from
$C'_c$ using the weakening rule to add $\ell$ or $\neg\ell$.
Since $C_\alpha=C_{\alpha'}\mlor\neg\ell$, the multiclause $M$
can be derived from
$C_{\alpha'}\mlor\bigmlor\limits_{c\in\chi_c(t)}C'_c$
using the weakening rule.

We first consider the case $c\in A$. If $c$ contains neither $x$
nor $\neg x$, then
$c_t^{\mathrm{int}}=c_{t'}^{\mathrm{int}}$ and hence
$\mathcal C_1(c_t^{\mathrm{int}})
=\mathcal C_1(c_{t'}^{\mathrm{int}})$, so set $C'_c:=C_c$.

Suppose next that $\neg\ell\in c$, and let
$a_c=\ell_i^c$ and $\neg\ell=\ell_j^c$. Since $\ell\notin c$, we have $c\in A_x\subseteq B$. If $i=j$,
then $a_c=\neg\ell$, and hence $\operatorname{var}(a_c)=x$.
By the definition of $B$, this implies $c\notin B$, a
contradiction. Therefore $i\neq j$. If $j<i$, then $C_c$ is
also a member of $\mathcal C_1(c_{t'}^{\mathrm{int}})$ with
leading literal $a_c$, so set $C'_c:=C_c$. If $i<j$, let $C'_c$
be the member of $\mathcal C_1(c_{t'}^{\mathrm{int}})$ with
leading literal $a_c$. Then $C_c=C'_c\mlor\neg\ell$.

Suppose now that $\ell\in c$. Let
$a_c=\ell_i^c$ and $\ell=\ell_j^c$. If $i=j$, then $c\notin B$. Let $C'_c$ be the member of
$\mathcal C_0(c_{t'}^{\mathrm{int}})$ consisting of the literals
$\ell_h^c\in c_{t'}^{\mathrm{int}}$ with $h>j$. Then
$C_c=\neg\ell\mlor C'_c$.
If $j<i$, then $c\in B$ and
$C_c\in\mathcal C_1(c_{t'}^{\mathrm{int}})$, so set
$C'_c:=C_c$.
If $i<j$, then $c\in B$. Let $C'_c$ be the member of
$\mathcal C_1(c_{t'}^{\mathrm{int}})$ with leading literal $a_c$.
Then $C_c=C'_c\mlor\ell$.

Finally, suppose that $c\in\chi_c(t)\setminus A$. Then
$c\notin B$. Since
$c_{t'}^{\mathrm{int}}
=c_t^{\mathrm{int}}\setminus\{\ell,\neg\ell\}$, it is easy to
check from the definition of $\mathcal C_0$ that there is a
$C'_c\in\mathcal C_0(c_{t'}^{\mathrm{int}})$ such that either
$C_c=C'_c$ or $C_c$ can be derived from $C'_c$ using the
weakening rule to add $\ell$ or $\neg\ell$.

We verify the three induction properties.
By definition, $\operatorname{var}(M)\subseteq\mathrm{Var}_t$.
If $M$ is derived from the axiom $\ell\mlor\neg\ell$ by
weakening, all three induction properties hold, since
$x\in\mathrm{Var}_t$ and the derivation contains no resolution.

Otherwise, let
$M':=C_{\alpha'}\mlor
\bigmlor\limits_{c\in\chi_c(t)}C'_c$.
By construction, $M$ is derived from $M'$ by weakening,
and every added literal is either $x$ or $\neg x$.
The first induction property follows from
$\mathrm{Var}_{t'}\subseteq\mathrm{Var}_t$ and
$\operatorname{var}(M)\subseteq\mathrm{Var}_t$.
The second induction property follows because every directed
path in the proof DAG ending at $M$ extends a path ending
at $M'$ by one weakening step.

For the third induction property,
$\operatorname{var}(M)\setminus\operatorname{var}(M')
\subseteq\{x\}$.
Since $x\notin\mathrm{Var}_{t'}$, the first induction property
for $M'$ implies that its derivation contains neither $x$
nor $\neg x$, so $x$ has not been resolved.
Together with the third induction property for $M'$, this
proves the third induction property for $M$.

\smallskip
\noindent
\emph{Introduce-clause node.}
Suppose that $t$ introduces a clause $c$ and has child $t'$.
Fix $M\in\mathcal E(t,\alpha,A)$, and suppose that
\[
M=C_\alpha\mlor
\bigmlor_{d\in\chi_c(t)}C_d.
\]
Since $t$ introduces only $c$, we have
$\chi_v(t)=\chi_v(t')$ and $\mathrm{Var}_t=\mathrm{Var}_{t'}$.
Consequently, $d_t^{\mathrm{int}}=d_{t'}^{\mathrm{int}}$ for every
$d\in\chi_c(t')$.

Suppose first that $c\notin A$. By
Lemma~\ref{lem:introduce-clause-zero}(a), the state
$(t',\alpha,A)$ is inconsistent. The multiclause
$M':=C_\alpha\mlor
\bigmlor\limits_{d\in\chi_c(t')}C_d$
belongs to $\mathcal E(t',\alpha,A)$. Since
$M=M'\mlor C_c$, $M$ is derived from $M'$ by the
weakening rule.

Suppose next that $c\in A$ and $\alpha$ satisfies $c$. By
Lemma~\ref{lem:introduce-clause-zero}(b), the state
$(t',\alpha,A\setminus\{c\})$ is inconsistent. Again,
$M':=C_\alpha\mlor
\bigmlor\limits_{d\in\chi_c(t')}C_d$
belongs to $\mathcal E(t',\alpha,A\setminus\{c\})$, and $M$ is derived from $M'$ by the weakening rule.

In both cases, the weakening rule adds only literals belonging to
$C_c$. By Lemma~\ref{lem:introduce-clause-local},
$\operatorname{var}(C_c)\subseteq\chi_v(t)$, while
$\operatorname{var}(C_\alpha)=\chi_v(t)$. Hence
$\operatorname{var}(C_c)\subseteq\operatorname{var}(M')$, and
therefore $\operatorname{var}(M)=\operatorname{var}(M')$.

Finally, suppose that $c\in A$ and $\alpha$ does not satisfy $c$.
Let $a_c$ be the leading literal of $C_c$. By Lemma~\ref{lem:introduce-clause-local},
$\operatorname{var}(a_c)\in\chi_v(t)$. Since $\alpha$ does not
satisfy $c$, the literal $a_c$ is falsified by $\alpha$. By the
definition of $C_\alpha$, this implies that $a_c$ belongs to
$C_\alpha$, whereas $\neg a_c$ belongs to $C_c$. Thus $M$ is
derived from the axiom $a_c\mlor\neg a_c$ by the weakening rule.

We now verify the three induction properties.
By the definition of $\mathcal E(t,\alpha,A)$,
$\operatorname{var}(M)\subseteq\mathrm{Var}_t$.
If $M$ is derived from the axiom $a_c\mlor\neg a_c$,
then $\operatorname{var}(a_c)\in\chi_v(t)\subseteq\mathrm{Var}_t$,
so every multiclause in the derivation contains only variables
from $\mathrm{Var}_t$. The derivation contains no resolution,
and hence all three induction properties hold.

In the other two cases, $M$ is derived from
$M'=C_\alpha\mlor
\bigmlor\limits_{d\in\chi_c(t')}C_d$
by the weakening rule.
By the first induction property, every multiclause in the
derivation of $M'$ contains only variables from
$\mathrm{Var}_{t'}=\mathrm{Var}_t$.
Since $\operatorname{var}(M)\subseteq\mathrm{Var}_t$,
the first induction property also holds for $M$.

Every directed path in the proof DAG ending at $M$ consists
of a directed path ending at $M'$ and the weakening step
from $M'$ to $M$. The second induction property therefore
holds for $M$.
As shown above, $\operatorname{var}(M)=\operatorname{var}(M')$.
By the third induction property for $M'$, no variable in
$\operatorname{var}(M)$ is resolved in a variable-eliminating
resolution on any directed path in the proof DAG ending at
$M'$. Since the final step is weakening, the third induction
property also holds for $M$.

\smallskip
\noindent
\emph{Forget-variable node.}
Suppose that $t$ forgets a variable $x$ and has child $t'$. For
$b\in\{0,1\}$, let $\alpha_b:=\alpha\cup\{x\mapsto b\}$. By
Lemma~\ref{lem:forget-var-zero}, both $(t',\alpha_0,A)$ and
$(t',\alpha_1,A)$ are inconsistent.

Fix $M\in\mathcal E(t,\alpha,A)$, and suppose that
\[
M=C_\alpha\mlor
\bigmlor_{c\in\chi_c(t)}C_c.
\]
Since $t$ forgets only $x$, we have
$\chi_c(t)=\chi_c(t')$ and
$\mathrm{Var}_t=\mathrm{Var}_{t'}$. Hence
$c_t^{\mathrm{int}}=c_{t'}^{\mathrm{int}}$ for every
$c\in\chi_c(t)$.

For $b\in\{0,1\}$, let
$M_b:=C_{\alpha_b}\mlor
\bigmlor\limits_{c\in\chi_c(t)}C_c$.
Then $M_b\in\mathcal E(t',\alpha_b,A)$. Moreover,
$C_{\alpha_0}=C_\alpha\mlor x$ and
$C_{\alpha_1}=C_\alpha\mlor\neg x$, so
$M_0=M\mlor x$ and $M_1=M\mlor\neg x$.
Applying the resolution rule to $M_0$ and $M_1$ on $x$, followed
by contraction, derives $M$.

We verify the three induction properties.
Since $\mathrm{Var}_{t'}=\mathrm{Var}_t$ and
$\operatorname{var}(M)\subseteq\mathrm{Var}_t$, the first
induction property for $M_0$ and $M_1$ gives the first
induction property for $M$.

Both $M_0$ and $M_1$ contain the variable $x$.
By their third induction property, $x$ has not been resolved
in a variable-eliminating resolution on any directed path
in the proof DAG ending at either multiclause.
Since the step deriving $M$ resolves only $x$, their second
induction property therefore gives the second induction
property for $M$.

Every variable in $\operatorname{var}(M)$ belongs to both
$\operatorname{var}(M_0)$ and $\operatorname{var}(M_1)$.
Their third induction property thus applies to all these
variables. The final resolution on $x$ is a
variable-eliminating resolution if and only if
$x\notin\operatorname{var}(M)$.
Hence the third induction property also holds for $M$.

\smallskip
\noindent
\emph{Forget-clause node.}
Suppose that $t$ forgets a clause $c$ and has child $t'$. By
Lemma~\ref{lem:forget-clause-zero},
$(t',\alpha,A\cup\{c\})$ is inconsistent.

Fix $M\in\mathcal E(t,\alpha,A)$, and suppose that
\[
M=C_\alpha\mlor
\bigmlor_{d\in\chi_c(t)}C_d.
\]
Since $t$ forgets only $c$, we have
$\chi_v(t')=\chi_v(t)$,
$\chi_c(t')=\chi_c(t)\cup\{c\}$, and
$\mathrm{Var}_{t'}=\mathrm{Var}_t$. Hence
$d_{t'}^{\mathrm{int}}=d_t^{\mathrm{int}}$ for every
$d\in\chi_c(t)$. Moreover,
Lemma~\ref{lem:forget-clause-internal} gives
$c_{t'}^{\mathrm{int}}=c$.

Recall that $c=\ell_1^c\lor\cdots\lor\ell_{q_c}^c$.
For each $i\in\{1,\ldots,q_c\}$, let
\[
P_i:=M\mlor\neg\ell_i^c
\mlor\bigmlor_{j=i+1}^{q_c}\ell_j^c.
\]
By the definition of $\mathcal C_1$,
$\neg\ell_i^c\mlor
\bigmlor\limits_{j=i+1}^{q_c}\ell_j^c
\in\mathcal C_1(c_{t'}^{\mathrm{int}})$.
For every $d\in\chi_c(t)$, we also have
$C_d\in\mathcal C_{\mathbf 1_A(d)}(d_{t'}^{\mathrm{int}})$.
Therefore $P_i\in\mathcal E(t',\alpha,A\cup\{c\})$, so each
$P_i$ has already been derived.

Let $R_0:=c$, and, for $1\leq i\leq q_c$, let
$R_i:=M\mlor\bigmlor\limits_{j=i+1}^{q_c}\ell_j^c$.
Starting with the initial clause $R_0$, we derive $R_i$ from
$R_{i-1}$ and $P_i$ by applying the resolution rule on
$\ell_i^c$, followed by contraction, for
$i=1,\ldots,q_c$. Since $R_{q_c}=M$, this derives $M$.

We verify the three induction properties.
By Lemma~\ref{lem:forget-clause-internal},
$\operatorname{var}(c)\subseteq\mathrm{Var}_{t'}=\mathrm{Var}_t$.
The first induction property for each $P_i$ applies to its
derivation. Since
$\operatorname{var}(R_i)\subseteq
\operatorname{var}(M)\cup\operatorname{var}(c)
\subseteq\mathrm{Var}_t$,
the first induction property also holds for the derivation of $M$.

We prove the second and third induction properties for $R_i$
by induction on $i=0,\ldots,q_c$.
Both hold for the initial clause $R_0=c$.
Fix $i\geq1$, and let $x:=\operatorname{var}(\ell_i^c)$.
Since the literals of $c$ have distinct variables, the
resolution using $R_{i-1}$ and $P_i$ is a
variable-eliminating resolution if and only if
$x\notin\operatorname{var}(M)$, equivalently,
$x\notin\operatorname{var}(R_i)$.

The variable $x$ belongs to both
$\operatorname{var}(R_{i-1})$ and $\operatorname{var}(P_i)$.
By their third induction property, $x$ has not been resolved
in a variable-eliminating resolution on any directed path
in the proof DAG ending at either multiclause.
Thus their second induction property also holds for $R_i$,
since the step deriving $R_i$ resolves only $x$.

Every variable in $\operatorname{var}(R_i)$ belongs to
$\operatorname{var}(P_i)$ and, when $i\geq2$, to
$\operatorname{var}(R_{i-1})$.
Their third induction property therefore applies to these
variables. When $i=1$, the derivation of $R_0$ contains no
resolution because $R_0$ is an initial clause.
The final resolution is variable-eliminating only when
$x\notin\operatorname{var}(R_i)$, so the third induction
property holds for $R_i$.
Since $R_{q_c}=M$, all three induction properties hold for $M$. 

\smallskip
\noindent

Since $F$ is unsatisfiable and $\chi(r)=\emptyset$,
the state $(r,\emptyset,\emptyset)$ is inconsistent.
Moreover,
$\mathcal E(r,\emptyset,\emptyset)=\{\bot\}$.
Hence the construction gives an $\mathrm{mRes}$
refutation of $F$.

We now bound the length of the refutation.
For each inconsistent state $(t,\alpha,A)$, the definition
of $\mathcal E(t,\alpha,A)$ gives
\[
\begin{aligned}
|\mathcal E(t,\alpha,A)|
&\leq
\prod_{c\in A}|c_t^{\mathrm{int}}|
\prod_{c\in\chi_c(t)\setminus A}
\bigl(|c_t^{\mathrm{int}}|+1\bigr)\\
&\leq
\prod_{c\in\chi_c(t)}(|c|+1)
\leq
2^{2\sum_{c\in\chi_c(t)}\max\{\log|c|,1\}}
\leq 2^{2(w_{\log}+1)}.
\end{aligned}
\]
Here the third inequality uses
$|c|+1\leq2^{2\max\{\log|c|,1\}}$.
The number of states at $t$ is
$2^{|\chi_v(t)|+|\chi_c(t)|}\leq2^{w_{\log}+1}$.
Thus at most $2^{3(w_{\log}+1)}$ multiclauses need to be
derived at each node.

For each multiclause $M$, the construction uses at most
one rule application at an introduce-variable,
introduce-clause, forget-variable, or join node, and at most
one additional axiom.
At a forget-clause node forgetting $c$, it uses the initial
clause $c$ and $|c|$ applications of the resolution rule.
Since $c\in\chi_c(t')$ for the child $t'$ of that node,
$\log|c|\leq w_{\log}+1$, and hence
$|c|+1=2^{O(w_{\log})}$.

The multiclauses required from the children have already
been derived and can be used in each local derivation.
Consequently, the total number of additional multiclauses
at each node is $2^{O(w_{\log})}$.
Summing over all nodes gives
$|\pi|\leq |V(T)|2^{O(w_{\log})}$.
\end{proof}

\RegularBounds*

\begin{proof}
Let $w_{\log}:=\mathrm{tw}_{\log}^*(F)$.
By Lemma~\ref{lem:nice-empty-log-incidence}, $G^*(F)$ has
a nice tree decomposition $(T,\chi,r)$ of log-weighted width
$w_{\log}$ with $O((n+m)(w_{\log}+1))$ nodes.
Lemmas~\ref{lem:td-to-mres} and~\ref{lem:mres-to-res}
give a regular resolution refutation of $F$ of length
$|V(T)|2^{O(w_{\log})}=(n+m)2^{O(w_{\log})}$.

If $F$ has maximum clause width $k\geq2$, then by
Lemma~\ref{lem:width-comparison}, $F$ has a regular resolution
refutation of length $(n+m)k^{O(\mathrm{tw}^*(F))}$.
\end{proof}

\section{Further discussion}
Next, we show some corollaries of our main results.

\subsection{FPT-sized resolution refutations under restrictions on clause width} \label{subsec:clause-width}

Using our upper bound, we can prove that, if
$f(n)=n^{o(1)}$ and satisfies certain computability conditions,
the class of unsatisfiable CNF formulas with $n$ variables
and maximum clause width at most $f(n)$ has FPT-sized resolution
refutations parameterized by incidence
treewidth. The computability conditions on $f$ arise from the
computability requirement in the definition of FPT-sized refutations.

\begin{theorem}\label{thm:subpolynomial-clause-width}
Let $f:\mathbb Z_{\geq1}\to[1,\infty)$ satisfy
$f(n)=n^{o(1)}$.
Then the following statements hold.
\begin{enumerate}
    \item There exists a function
    $g:\mathbb Z_{\geq0}\to\mathbb Z_{\geq1}$,
    depending only on $f$, such that every unsatisfiable
    CNF formula $F$ with $n$ variables, $m$ clauses, and
    maximum clause width $k\leq f(n)$ has a resolution
    refutation of length at most
    $g(\mathrm{tw}^*(F))(n+m)^2$.

    \item If there exists a computable function
    $N:\mathbb Z_{\geq1}\to\mathbb Z_{\geq2}$ such that
    $f(n)\leq n^{1/q}$ for all integers $q\geq1$ and
    $n\geq N(q)$, then the function $g$ in (1) can be
    chosen computable. Consequently, the class of
    unsatisfiable CNF formulas with $n$ variables and
    maximum clause width at most $f(n)$ has FPT-sized resolution refutations parameterized by incidence
    treewidth.
\end{enumerate}
\end{theorem}

\begin{proof}
By Theorem~\ref{thm:main-bounds}(3), there exists a
positive integer $a$ such that every unsatisfiable CNF
formula $F$ with $n$ variables, $m$ clauses, and maximum
clause width $k$ has a resolution refutation of length
at most $a(n+m)k^{a\cdot\mathrm{tw}^*(F)}$.

Since $f(n)=n^{o(1)}$, there exists a function
$N:\mathbb Z_{\geq1}\to\mathbb Z_{\geq2}$ such that
$f(n)\leq n^{1/q}$ whenever $q\geq1$ and $n\geq N(q)$.
Define $g(0):=a$ and
$g(q):=aN(aq)^{aq}$ for every integer $q\geq1$.

Let $F$ be an unsatisfiable CNF formula with $n$ variables,
$m$ clauses, and maximum clause width $k\leq f(n)$.
Since $F$ is unsatisfiable and has no empty clauses,
$G^*(F)$ contains an edge, so $\mathrm{tw}^*(F)\geq1$.

If $n\geq N(a\mathrm{tw}^*(F))$, then
$f(n)\leq n^{1/(a\mathrm{tw}^*(F))}$ by the choice of $N$,
and hence
$k^{a\mathrm{tw}^*(F)}
\leq f(n)^{a\mathrm{tw}^*(F)}\leq n$.
If $n<N(a\mathrm{tw}^*(F))$, each clause contains at most
$n$ literals, so $k\leq n<N(a\mathrm{tw}^*(F))$ and
$k^{a\mathrm{tw}^*(F)}
\leq N(a\mathrm{tw}^*(F))^{a\mathrm{tw}^*(F)}$.
Since $g(\mathrm{tw}^*(F))\geq a$ and $n\geq1$,
both cases imply
$a k^{a\mathrm{tw}^*(F)}\leq g(\mathrm{tw}^*(F))n$.
Thus $F$ has a resolution refutation of length at most
\[
a(n+m)k^{a\mathrm{tw}^*(F)}
\leq g(\mathrm{tw}^*(F))n(n+m)
\leq g(\mathrm{tw}^*(F))(n+m)^2,
\]
proving (1).

Under the additional assumption in (2), choose $N$
to be computable.
The function $g$ defined above is then computable,
and the polynomial degree in the length bound is
independent of incidence treewidth.
This proves (2).
\end{proof}

\subsection{Implications for lower bounds} \label{subsec:lower-bounds}

To our knowledge, no lower bound on resolution refutation
length exceeding $n^{O(1)}2^{O(\mathrm{tw}^*(F))}$ is known. Although our results only partially improve previous upper
bounds, they also help rule out certain formula families
as candidates for proving lower bounds stronger than
those currently known by particular methods.

For example, for every unsatisfiable CNF formula $F$
with $n$ variables, $m$ clauses, and maximum clause
width $k$, Theorem~\ref{thm:main-bounds}(3) constructs
a resolution refutation of width at most
$\mathrm{tw}^*(F)+k$.
Hence a direct application of the famous length--width relation
of Ben-Sasson and
Wigderson~\cite{BenSassonWigderson2001}
can yield a length lower bound of at most
$2^{O(\mathrm{tw}^*(F)^2/n)}$.
Since $\mathrm{tw}^*(F)\leq n$, this approach cannot
obtain a better lower bound than
$2^{O(\mathrm{tw}^*(F))}$. Therefore, improving the lower bounds beyond
$2^{O(\mathrm{tw}^*(F))}$ requires more than a direct
application of the length--width relation.

Furthermore, Theorems~\ref{thm:main-bounds}
and~\ref{thm:subpolynomial-clause-width} can help us identify desirable
properties of candidate formula families for proving
lower bounds. 
In particular, formula families with log-weighted or partially
log-weighted incidence treewidth bounded by
$O(\mathrm{tw}^*(F)+\log n)$ can be ruled out as candidates for
proving lower bounds beyond
$(n+m)^{O(1)}2^{O(\mathrm{tw}^*(F))}$.
For example, the bound in Theorem~\ref{thm:main-bounds}(2)
suggests that candidate formula families should satisfy
$\mathrm{tw}_{\log}^*(F)=\omega(\mathrm{tw}^*(F)+\log n)$.
Such a separation requires every tree decomposition
$(T,\chi)$ of $G^*(F)$ of width
$O(\mathrm{tw}^*(F))$ to have a node $t\in V(T)$ such that
\[
\sum_{C\in\chi_c(t)}\max\{\log|C|,1\}
=\omega(\mathrm{tw}^*(F)+\log n).
\]
Since $|\chi(t)|=O(\mathrm{tw}^*(F))$, 
$\chi(t)$ must contain sufficiently many clauses of sufficiently large width.

\subsection{FPT-length refutations with restricted extension variables}
\label{sec:extension-variables}

Atserias and Bonet~\cite{AtseriasBonet2004} established
a correspondence between
$k\mathrm{\text{-}DNF\text{-}Res}$ and resolution with extension variables representing conjunctions of at most $k$ literals.
By this correspondence and
Theorem~\ref{thm:main-bounds}(1), we obtain
resolution refutations of FPT length parameterized
by incidence treewidth after introducing extension
variables.
In fact, we only need to introduce variables representing
nonempty subclauses of clauses in the
original formula, as shown in the following theorem.

The refutation in
Theorem~\ref{thm:td-to-dnfres-tight} is constructed
directly from an incidence tree decomposition.
The proof of Theorem~\ref{thm:ext-to-res} follows that
construction and therefore also uses the structure of
the incidence tree decomposition.
In contrast, as discussed in
Section~\ref{sec:related-work}, the approach used in \cite{SamerSzeider2010CSP,Furer2012} first transforms the
original formula into an equisatisfiable CNF formula
of small primal treewidth and then constructs a
refutation using a primal tree decomposition of the
new formula.
Our construction does not require this intermediate
transformation and uses the incidence tree decomposition
of the original formula directly.

Let $F$ be a CNF formula.
For each nonempty subclause $C$ of a clause in $F$,
introduce a distinct new variable $y_C$ not in $\mathrm{var}(F)$.
Define
\[
F':=
F
\cup
\{y_C\lor\neg\ell:
  \emptyset\neq C\subseteq c,\ c\in F,\ \ell\in C\}
\cup
\{\neg y_C\lor C:
  \emptyset\neq C\subseteq c,\ c\in F\}.
\]
For each nonempty subclause $C$ of a clause in $F$,
the clauses $\neg y_C\lor C$ and
$y_C\lor\neg\ell$ for all $\ell\in C$ together
express $y_C\leftrightarrow C$.
Hence $F$ and $F'$ are equisatisfiable, and $y_C$
represents $C$.

The proof of the following theorem is given in
Appendix~\ref{app:ext-to-res}.

\begin{restatable}{theorem}{ExtToRes}\label{thm:ext-to-res}
Let $F$ be an unsatisfiable CNF formula of maximum
clause width $k$, with $n$ variables and $m$ clauses,
and let
\[
F'=
F
\cup
\{y_C\lor\neg\ell:
  \emptyset\neq C\subseteq c,\ c\in F,\ \ell\in C\}
\cup
\{\neg y_C\lor C:
  \emptyset\neq C\subseteq c,\ c\in F\}.
\]
Then $F'$ has a resolution refutation of length
$k(n+m)2^{O(\mathrm{tw}^*(F))}$.
\end{restatable}

\section{Conclusion and future directions}

We introduced two weighted variants of incidence treewidth
and established the bounds in
Theorems~\ref{thm:main-bounds} and~\ref{thm:regular-main-bounds}.
These results improve previous upper bounds, but it remains
open whether every unsatisfiable CNF formula has an FPT-sized
resolution refutation parameterized by incidence treewidth.
It seems that fully answering this question is still difficult.

To answer this question positively, one would need to construct
FPT-sized resolution refutations parameterized by incidence
treewidth for all unsatisfiable CNF formulas.
Our results improve previous upper bounds and identify further
classes of formulas with short refutations, but extending
these results to all unsatisfiable CNF formulas appears
to require new techniques.

To answer this question negatively, one would need to prove
lower bounds that rule out an FPT bound on resolution
refutation length parameterized by incidence treewidth.
As discussed in Subsection~\ref{subsec:lower-bounds},
our constructions exclude some simple formula families
as candidates for such lower bounds.
Establishing a negative answer may therefore require
formulas with more complex structure.

Our results also suggest several further questions.

One direction is to study how to handle long clauses in the construction of refutations.
Cal\`i and Razgon~\cite{CaliRazgon2022} use a technique
in their preprint to handle additional long clauses.
It would be interesting to investigate whether this technique
can be adapted to our constructions.

Another question is whether every unsatisfiable CNF formula
has an FPT-sized regular resolution refutation parameterized
by partially log-weighted incidence treewidth.
Theorem~\ref{thm:regular-main-bounds} proves that such
refutations exist when parameterized by log-weighted
incidence treewidth, and we believe that an analogous result
may hold for partially log-weighted incidence treewidth.

Finally, Amir~\cite{Amir2010} gave approximation algorithms
for weighted treewidth, which can be used to approximate our log-weighted
incidence treewidth. It would be interesting to investigate whether there are
efficient algorithms for computing or approximating
partially log-weighted incidence treewidth.

\bibliography{references}

\appendix

\section{Proofs of Lemmas~\ref{lem:nice-empty-log-incidence}
and~\ref{lem:plog-decomposition-size}}
\label{app:weighted-decompositions}

\NiceLogDecomposition*

\begin{proof}
Let $(T_0,\chi_0)$ be a tree decomposition of $G^*(F)$
of log-weighted width
$w_{\log}:=\mathrm{tw}_{\log}^*(F)$.
By~\cite[Section~10.4]{KleinbergTardos2006},
there exists a tree decomposition $(T_1,\chi_1)$ of $G^*(F)$
with at most $n+m$ nodes such that every bag of
$(T_1,\chi_1)$ is a bag of $(T_0,\chi_0)$.
Hence $w_{\log}(T_1,\chi_1)\leq w_{\log}$.
Let $w:=w(T_1,\chi_1)$. Then $w\leq w_{\log}$.

We transform $(T_1,\chi_1)$ into a nice tree decomposition
$(T,\chi,r)$ by the following standard operations.
Choose an arbitrary node $r_1\in V(T_1)$ as the root of $T_1$.
Replace each node $t$ with $d\geq2$ children by a rooted
binary tree with $d$ leaves, all of whose bags equal
$\chi_1(t)$, and attach one original child to each leaf.
For each pair of adjacent bags $X,Y$, with $X$ above $Y$,
replace their connecting edge by a path whose bags are
obtained from $X$ by deleting the vertices of $X\setminus Y$
one at a time and then adding the vertices of $Y\setminus X$
one at a time.
Add a path from a new empty root bag to the current root bag
by adding its vertices one at a time.
Similarly, extend each nonempty leaf bag to an empty leaf bag
by deleting its vertices one at a time.
Finally, contract any node with its unique child
when their bags are equal.

The resulting decomposition $(T,\chi,r)$ is a nice tree decomposition.
Every bag in $(T,\chi,r)$ is contained in a bag of $(T_1,\chi_1)$,
so $w_{\log}(T,\chi)\leq w_{\log}$.
Then $w_{\log}(T,\chi)=w_{\log}$ by the definition of
$\mathrm{tw}_{\log}^*(F)$.
The tree obtained by replacing each node with $d\geq2$
children by a binary tree with $d$ leaves and $2d-1$ nodes
has $O(|V(T_1)|)$ nodes, since the sum of the numbers
of children over all nodes of $T_1$ is $|V(T_1)|-1$.
Each path used to connect adjacent bags or to obtain an empty
root or leaf bag has $O(w+1)$ nodes.
So we have $|V(T)|=O(|V(T_1)|(w+1))
=O((n+m)(w_{\log}+1))$.
\end{proof}

\PlogDecompositionSize*

\begin{proof}
Choose a clause-path family $\mathcal P$ such that
$w_{\mathrm{plog}}(T,\chi,r,\mathcal P)=w$. For each vertex $v$
of $G^*(F)$, choose a node $s_v$ whose bag contains $v$. For each
edge $\{x,c\}$ of $G^*(F)$, choose a node $s_{x,c}$ whose bag
contains both $x$ and $c$.

Let $M$ be the set consisting of the root $r$, all nodes
$s_v$ and $s_{x,c}$, and the endpoints of every nonempty
path $P_c$. Then $|M|=O(k(n+m))$.

Let $S$ be the minimal subtree of $T$ containing all nodes in $M$.
For each $u$ which is a leaf of $S$, choose a leaf $t_u$ of $T$
that is a descendant of $u$, and let $Q_u$ be the unique path
in $T$ from $u$ to $t_u$. Let $\widehat S$ be the subtree of $T$
induced by
\[
V(S)\cup\bigcup_{u\in\mathcal L(S)}V(Q_u).
\]
That is, $\widehat S$ is obtained from $S$ by extending each leaf
of $S$ to a leaf of $T$.

Next, we show that
$\bigl(\widehat S,\left.\chi\right|_{V(\widehat S)}\bigr)$
is a tree decomposition of $G^*(F)$. Indeed, every vertex $v$
belongs to the bag $\chi(s_v)$, and every edge $\{x,c\}$
is contained in the bag $\chi(s_{x,c})$. All these nodes
belong to $M\subseteq V(\widehat S)$. Moreover, for each
vertex $v$ of $G^*(F)$, the set
$\{t\in V(\widehat S):v\in\chi(t)\}$ is the intersection
of two subtrees of $T$ and is therefore connected.

In $\widehat S$, a join node of $T$ may have only one child, in
which case the node and its child have the same bag. Let $T'$ be
obtained from $\widehat S$ by repeatedly contracting an edge
$uv$ such that $v$ is the unique child of $u$ and
$\chi(u)=\chi(v)$. For each $t\in V(T')$, let $X_t$ be the set of
nodes of $\widehat S$ contracted into $t$. All nodes in $X_t$
have the same bag, so we may define
$\chi'(t):=\chi(u)$ for any $u\in X_t$. Let $r'$ be the root of
$T'$. It is easy to check that the contractions preserve the
tree decomposition properties, and $(T',\chi',r')$ is a nice
tree decomposition.

If $P_c\neq\emptyset$, then its endpoints belong to $M$, and hence
$P_c\subseteq\widehat S$. Let $P'_c$ be the path obtained from
$P_c$ by the contractions used to construct $T'$. Then $P'_c$
is a path from the root of $T'(c)$ to a leaf of $T'(c)$. If
$P_c=\emptyset$, set $P'_c=\emptyset$.
$\mathcal P'=(P'_c)_{c\in F}$ is a clause-path family.

Since $P'_c$ is obtained from $P_c$ by contracting some nodes that
have the same bag as their unique child,
\[
\bigcup_{u\in P'_c}\chi'_v(u)
=
\bigcup_{u\in P_c}\chi_v(u).
\]
Hence $\rho_{\mathcal P'}(c)=\rho_{\mathcal P}(c)$ for every
$c\in F$.

Consider an edge $uv$ contracted in the construction of $T'$,
where $v$ is the unique child of $u$ and
$\chi(u)=\chi(v)$. If $v\in P_c$, then $u\in P_c$, since $u$
belongs to $T(c)$ and lies between the root of $T(c)$ and $v$.
Conversely, suppose that $u\in P_c$. Since $v\in T(c)$, the node
$u$ is not a leaf of $T(c)$. If $v\notin P_c$, then $P_c$ passes
through another child of $u$. This child belongs to
$\widehat S$ because $P_c\subseteq\widehat S$, contradicting the
fact that $v$ is the unique child of $u$ in $\widehat S$.
Therefore
\[
u\in P_c\quad\Longleftrightarrow\quad v\in P_c.
\]
For every $t\in V(T')$, all nodes of $X_t$ are contracted
into $t$ in the construction of $T'$.
Since $u\in P_c$ if and only if $v\in P_c$ for every
contracted edge $uv$, either all nodes of $X_t$
belong to $P_c$ or none do.
Since $P'_c$ is the image of $P_c$ under these contractions
and $c\in F$ was arbitrary, it follows that, for every
$t\in V(T')$, $u\in X_t$, and $c\in F$,
\[
t\in P'_c\quad\Longleftrightarrow\quad u\in P_c.
\]
Since $\chi'(t)=\chi(u)$,
\[
|\chi'_v(t)|
+\sum_{c\in\chi'_c(t)}\omega_{\mathcal P'}(t,c)
=
|\chi_v(u)|
+\sum_{c\in\chi_c(u)}\omega_{\mathcal P}(u,c)
\leq w+1.
\]
Thus
$w_{\mathrm{plog}}(T',\chi',r',\mathcal P')\leq w$.

It remains to bound $|V(T')|$. Since every vertex in
a bag has weight at least one, every bag of $T$ contains at most
$w+1$ vertices.

Let $Z$ consist of the nodes in $M$, the leaves of $\widehat S$,
and the nodes of $\widehat S$ with two children. Every leaf of
$S$ belongs to $M$, so $\widehat S$ has $O(|M|)$ leaves.
Since $\widehat S$ is a binary tree and has $O(|M|)$ leaves,
$\widehat S$ has $O(|M|)$ nodes with two children.
Hence $|Z|=O(|M|)$.
Let $\widehat S-Z$ denote the graph obtained from $\widehat S$ by
deleting the nodes in $Z$ and all edges incident with them. Every
connected component of $\widehat S-Z$ is a path, and
$\widehat S-Z$ has $O(|M|)$ connected components.

For every connected component $R$ of $\widehat S-Z$, there are
unique nodes $u_R,v_R\in Z$ such that $R\cup\{u_R,v_R\}$
is the path in $\widehat S$ from $u_R$ to $v_R$.

Let $a\in\chi(t)$ for some $t\in V(R)$. The node $s_a$
belongs to $M\subseteq Z$, and its bag contains $a$.
Since $R\cap M=\emptyset$, the path in $\widehat S$ from $t$
to $s_a$ is not contained in $R$, so it contains either
$u_R$ or $v_R$. Since the bags containing $a$ form a connected
subtree of $T$, we have $a\in\chi(u_R)\cup\chi(v_R)$.
So we have
\[
\bigcup_{t\in V(R)}\chi(t)
\subseteq
\chi(u_R)\cup\chi(v_R).
\]

For each
$a\in\bigcup_{t\in V(R)}\chi(t)$, the nodes of $R$ whose bags
contain $a$ form a subpath of $R$. Hence $a$ is introduced at
most once and forgotten at most once on $R$. No node of $R$ has
two children in $\widehat S$. Moreover, if two consecutive nodes
of $R$ have the same bag, then the edge between them is contracted
in the construction of $T'$. Consequently, any two consecutive
nodes of $T'$ corresponding to nodes of $R$ differ by the
introduction or forgetting of one vertex. Since
\[
\left|\bigcup_{t\in V(R)}\chi(t)\right|
\leq |\chi(u_R)\cup\chi(v_R)|
\leq 2(w+1),
\]
we have
\[
\left|\left\{t\in V(T'):X_t\cap V(R)\neq\emptyset\right\}\right|
\leq 4(w+1)+1.
\]

Let $\mathcal R$ be the set of connected components of
$\widehat S-Z$. Then
\[
\begin{aligned}
|V(T')|
&\leq |Z|+
\sum_{R\in\mathcal R}
\left|\left\{t\in V(T'):X_t\cap V(R)\neq\emptyset\right\}\right|\\
&\leq |Z|+(|Z|-1)\bigl(4(w+1)+1\bigr)\\
&=O\left(
|M|(w)
\right)\\
&=O\left(
k(n+m)w
\right).
\end{aligned}
\]
\end{proof}

\section{Proofs of Lemmas~\ref{lem:mres-to-res}
and~\ref{lem-mdnfres-to-dnfres}}
\label{app:multiset-proofs}

\MResToRes*

\begin{proof}
To simplify the presentation of the proof, we add the
repetition rule
\[
\frac{C}{C}
\]
to the resolution proof system in this proof. It is easy to see that any refutation using this
rule can be converted into a resolution refutation without
this rule, of no greater length or width. Moreover, if the
original refutation is regular, the resulting refutation
is also regular. Thus it suffices to construct the refutation in the
resolution proof system with the repetition rule.

Let $\pi'$ be the sequence obtained from $\pi$ by replacing
each multiclause $E$ with $\operatorname{red}(E)$ and deleting
all tautological clauses. We show that $\pi'$ is a resolution
refutation of $F$ by specifying how each clause is derived.

If $E$ is an initial multiclause in $F$ in $\pi$, then
$\operatorname{red}(E)=E$ is an initial clause of $F$.
If $E=x\mlor\neg x$ is an axiom, then
$\operatorname{red}(E)$ is tautological and is therefore
deleted when constructing $\pi'$.

Suppose that $E$ is obtained from an earlier multiclause
$E'$ by zero or more applications of the contraction rule.
Then $\operatorname{red}(E)=\operatorname{red}(E')$.
If $\operatorname{red}(E)$ is non-tautological, it follows
from $\operatorname{red}(E')$ by the repetition rule.

We now consider a multiclause $E$ obtained by applying the
weakening rule or the resolution rule and then applying
the contraction rule, if needed.
We need only consider the case that
$\operatorname{red}(E)$ is non-tautological.

Suppose first that $E$ is obtained using the weakening rule. Then there are an earlier multiclause $E'$ in $\pi$ and a
multiclause $D$ such that $E'\mlor D$ is derived from $E'$
by the weakening rule, and $E$ is obtained from $E'\mlor D$
by applying the contraction rule, if needed.

If a multiclause $C'$ is obtained from a multiclause $C$ by
applying the contraction rule, then
$\operatorname{red}(C')=\operatorname{red}(C)$.
Hence $\operatorname{red}(E)=\operatorname{red}(E'\mlor D)
=\operatorname{red}(E')\lor\operatorname{red}(D)$.
Thus $\operatorname{red}(E')$ is non-tautological, and $\operatorname{red}(E)$
can be derived from $\operatorname{red}(E')$ by the weakening rule.

Suppose next that $E$ is obtained using the resolution rule.
Then there are earlier multiclauses $P=C\mlor x$ and
$Q=D\mlor\neg x$ in $\pi$ such that $R:=C\mlor D$ is derived
from $P$ and $Q$ by the resolution rule, and $E$ is obtained
from $R$ by applying the contraction rule, if needed.

We have
$\operatorname{red}(E)=\operatorname{red}(R)
=\operatorname{red}(C)\lor\operatorname{red}(D)$.
We next show that $\operatorname{red}(E)=\operatorname{red}(R)$
can be derived from earlier clauses of $\pi'$ by at most one
application of the weakening rule or the resolution rule.

If neither $x$ nor $\neg x$ belongs to $\operatorname{red}(R)$,
then neither $C$ nor $D$ contains $x$ or $\neg x$.
Consequently, the resolution step using $P$ and $Q$ is a variable-eliminating resolution.
Moreover,
$\operatorname{red}(P)=\operatorname{red}(C)\lor x$ and
$\operatorname{red}(Q)=\operatorname{red}(D)\lor\neg x$
are non-tautological.
Hence $\operatorname{red}(P)$ and $\operatorname{red}(Q)$
have been derived in $\pi'$. Applying the resolution rule
to $\operatorname{red}(P)$ and $\operatorname{red}(Q)$ on $x$ derives
$\operatorname{red}(C)\lor\operatorname{red}(D)
=\operatorname{red}(R)$.

If $x\in\operatorname{red}(R)$, then
$\operatorname{red}(R)=\operatorname{red}(R)\lor x
=\operatorname{red}(C)\lor x\lor\operatorname{red}(D)
=\operatorname{red}(P)\lor\operatorname{red}(D)$.
Since $\operatorname{red}(R)$ is non-tautological,
$\operatorname{red}(P)$ is also non-tautological and has been derived in $\pi'$. Thus $\operatorname{red}(R)$
is derived from $\operatorname{red}(P)$ by the weakening rule.

Similarly, if $\neg x\in\operatorname{red}(R)$, then
$\operatorname{red}(R)$ is derived from $\operatorname{red}(Q)$
by the weakening rule.

We have shown that every clause of $\pi'$ is either an initial
clause of $F$ or is derived from earlier clauses by the weakening rule, the resolution rule, or the
repetition rule.
Since $\pi$ ends with the empty multiclause, $\pi'$ ends with
the empty clause.
Therefore $\pi'$ is a resolution refutation of $F$.

By the definition of $\pi'$, we have $|\pi'|\leq|\pi|$.
For every multiclause $E$ in $\pi$, the clause
$\operatorname{red}(E)$ contains each literal of $E$ exactly
once. Thus
$\operatorname{width}(\operatorname{red}(E))\leq\operatorname{width}(E)$,
and consequently
$\operatorname{width}(\pi')\leq\operatorname{width}(\pi)$.

A clause $\operatorname{red}(E)$ in $\pi'$ is derived by
resolution on $x$ if and only if $E$ is derived by a
variable-eliminating resolution on $x$ in $\pi$.
The proof DAG of $\pi'$ can be obtained from the proof DAG of $\pi$ by
deleting the vertices for tautological multiclauses and, for
each resolution step replaced by weakening, deleting the
incoming edge from the multiclause not used by the weakening
rule. Each remaining vertex for a multiclause $E$ represents
$\operatorname{red}(E)$ in $\pi'$.
Hence every directed path in the proof DAG of $\pi'$ is also
a directed path in the proof DAG of $\pi$.

If a directed path in the
proof DAG of $\pi'$ contained two resolutions on the same
variable, the same path in the proof DAG of $\pi$ would contain
two variable-eliminating resolutions on that variable. So 
if on every directed path in the proof DAG of $\pi$, each
variable is resolved in at most one variable-eliminating resolution
step, then the resulting refutation is regular.
\end{proof}

\MDNFResToDNFRes*

\begin{proof}
We construct a $k\mathrm{\text{-}DNF\text{-}Res}$ derivation
$\pi'$ from $F$ by processing the mDNFs of $\pi$ in order. For every mDNF $E$ in $\pi$, we derive
$\operatorname{red}(E)$.

If $E$ is an initial clause of $F$, then
$\operatorname{red}(E)=E$ is also an initial clause of $F$.
If $E=x\mlor\neg x$ is an axiom, then
$\operatorname{red}(E)=x\lor\neg x$ is an axiom.
For the axiom $\top$, choose a variable
$x\in\operatorname{var}(F)$.
Applying the $\land$-elimination rule to the axiom
$x\lor\neg x$ gives $\top\lor\neg x$.
A second application of the $\land$-elimination rule
gives $\top$.

For every other mDNF $E$, the definition of a derivation
provides an mDNF $\widetilde E$ from which $E$ is obtained
by zero or more applications of the contraction rule.
Since the contraction rule does not change the set of
terms, we have
$\operatorname{red}(E)=\operatorname{red}(\widetilde E)$.
There are two cases.

If $\widetilde E$ is an earlier mDNF in $\pi$, then
$\operatorname{red}(\widetilde E)$ has already been derived,
so no additional DNF is needed.

Otherwise, $\widetilde E$ is derived from earlier mDNFs
by one application of a rule other than contraction.
We consider each such rule and show how to derive
$\operatorname{red}(E)=\operatorname{red}(\widetilde E)$
from the reductions of those earlier mDNFs.

Suppose first that $\widetilde E=D_0\mlor D_1$ is obtained
from $D_0$ by the weakening rule.
Then
\[
\operatorname{red}(E)
=\operatorname{red}(D_0)\lor\operatorname{red}(D_1).
\]
Hence $\operatorname{red}(E)$ follows from
$\operatorname{red}(D_0)$ by the weakening rule.

Suppose next that $\widetilde E=D_0\mlor T'$ is obtained
from $D_0\mlor T$ by the $\land$-elimination rule, where
$T'\subseteq T$.
Let $S:=\operatorname{red}(D_0)$.
Then $\operatorname{red}(E)=S\lor T'$.

If $m_{D_0}(T)=0$, then
$\operatorname{red}(D_0\mlor T)=S\lor T$, with $T\notin S$.
Applying the $\land$-elimination rule to $S\lor T$ gives
$S\lor T'=\operatorname{red}(E)$.
If $m_{D_0}(T)>0$, then
$\operatorname{red}(D_0\mlor T)=S$.
In this case, $\operatorname{red}(E)=S\lor T'$ follows
from $S$ by the weakening rule.

Suppose next that
$\widetilde E=D_1\mlor D_2\mlor(T_1\land T_2)$ is obtained
from $D_1\mlor T_1$ and $D_2\mlor T_2$ by the
$\land$-introduction rule.
Let $S_i:=\operatorname{red}(D_i)$ for $i\in\{1,2\}$,
and let $T:=T_1\land T_2$.
Then
\[
\operatorname{red}(E)=S_1\lor S_2\lor T.
\]

If $m_{D_1}(T_1)=m_{D_2}(T_2)=0$, then the two earlier
mDNFs reduce to $S_1\lor T_1$ and $S_2\lor T_2$, with
$T_i\notin S_i$ for $i\in\{1,2\}$.
Applying the $\land$-introduction rule to these two
DNFs gives $S_1\lor S_2\lor T=\operatorname{red}(E)$.
Otherwise, $m_{D_i}(T_i)>0$ for some $i\in\{1,2\}$,
so $\operatorname{red}(D_i\mlor T_i)=S_i$.
Since $\operatorname{red}(E)=S_1\lor S_2\lor T$,
we can derive $\operatorname{red}(E)$ from $S_i$
by the weakening rule.

Finally, suppose that $\widetilde E=D_1\mlor D_2$ is
obtained by the cut rule from
\[
D_1\mlor\ell_1\mlor\cdots\mlor\ell_q
\qquad\text{and}\qquad
D_2\mlor(\neg\ell_1\land\cdots\land\neg\ell_q).
\]
Let $a_1,\ldots,a_r$ be the distinct literals among
$\ell_1,\ldots,\ell_q$, let
$N:=\neg a_1\land\cdots\land\neg a_r$, and let
$S_i:=\operatorname{red}(D_i)$ for $i\in\{1,2\}$.
The two earlier mDNFs reduce to
\[
P:=S_1\lor a_1\lor\cdots\lor a_r
\qquad\text{and}\qquad
Q:=S_2\lor N.
\]
Moreover, $\operatorname{red}(E)=S_1\lor S_2$.

Let $P^{-}$ be obtained from $P$ by deleting the
singleton terms $a_1,\ldots,a_r$, and let $Q^{-}$ be
obtained from $Q$ by deleting the term $N$.
Then
$P=P^{-}\lor a_1\lor\cdots\lor a_r$ and
$Q=Q^{-}\lor N$.
Applying the cut rule to $P$ and $Q$ gives
$U:=P^{-}\lor Q^{-}$.
Since $P^{-}\subseteq S_1$ and $Q^{-}\subseteq S_2$,
every term of $U$ belongs to
$\operatorname{red}(E)=S_1\lor S_2$.
Thus $\operatorname{red}(E)$ follows from $U$ by
at most one application of the weakening rule.

Every term appearing in $\pi'$ either appears in $\pi$
or is a singleton term or the empty term.
Thus every term in $\pi'$ contains at most $k$ literals.
Since $\pi$ ends with $\bot$ and
$\operatorname{red}(\bot)=\bot$, the derivation $\pi'$
is a $k\mathrm{\text{-}DNF\text{-}Res}$ refutation of $F$.
For each mDNF in $\pi$, the construction adds at most
three DNFs to $\pi'$.
Hence $|\pi'|\leq 3|\pi|$.
\end{proof}

\section{Proofs of the lemmas in
Subsection~\ref{subsec:inconsistent-states}} \label{app:inconsistent-states}

\JoinVarStructure*

\begin{proof}
Since $t$ is a join node,
\[
\chi_v(t)=\chi_v(t_1)=\chi_v(t_2)\subseteq\mathrm{Var}_{t_1}\cap\mathrm{Var}_{t_2},
\qquad
\chi_c(t)=\chi_c(t_1)=\chi_c(t_2)\subseteq\mathrm{Cls}_{t_1}\cap\mathrm{Cls}_{t_2}.
\]
Moreover, $T_t$ consists of $t$, $T_{t_1}$, and $T_{t_2}$. Therefore,
\[
\begin{aligned}
\mathrm{Var}_t
&=\chi_v(t)\cup\mathrm{Var}_{t_1}\cup\mathrm{Var}_{t_2}
 =\mathrm{Var}_{t_1}\cup\mathrm{Var}_{t_2},\\
\mathrm{Cls}_t
&=\chi_c(t)\cup\mathrm{Cls}_{t_1}\cup\mathrm{Cls}_{t_2}
 =\mathrm{Cls}_{t_1}\cup\mathrm{Cls}_{t_2}.
\end{aligned}
\]

It remains to prove the intersection identity. The inclusion $\chi_v(t)\subseteq\mathrm{Var}_{t_1}\cap\mathrm{Var}_{t_2}$ follows from $\chi_v(t)=\chi_v(t_1)=\chi_v(t_2)$. Conversely, if $x\in\mathrm{Var}_{t_1}\cap\mathrm{Var}_{t_2}$, then $T_{t_1}$ and $T_{t_2}$ both contain a node whose bag contains $x$. The unique path between these two nodes passes through $t$. Since the nodes whose bags contain $x$ induce a connected subtree of $T$, every node on this path has a bag containing $x$. Hence $x\in\chi_v(t)$. Together with $\chi_v(t)\subseteq\mathrm{Var}_{t_1}\cap\mathrm{Var}_{t_2}$, we have $\mathrm{Var}_{t_1}\cap\mathrm{Var}_{t_2}=\chi_v(t)$.

Finally, let $c\in\mathrm{Cls}_t\setminus\chi_c(t)$ with $c\in\mathrm{Cls}_{t_i}$. Since the nodes whose bags contain $c$ induce a connected subtree and $c\notin\chi(t)$, all such nodes belong to $T_{t_i}$. For every $x\in\mathrm{var}(c)$, some bag contains both $x$ and $c$, and hence $x\in\mathrm{Var}_{t_i}$. Thus we have $\mathrm{var}(c)\subseteq\mathrm{Var}_{t_i}$. Consequently, any literal of $c$ satisfied by $\tau$ is also satisfied by $\tau|_{\mathrm{Var}_{t_i}}$, so the latter satisfies $c$.
\end{proof}

\IntroduceNoX*

\begin{proof}
Suppose that some $c\in\mathrm{Cls}_{t'}\setminus\chi_c(t')$ contains a literal on $x$. Since $\{x,c\}$ is an edge of $G^*(F)$, there is a node $s\in V(T)$ such that $\{x,c\}\subseteq\chi(s)$.

Since $c\in\mathrm{Cls}_{t'}$, some node in $T_{t'}$ has a bag containing $c$. If $s\notin V(T_{t'})$, the path from $s$ to this node passes through $t'$. As the nodes whose bags contain $c$ induce a connected subtree, we would have $c\in\chi_c(t')$, a contradiction. Hence $s\in V(T_{t'})$.

Now both $\chi(s)$ and $\chi(t)$ contain $x$, and the path from $s$ to $t$ passes through $t'$. Since the nodes whose bags contain $x$ induce a connected subtree, we have $x\in\chi_v(t')$, contradicting that $t$ introduces $x$.
\end{proof}

\IntroduceClauseLocal*

\begin{proof}
The inclusion $\mathrm{var}(c)\cap\chi_v(t)\subseteq\mathrm{var}(c)\cap\mathrm{Var}_t$ is obvious. For the reverse inclusion, let $x\in\mathrm{var}(c)\cap\mathrm{Var}_t$. Since $c\in\chi_c(t)\setminus\chi_c(t')$, no node in $T_{t'}$ has a bag containing $c$; otherwise, the path from such a node to $t$ would pass through $t'$, forcing $c\in\chi_c(t')$. Suppose that $x\notin\chi_v(t)$. Since $x\in\mathrm{Var}_t$, some node in $T_{t'}$ has a bag containing $x$. On the other hand, the edge $\{x,c\}$ is covered by a bag outside $T_{t'}$. The path between these two nodes passes through $t$, and the nodes whose bags contain $x$ induce a connected subtree. Hence we have $x\in\chi_v(t)$, a contradiction. Therefore, $\mathrm{var}(c)\cap\mathrm{Var}_t=\mathrm{var}(c)\cap\chi_v(t)$.

If $\tau:\mathrm{Var}_t\to\{0,1\}$ satisfies $c$, then some satisfied literal of $c$ has its variable in $\mathrm{Var}_t$, hence by the equality above its variable lies in $\chi_v(t)$. Therefore $\tau|_{\chi_v(t)}$ also satisfies $c$.
\end{proof}

\JoinZeroRecursion*

\begin{proof}
Since $t$ is a join node, $\chi_v(t)=\chi_v(t_1)=\chi_v(t_2)$ and $\chi_c(t)=\chi_c(t_1)=\chi_c(t_2)$.

$(\Rightarrow)$ Assume $(t,\alpha,A)$ is inconsistent. Suppose there exist $A_1,A_2\subseteq \chi_c(t)$ with $A_1\cup A_2=A$ such that both $(t_1,\alpha,A_1)$ and $(t_2,\alpha,A_2)$ are consistent. Then there are assignments $\tau_i:\mathrm{Var}_{t_i}\to\{0,1\}$ extending $\alpha$ such that every clause in $A_i$ and every clause in $\mathrm{Cls}_{t_i}\setminus \chi_c(t)$ is satisfied by $\tau_i$ for $i=1,2$.

By Lemma~\ref{lem:join-var-structure}(1), $\mathrm{Var}_{t_1}\cap \mathrm{Var}_{t_2}=\chi_v(t)$ and $\mathrm{Var}_{t_1}\cup \mathrm{Var}_{t_2}=\mathrm{Var}_t$. Since $\tau_1$ and $\tau_2$ agree on $\chi_v(t)$, they combine into $\tau:\mathrm{Var}_t\to\{0,1\}$ extending $\alpha$.

We claim that $\tau\in N(t,\alpha,A)$. First, if $c\in A$, then $c\in A_1$ or $c\in A_2$, hence $c$ is satisfied by $\tau$. Next, let $c\in \mathrm{Cls}_t\setminus \chi_c(t)$. Then by Lemma~\ref{lem:join-var-structure}(2),  $c\in \mathrm{Cls}_{t_i}\setminus \chi_c(t)=\mathrm{Cls}_{t_i}\setminus \chi_c(t_i)$ for some $i\in\{1,2\}$, so it is satisfied by $\tau_i$. Thus $c$ is also satisfied by $\tau$. So $\tau\in N(t,\alpha,A)$, contradicting the assumption.

$(\Leftarrow)$ Assume that for all $A_1,A_2\subseteq \chi_c(t)$ with $A_1\cup A_2=A$, at least one of $(t_1,\alpha,A_1)$ or $(t_2,\alpha,A_2)$ is inconsistent. Suppose that $(t,\alpha,A)$ is consistent, and $\tau \in N(t,\alpha,A)$.

Let $\tau_i:=\tau|_{\mathrm{Var}_{t_i}}$ and define $A_i:=\{c\in A:\tau_i \text{ satisfies } c\}$. Then $A_1\cup A_2=A$, since any clause $c\in A$ satisfied by $\tau$ has a literal $\ell$ satisfied by $\tau$, and $\mathrm{var}(\ell) \in \mathrm{Var}_{t}=\mathrm{Var}_{t_1} \cup \mathrm{Var}_{t_2}$ by Lemma~\ref{lem:join-var-structure}(1). 

We claim that $(t_i,\alpha,A_i)$ is consistent. Indeed, $\tau_i$ extends $\alpha$. If $c\in A_i$, then $c$ is satisfied by $\tau_i$ by definition. If $c\in \mathrm{Cls}_{t_i}\setminus \chi_{c}(t_i)=\mathrm{Cls}_{t_i}\setminus \chi_{c}(t)$, then $c\in \mathrm{Cls}_t\setminus \chi_c(t)$, so $\tau$ satisfies $c$. It follows that $\tau_i$ satisfies $c$ by Lemma~\ref{lem:join-var-structure}(2). Thus $\tau_i\in N(t_i,\alpha,A_i)$.

Hence both $(t_1,\alpha,A_1)$ and $(t_2,\alpha,A_2)$ are consistent, contradicting the assumption. Therefore $(t,\alpha,A)$ is inconsistent.
\end{proof}

\IntroduceVarZero*

\begin{proof}
Let $\alpha':=\alpha|_{\chi_v(t')}$.
Since $t$ introduces a variable, we have $\chi_v(t)=\chi_v(t') \cup \{x \}$, $\mathrm{Var}_{t}=\mathrm{Var}_{t'} \cup \{ x\}$, $\mathrm{Cls}_{t'}=\mathrm{Cls}_t$ and $\chi_c(t')=\chi_c(t)$. Therefore, we have $\mathrm{Cls}_{t'}\setminus \chi_c(t')=\mathrm{Cls}_t\setminus \chi_c(t)$.

$(\Leftarrow)$
Suppose $(t,\alpha,A)$ is consistent, and let $\tau\in N(t,\alpha,A)$. Set $\tau':=\tau|_{\mathrm{Var}_{t'}}$. Then $\tau'|_{\chi_v(t')}=\tau|_{\chi_v(t')}=\alpha|_{\chi_v(t')}=\alpha'$. 
If $c\in A_x$, then $c$ is not satisfied by $\tau(x)=\alpha(x)$. Since $\tau$ satisfies $c$, some other literal of $c$ is satisfied by $\tau$, and hence by $\tau'$. 
Moreover, by Lemma~\ref{lem:introduce-no-x}, every clause in $\mathrm{Cls}_{t'} \setminus \chi_c(t')$ contains no literal on $x$, and hence is satisfied by $\tau'$ since it is satisfied by $\tau$.
Thus $\tau'\in N(t',\alpha',A_x)$, contradicting the assumption.

$(\Rightarrow)$
Suppose $(t',\alpha',A_x)$ is consistent, and let $\tau'\in N(t',\alpha',A_x)$. Extend $\tau'$ to $\tau$ by setting $\tau(x):=\alpha(x)$. For $c \in A$, if $c\in A_x$, then $c$ is satisfied by $\tau'$; otherwise the literal of $c$ on $x$ is satisfied by $\tau$. Moreover, every clause in $\mathrm{Cls}_t\setminus \chi_c(t)=\mathrm{Cls}_{t'}\setminus \chi_c(t')$ is satisfied by $\tau'$, and hence by $\tau$.  
Thus $\tau\in N(t,\alpha,A)$, contradicting the assumption.
\end{proof}

\IntroduceClauseZero*

\begin{proof}
Since $t$ introduces $c$, we have $\chi_v(t)=\chi_v(t')$, $\mathrm{Var}_t=\mathrm{Var}_{t'}$, $\chi_c(t)=\chi_c(t')\cup\{c\}$, and $\mathrm{Cls}_t=\mathrm{Cls}_{t'}\cup\{c\}$. In particular, $\mathrm{Cls}_t\setminus\chi_c(t)=\mathrm{Cls}_{t'}\setminus\chi_c(t')$.

Suppose first that $c\notin A$. Then $A\subseteq\chi_c(t')$. For every assignment $\tau$ on $\mathrm{Var}_t=\mathrm{Var}_{t'}$, we have $\tau\in N(t,\alpha,A)$ if and only if $\tau$ extends $\alpha$, satisfies every clause in $A$, and satisfies every clause in $\mathrm{Cls}_{t'}\setminus\chi_c(t')$. These are precisely the conditions for $\tau\in N(t',\alpha,A)$. Hence $N(t,\alpha,A)=N(t',\alpha,A)$, and (a) follows.

Now suppose that $c\in A$ and $\alpha$ satisfies $c$. Every assignment extending $\alpha$ satisfies $c$. Therefore, for every assignment $\tau$ on $\mathrm{Var}_t=\mathrm{Var}_{t'}$, we have $\tau\in N(t,\alpha,A)$ if and only if $\tau\in N(t',\alpha,A\setminus\{c\})$. Hence $N(t,\alpha,A)=N(t',\alpha,A\setminus\{c\})$, and (b) follows.

Finally, suppose that $c\in A$ and $\alpha$ does not satisfy $c$. If $\tau\in N(t,\alpha,A)$, then $\tau$ satisfies $c$. By Lemma~\ref{lem:introduce-clause-local}, $\tau|_{\chi_v(t)}$ also satisfies $c$. But $\tau$ extends $\alpha$, so $\tau|_{\chi_v(t)}=\alpha$, contradicting the assumption on $\alpha$. Thus $N(t,\alpha,A)=\emptyset$, which gives (c).

This proves the lemma.
\end{proof}

\ForgetVarZero*

\begin{proof}
Since $t$ forgets a variable, we have $\chi_v(t) \cup \{x \}=\chi_v(t')$ and $\mathrm{Var}_t=\mathrm{Var}_{t'}$, $\chi_c(t)=\chi_c(t')$ and $\mathrm{Cls}_t=\mathrm{Cls}_{t'}$, hence $\mathrm{Cls}_t\setminus \chi_c(t)=\mathrm{Cls}_{t'}\setminus \chi_c(t')$.

We show
\[
N(t,\alpha,A)\neq\emptyset \Longleftrightarrow N(t',\alpha_0,A)\neq\emptyset \text{ or } N(t',\alpha_1,A)\neq\emptyset.
\]

$(\Rightarrow)$
Let $\tau\in N(t,\alpha,A)$ and set $b:=\tau(x)$. Then $\tau$ is defined on $\mathrm{Var}_t=\mathrm{Var}_{t'}$, extends $\alpha_b$ and satisfies every clause in $A$. Moreover, every clause in $\mathrm{Cls}_{t'}\setminus \chi_c(t')=\mathrm{Cls}_t\setminus \chi_c(t)$ is satisfied by $\tau$. Thus $\tau\in N(t',\alpha_b,A)$.

$(\Leftarrow)$
Let $\tau\in N(t',\alpha_b,A)$ for some $b\in\{0,1\}$. Since $\mathrm{Var}_{t'}=\mathrm{Var}_t$, we view $\tau$ as an assignment on $\mathrm{Var}_t$. As $\alpha_b$ extends $\alpha$, $\tau$ extends $\alpha$ and satisfies every clause in $A$. Moreover, every clause in $\mathrm{Cls}_t\setminus \chi_c(t)=\mathrm{Cls}_{t'}\setminus \chi_c(t')$ is satisfied by $\tau$. Thus $\tau\in N(t,\alpha,A)$.
\end{proof}

\ForgetClauseZero*

\begin{proof}
Since $t$ forgets a clause $c$, we have $\chi_v(t)=\chi_v(t')$, $\mathrm{Var}_t=\mathrm{Var}_{t'}$, $\chi_c(t')=\chi_c(t)\cup\{c\}$, and $\mathrm{Cls}_{t'}=\mathrm{Cls}_t$. Hence $\mathrm{Cls}_t\setminus \chi_c(t)=\bigl(\mathrm{Cls}_{t'}\setminus \chi_c(t')\bigr)\cup\{c\}$.

We need to show that
\[
N(t,\alpha,A)\neq\emptyset
\Longleftrightarrow
N(t',\alpha,A\cup\{c\})\neq\emptyset.
\]

$(\Rightarrow)$
Let $\tau\in N(t,\alpha,A)$. Then $\tau$ extends $\alpha$, satisfies every clause in $A$, and satisfies every clause in $\mathrm{Cls}_t\setminus \chi_c(t)$. Since $\mathrm{Cls}_t\setminus \chi_c(t)=\bigl(\mathrm{Cls}_{t'}\setminus \chi_c(t')\bigr)\cup\{c\}$, it follows that $\tau$ satisfies every clause in $\mathrm{Cls}_{t'}\setminus \chi_c(t')$ and satisfies $c$. Thus $\tau\in N(t',\alpha,A\cup\{c\})$.

$(\Leftarrow)$
Let $\tau\in N(t',\alpha,A\cup\{c\})$. Then $\tau$ extends $\alpha$, satisfies every clause in $A\cup\{c\}$, and satisfies every clause in $\mathrm{Cls}_{t'}\setminus \chi_c(t')$. Since $\mathrm{Cls}_t\setminus \chi_c(t)=\bigl(\mathrm{Cls}_{t'}\setminus \chi_c(t')\bigr)\cup\{c\}$, it follows that $\tau$ satisfies every clause in $A$ and every clause in $\mathrm{Cls}_t\setminus \chi_c(t)$. Hence $\tau\in N(t,\alpha,A)$.
\end{proof}

\section{Proof of Theorem~\ref{thm:ext-to-res}}
\label{app:ext-to-res}

\ExtToRes*

\begin{proof}
By Lemma~\ref{lem:nice-empty-incidence} and
Theorem~\ref{thm:td-to-dnfres-tight}, we obtain a
$k\mathrm{\text{-}mDNF\text{-}Res}$ refutation $\pi$
of $F$ of length $(n+m)2^{O(\mathrm{tw}^*(F))}$.
Every term in $\pi$ is either a singleton term or
of the form $\neg C$ for some possibly empty
subclause $C\subseteq c$ of a clause $c\in F$.

For every nonempty term $T$ in $\pi$, let
$p_T:=\ell$ if $T=\{\ell\}$, and let
$p_T:=\neg y_C$ if $T=\neg C$ and $|T|\geq2$.
For an mDNF $D=T_1\mlor\cdots\mlor T_r$ containing
no empty term, let
\[
\mathcal C(D):=p_{T_1}\mlor\cdots\mlor p_{T_r}.
\]
Thus $m_{\mathcal C(D)}(p_T)=m_D(T)$ for every
term $T$ in $D$, and $\mathcal C(\bot)=\bot$.
We construct an $\mathrm{mRes}$ derivation $\pi'$
from $F'$ by processing the mDNFs of $\pi$ in order
and deriving $\mathcal C(D)$ whenever $D$ contains
no empty term.
We skip every mDNF containing the empty term $\top$.

For every nonempty term $T$ in $\pi$ and every
$\ell\in T$, let
\[
P_{T,\ell}:=\neg p_T\mlor\ell,
\qquad
Q_T:=p_T\mlor\bigmlor_{\ell\in T}\neg\ell.
\]
If $T=\neg C$ and $|T|\geq2$, then
$p_T=\neg y_C$ and $\neg\ell\in C$.
Hence $P_{T,\ell}=y_C\mlor\ell$ and
$Q_T=\neg y_C\mlor\bigmlor\limits_{\ell\in C}\ell$
are initial multiclauses of $F'$.
If $T=\{\ell\}$, both $P_{T,\ell}$ and $Q_T$
are the axiom $\ell\mlor\neg\ell$.

If $D=c$ is an initial clause in $F$, then
$\mathcal C(D)=c$ is an initial multiclause of $F'$.
If $D=x\mlor\neg x$ is an axiom, then
$\mathcal C(D)=x\mlor\neg x$ is also an axiom.
We skip the axiom $\top$.

Suppose that $D$ contains no empty term and is
obtained from an earlier mDNF $E$ by zero or more
applications of the contraction rule.
The contraction rule does not change the set of
terms, so $E$ also contains no empty term.
By applying the contraction rule to $p_T$ for each application of the contraction rule to
a term $T$,
we can derive $\mathcal C(D)$ from
$\mathcal C(E)$ in at most one step.

We now consider an mDNF $D$ obtained by one
application of a rule other than the contraction
rule from earlier mDNFs, followed by zero or more applications of
the contraction rule.
Let $E$ be the mDNF obtained before applying the
contraction rule.
We need only consider the case that $D$ contains
no empty term.
Since the contraction rule does not change the
set of terms, $E$ also contains no empty term.
We first show how to derive $\mathcal C(E)$.

Suppose that the weakening rule derives
$E=D_0\mlor D_1$ from $D_0$.
Neither $D_0$ nor $D_1$ contains an empty term.
By applying the weakening rule to $\mathcal C(D_0)$,
we derive
$\mathcal C(E)=\mathcal C(D_0)\mlor\mathcal C(D_1)$
in one step.

Suppose that the $\land$-elimination rule derives
$E=D_0\mlor S$ from $D_0\mlor T$, where $S\subseteq T$.
Since $E$ contains no empty term, $S$ is nonempty,
and hence $T$ is also nonempty.
Let $S=\{\ell_1,\ldots,\ell_r\}$.
Starting from $Q_S$, we apply the resolution rule
with $P_{T,\ell_1},\ldots,P_{T,\ell_r}$ in order,
on $\ell_1,\ldots,\ell_r$, respectively, and obtain
$p_S\mlor\bigmlor\limits_{j=1}^{r}\neg p_T$.
By applying the contraction rule to this multiclause,
we derive $\neg p_T\mlor p_S$.
By applying the resolution rule to
$\neg p_T\mlor p_S$ and $\mathcal C(D_0)\mlor p_T$,
we derive $\mathcal C(D_0)\mlor p_S=\mathcal C(E)$.
The derivation has length $O(k)$.

Suppose that the $\land$-introduction rule derives
$E=D_1\mlor D_2\mlor T$ from
$D_1\mlor T_1$ and $D_2\mlor T_2$, where
$T=T_1\cup T_2$.
Neither $D_1$ nor $D_2$ contains an empty term,
and $T$ is nonempty.
If $T_1=\top$, then $T=T_2$.
By applying the weakening rule to
$\mathcal C(D_2)\mlor p_{T_2}$, we derive
$\mathcal C(D_1)\mlor\mathcal C(D_2)\mlor p_{T_2}
=\mathcal C(E)$ in one step.
If $T_2=\top$, then $T=T_1$.
By applying the weakening rule to
$\mathcal C(D_1)\mlor p_{T_1}$, we derive
$\mathcal C(D_1)\mlor\mathcal C(D_2)\mlor p_{T_1}
=\mathcal C(E)$ in one step.

Now suppose that $T_1$ and $T_2$ are nonempty.
Let $T=\{\ell_1,\ldots,\ell_r\}$.
For each $j$, choose $i(j)\in\{1,2\}$ such that
$\ell_j\in T_{i(j)}$.
Starting from $Q_T$, we apply the resolution rule
with $P_{T_{i(1)},\ell_1},\ldots,P_{T_{i(r)},\ell_r}$
in order, on $\ell_1,\ldots,\ell_r$, respectively,
and obtain
$p_T\mlor\bigmlor\limits_{j=1}^{r}\neg p_{T_{i(j)}}$.
Since every $\neg p_{T_{i(j)}}$ belongs to
$\{\neg p_{T_1},\neg p_{T_2}\}$, by applying the
contraction rule and the weakening rule to this
multiclause, we derive
$H:=\neg p_{T_1}\mlor\neg p_{T_2}\mlor p_T$.
By applying the resolution rule to $H$ and
$\mathcal C(D_1)\mlor p_{T_1}$, we obtain
$\mathcal C(D_1)\mlor\neg p_{T_2}\mlor p_T$.
By applying the resolution rule to this multiclause
and $\mathcal C(D_2)\mlor p_{T_2}$, we derive
$\mathcal C(D_1)\mlor\mathcal C(D_2)\mlor p_T
=\mathcal C(E)$.
The derivation has length $O(k)$.

Finally, suppose that the cut rule derives
$E=D_1\mlor D_2$ from
$D_1\mlor\ell_1\mlor\cdots\mlor\ell_q$ and
$D_2\mlor T$, where
$T=\neg\ell_1\land\cdots\land\neg\ell_q$
and $1\leq q\leq k$.
Since $E$ contains no empty term and $T$ is
nonempty, both mDNFs used in the cut rule contain
no empty term.
For each $j\in\{1,\ldots,q\}$, by applying the
resolution rule to $\mathcal C(D_2)\mlor p_T$ and
$P_{T,\neg\ell_j}$, we derive
$B_j:=\mathcal C(D_2)\mlor\neg\ell_j$.
Let
\[
\begin{aligned}
A_0&:=\mathcal C(D_1)\mlor\ell_1\mlor\cdots\mlor\ell_q,\\
A_j&:=\mathcal C(D_1)\mlor\mathcal C(D_2)
       \mlor\ell_{j+1}\mlor\cdots\mlor\ell_q
       \qquad(1\leq j\leq q).
\end{aligned}
\]
By applying the resolution rule to $A_0$ and $B_1$
on $\ell_1$, we derive $A_1$.
For $2\leq j\leq q$, by applying the resolution rule
to $A_{j-1}$ and $B_j$ on $\ell_j$, we obtain
$A_j\mlor\mathcal C(D_2)$.
Since every literal of $\mathcal C(D_2)$ belongs
to $A_j$, by applying the contraction rule to
$A_j\mlor\mathcal C(D_2)$, we derive $A_j$.
Thus we derive $A_q=\mathcal C(E)$.
The derivation has length $O(k)$.

In each case, we derive $\mathcal C(E)$ in $O(k)$
steps.
Since $D$ is obtained from $E$ by applications of
the contraction rule, we can obtain $\mathcal C(D)$
from $\mathcal C(E)$ by applying the contraction rule to $p_T$ for each application of the contraction rule to
a term $T$.
By the definition of an $\mathrm{mRes}$ derivation,
this requires at most one further step.
Thus we derive $\mathcal C(D)$ in $O(k)$ steps.

We have therefore constructed an $\mathrm{mRes}$
derivation $\pi'$ from $F'$ of length $O(k|\pi|)$.
Since $\pi$ ends with the empty mDNF and
$\mathcal C(\bot)=\bot$, $\pi'$ ends with the
empty multiclause.
Hence $\pi'$ is an $\mathrm{mRes}$ refutation of $F'$.
By Lemma~\ref{lem:mres-to-res}, we obtain a resolution
refutation of $F'$ of length at most $|\pi'|$,
and therefore of length
$k(n+m)2^{O(\mathrm{tw}^*(F))}$.
\end{proof}

\end{document}